\ifdefined\EC

\documentclass[format=acmsmall, review=false]{acmart}
\usepackage{acm-ec-22}
\usepackage{booktabs} 
\usepackage[ruled]{algorithm2e} 

\SetAlFnt{\small}
\SetAlCapFnt{\small}
\SetAlCapNameFnt{\small}
\SetAlCapHSkip{0pt}
\IncMargin{-\parindent}

\setcitestyle{authoryear}

\else

\documentclass[12pt,a4paper]{article}
\usepackage[utf8]{inputenc}
\usepackage{natbib}
\setcitestyle{authoryear,open={(},close={)}}

\usepackage{amsmath}
\usepackage{amsthm}
\usepackage{amssymb}
\usepackage{amsfonts}

\usepackage[ruled]{algorithm2e} 

\fi

\usepackage{comment}
\usepackage{xcolor}
\usepackage{tikz}
\usepackage{graphics}
\usepackage{graphicx}
\usepackage{color}
\usepackage{gensymb} 

\usepackage{geometry}
\usepackage[normalem]{ulem}
\usepackage{float}

\RequirePackage{mathrsfs}
\RequirePackage{amsopn}
\RequirePackage{esvect}

\RequirePackage{graphicx}
\RequirePackage{epstopdf}

\RequirePackage{mathtools}
\RequirePackage{enumitem}
\RequirePackage{bbm}

\usepackage{hyperref}
\hypersetup{colorlinks=true,linkcolor=blue,citecolor=blue}
\usepackage{dsfont} 

\RequirePackage[capitalize,nameinlink]{cleveref}[0.19]

\RequirePackage{csquotes}

\usepackage{nicefrac}

\definecolor{ForestGreen}{rgb}{.13,.54,.13}
\definecolor{BrickRed}{rgb}{.80,.26,.33}

\ifdefined\DRAFT
\newcommand{\fed}[1]{{\color{ForestGreen}{(\textbf{Fedor:} #1)}}}
\newcommand{\ed}[1]{{\color{BrickRed} {#1}}}
\else
\newcommand{\fed}[1]{}
\newcommand{\ed}[1]{{{#1}}}
\fi

\newtheorem{theorem}{Theorem}
\newtheorem{lemma}{Lemma}

\newtheorem{corollary}{Corollary}

\newtheorem{proposition}{Proposition}
\newtheorem*{proposition*}{Proposition}

\newtheorem*{theorem*}{Theorem}
\newtheorem*{lemma*}{Lemma}
\newtheorem{problem}{Problem}

\theoremstyle{definition}

\newtheorem{definition}{Definition}

\newcommand{\R}{\mathbb{R}}

\newcommand{\N}{\mathbb{N}}

\newcommand{\lb}{\left\langle}
\newcommand{\rb}{\right\rangle}
\newcommand{\rev}{\mathrm{Rev}}

\newcommand{{\dd}}{\,\mathrm{d}}

\newcommand{\opt}{\mathrm{opt}}
\newcommand{\mes}{\mathrm{mes}}
\newcommand{\sing}{\mathrm{sing}}
\renewcommand{\div}{\mathrm{div}}
\newcommand{\B}{\mathrm{Beck}}
\newcommand{\lip}{\mathrm{Lip}}
\DeclareMathOperator*{\argmax}{arg\,max}
\DeclareMathOperator*{\argmin}{arg\,min}

\theoremstyle{definition}
\newtheorem{example}{Example}
\newtheorem{assumption}{Assumption}

\theoremstyle{remark}
\newtheorem{remark}{Remark}

\crefname{remark}{Remark}{Remarks}
\crefname{problem}{Problem}{Problems}
\crefname{example}{Example}{Examples}

\usepackage{tikz}

\RequirePackage[tikz]{mdframed}

\mdfdefinestyle{blockquote}{
    linecolor=red,
    linewidth=3pt,
    leftmargin=0,
    rightmargin=0,
    backgroundcolor=yellow!40,
    topline=false,
    bottomline=false
}

\crefformat{equation}{\textup{#2(#1)#3}}
\crefrangeformat{equation}{\textup{#3(#1)#4--#5(#2)#6}}
\crefmultiformat{equation}{\textup{#2(#1)#3}}{ and \textup{#2(#1)#3}}
{, \textup{#2(#1)#3}}{, and \textup{#2(#1)#3}}
\crefrangemultiformat{equation}{\textup{#3(#1)#4--#5(#2)#6}}%
{ and \textup{#3(#1)#4--#5(#2)#6}}{, \textup{#3(#1)#4--#5(#2)#6}}{, and \textup{#3(#1)#4--#5(#2)#6}}
\ifdefined\EC

\title[Beckmann's approach to multi-item multi-bidder auctions]{Beckmann's approach to multi-item multi-bidder auctions}

\author{Submission 8}

\begin{abstract}
\end{abstract}

\else

\title{Beckmann's approach to multi-item multi-bidder auctions\footnote{We are grateful to Federico Echenique, Luciano Pomatto, Joseph Root, and Omer Tamuz for conversations that inspired this work. {We also thank Luciano for multiple suggestions improving the manuscript.} {The paper benefited from discussions with} 
		Kim Border, Benjamin Brooks, {Alkis Georgiadis-Harris,}
		Sergiu Hart, Jason Hartline, {Andreas Kleiner,} {Alexey Kushnir,} Alejandro Manelli, Robert McCann, Benny Moldovanu, and Philip J.~Reny 
		{and from comments by seminar participants at the University of Bonn,  INFORMS Workshop on Market Design 2022, the 33rd Stony Brook Game Theory Conference, and Yale University.} }}
\author{Alexander V. Kolesnikov \ed{(HSE University)}\thanks{{Kolesnikov acknowledges the support of  RSF Grant \textnumero  22-21-00566  https://rscf.ru/en/project/22-21-00566/. The article was prepared within the framework of the HSE University Basic Research Program.}}
\\
Fedor Sandomirskiy \ed{(Caltech)}\thanks{
Sandomirskiy thanks Linde Institute at Caltech and National Science Foundation (grant
CNS 1518941).}
\\ 
Aleh Tsyvinski \ed{(Yale)} 
\\ 
Alexander P. Zimin \ed{(MIT and HSE University) }
}

\date{}

\fi

\begin{document}

\ifdefined\EC
\begin{titlepage}
\maketitle
\end{titlepage}
\else
\setcounter{page}{0}
\maketitle
\thispagestyle{empty}
\begin{abstract}
We consider the problem of revenue-maximizing Bayesian auction design with several 
bidders \ed{having independent private values} over
 several 
items. 
We show that it can be reduced to the problem of continuous optimal transportation introduced by~\cite{beckmann1952continuous} \ed{where the optimal transportation flow generalizes the concept of ironed virtual valuations to the multi-item setting.}
We establish the strong duality between the two problems and the existence of solutions. 
\ed{The results rely on insights from majorization and optimal transportation theories and on the characterization of feasible interim mechanisms by~\cite{hart2015implementation}.} 
\end{abstract}
\fi


\newpage
\section{Introduction}


The current understanding of multi-bidder multi-item revenue-maximizing auctions  is far from being complete even in the basic setting of several bidders competing for several items and having i.i.d.~additive utilities over them. Only the case of one item and several bidders was analyzed completely \citep{myerson1981optimal}. A seemingly innocent problem with one bidder and several items already turns out to be notoriously difficult to analyze, optimal mechanisms are known only in a few particular cases and exhibit complicated structure \citep{armstrong1999multi,rochet2003economics,daskalakis2015multi}. \ed{If both the number of bidders and the number items exceed one, we get a benchmark problem combining the difficulty of mechanism design with multi-dimensional types and that with  multiple agents.  Essentially, nothing has been known about optimal auctions in this setting.}

The main contribution of our paper is to establish an unexpected  connection between the problem of auction design and an optimal transportation problem in the classic model of \cite{beckmann1952continuous}. 
In contrast to the prevalent  Monge-Kantorovich approach, Beckmann's paradigm of  ``continuous transportation'' captures the trajectories along which \ed{transportation occurs.}

\ed{Imagine a commodity that  is produced and consumed at different geographic locations and so the product has to be transported. Theory of optimal transportation aims to find the least costly way of doing that for given spacial distributions of production and consumption and given transportation costs. \cite{beckmann1952continuous} modelled the process of transportation as a continuous flow. Its intensity and direction at each point defines a vector and the problem is to find the vector field with minimal cost. 
 More formally,} let $\pi_p(x)$ and $\pi_c(x)$ be the density of production and consumption at a given \ed{point} 
 $x$ 
 \ed{in some Euclidean space} and $\rho=\rho(x)$ be a weight function. Let $c=c(x)$ be a vector field representing the direction and the intensity of the 
 \ed{flow. The flow is chosen to compensate supply-demand imbalances, i.e., the difference between the weighted inflow and outflow in a region has to be equal to the difference between supply and demand in it.} 
 This compensation boils down to the condition that the divergence   ${\div}[\rho\cdot c]$ \ed{(the sum of partial derivatives)} must be equal to $\pi_p-\pi_c$. Let $\Phi(c(x))$ be the local cost of transportation. For given $\pi_p,\pi_c$, $\rho$, and $\Phi$, Beckmann's problem is to find the flow $c$ 
  \ed{compensating imbalances and having the minimal total cost:}
\begin{equation}\label{eq_Beckmann}
\min_{c:\ {\div} [\rho\cdot c]+\pi_c- \pi_p=0}\int \Phi \left(c(x)\right)\rho(x) {\dd} x.
\end{equation}

We demonstrate that a dual problem to revenue maximization  takes the form of Beckmann's problem~\eqref{eq_Beckmann} with  a particular cost function $\Phi$ and marginals $\pi_p$ and $\pi_c$ satisfying a certain majorization constraint. \ed{Locations $x$ represent bidders' types and the field $c$ corresponds to ironed virtual valuations.} Similar generalizations of  Beckmann's problem are related to mean-field limits of the Wardrop equilibria for congested optimal transportation games 
\citep{santambrogio2015optimal,Carlier} but have not appeared in the context of auction design. 

We establish the strong duality, namely, the optimal revenue of the auctioneer equals the optimal value of the dual problem. 
\ed{The strong duality is especially useful if combined with the existence of solutions --- i.e., if the optima are attained --- as this combination enables complementary slackness conditions.} 
We demonstrate both the existence of an optimal auction and an optimal vector field $c$, which \ed{are one of} the most technically challenging parts of the paper. 

\ed{\paragraph{Applications and simulations.}
We illustrate a use of duality by recovering the result of \cite{jehiel2007mixed} that selling several items with independent values separately is never optimal provided that values are continuously distributed. This amounts to checking that complementary slackness conditions become incompatible whenever the allocation of each item depends on the values for this item only.

Complementary slackness suggests a guess-and-verify approach: first one guesses a solution to the primal problem, uses complementary slackness to find a dual solution, and this dual solution plays a role of a certificate verifying the initial guess. We illustrate this approach in the cases of one item and multiple bidders or one bidder and two items with i.i.d.  uniform values where the optimal mechanism 
was found by \cite{manellivincent}; see Appendix~\ref{app_examples}.

A prerequisite for the guess-and-verify approach is the existence of a simple explicit solution. To get insights into the structure of optimal auctions, we compute them numerically for several bidders having uniformly distributed values over two items. These simulations indicate the complexity of the optimal mechanism even in this benchmark setting. In particular, the optimal auction does not seem to be given by an elementary function.\footnote{\ed{Proposing  candidates for an optimal auction format remains a major open problem in the multi-item multi-bidder setting. Note that selling separately or selling the grand bundle are never optimal  \citep{jehiel2007mixed}.}}} \ed{We bound  the  revenue  loss  from  using  sub-optimal  designs  and  show how the conclusions change as the number of bidders grows.}

 \ed{To compute the optimal auctions, we develop a new numerical approximation scheme that allows one} to conduct simulations that were previously out of reach.\footnote{\ed{Even the advanced neural-network approach of~\citep{dutting2019optimal} does not produce the outcome detailed and reliable enough to make structural conclusions because of the curse of dimensionality.}} The approach relies on a combination of 
multi-to-single-agent reduction
of~\cite{cai2012algorithmic} and~\cite{alaei2019efficient} avoiding the curse of dimensionality at the cost of dealing with a non-linear feasibility constraint, majorization theory insights~\citep{kleiner2021extreme} to \ed{linearize} this constraint, \ed{duality to Beckmann's problem to guarantee that a solution of a discretized problem is close to that of the continuous one,
and cutting-edge numerical methods to handle optimization over convex functions speeding up the algorithm in practice.} As far as we know, algorithms based on multi-to-single-agent reduction have never been 
previously implemented.

\paragraph{Single bidder versus multiple bidders.}

To get more intuition about our approach and to highlight the specific features of the multi-bidder setting, we compare it  to the single-bidder benchmark of the monopolist's problem.

An important advance in understanding the monopolist's problem was made by~\cite{daskalakis2017strong} who 
 showed how to
reduce it to an  optimal transportation one. Instead of Beckmann's problem arising in the multi-bidder setting, the dual derived by~\cite{daskalakis2017strong} is the Monge–Kantorovich optimal transportation problem with a majorization constraint.  The Monge–Kantorovich problem has the following form:
\begin{equation}\label{eq_Kantorovich}
    \min_{\gamma:\ \gamma_1=\pi_p, \ \gamma_2=\pi_c}\int |x-y|{\dd}\gamma(x,y),
\end{equation}
where $\pi_p(x){\dd} x$ is the geographical distribution of production, $\pi_c(y){\dd} y$ is the distribution of consumption, and the goal is to find a transportation plan $\gamma$ such that the total transportation cost given by the integral in~\eqref{eq_Kantorovich}
 is minimal and supply meets demand, i.e., the marginal of $\gamma$ on the first coordinate is $\pi_p$ and on the second,~$\pi_c$. In contrast to Beckmann's problem, the transportation happens momentarily: only initial and final destinations are captured by the plan $\gamma$, not the trajectories connecting them.
 
We conclude that the revenue-maximization problem for a single-bidder has two differently looking optimal-transportation duals: the Monge–Kantorovich dual and  Beckmann's one. This indicates the connection between the two duals themselves, in particular, their values must be equal. It turns out that for a single bidder, the cost function $\Phi$ in Beckmann's problem can be simplified to  $\Phi(x)=|x|$. For this cost function,
the values of~\eqref{eq_Beckmann} and~\eqref{eq_Kantorovich} are known to coincide by the so-called Beckmann's duality  \citep[Section~4.2]{santambrogio2015optimal}. The presence of the two duals for the monopolist's problem is a repercussion of this duality.

For a single bidder, one can use any of the two duals. However, the link between revenue maximization and the Monge–Kantorovich problem turns out to be limited to the single-bidder case. By contrast,  the connection to Beckmann's problem generalizes to any number of bidders and items.

Let us highlight \ed{the key features of the approach allowing us to handle the multi-bidder case.} 
 The standard first step in the analysis of the monopolist's problem is replacing the non-tractable maximization over mechanism via a handy maximization over interim utility functions $u=u(x)$, where $x$ is the buyer's type  
\citep{rochet1998ironing}. The Rochet-Chon\'e representation is the starting point for the analysis of \cite{daskalakis2017strong}. The corresponding optimization problem  is of the form
$$\max_u\int \Big(\langle \nabla u(x), x\rangle -u(x)\Big) \rho(x){\dd} x,$$
where the maximum is taken over non-decreasing convex non-negative $1$-Lipschitz functions~$u$.

The Rochet-Chon\'e representation can be generalized to multi-bidder problems 
at the cost of getting an extra constraint capturing feasibility of the corresponding interim allocation rule. \ed{Most of the literature relies on a form of this constraint conjectured by~\cite{matthews1984implementability} and proved by~\cite{border1991implementation}.

\ed{The key role in our approach is played by a less known form of this feasibility  condition  discovered by~\cite{hart2015implementation} and extended to the multi-item setting in our paper. The condition takes a form of a majorization constraint on the distribution of $u$'s  gradient.}}
 The connection to majorization theory fuels  our analysis and simulations.
This theory has multiple recent applications in economic design; see, e.g., \citep{kleiner2021extreme,arieli2019optimal,candogan2021optimal,nikzad2022constrained, gershkov2021theory}. 

\ed{The non-local majorization constraint determines  the crucial difference between the resulting multi-bidder Rochet-Chon\'e representation and its single-bidder version. It does not allow us to get rid of the derivatives of $u$, which was crucial for the approach of~\cite{daskalakis2017strong}.} This obstacle explains why their approach does not generalize to the multi-bidder setting and why our dual problem  does not look similar to the Monge-Kantorovich one. 
\ed{The non-local constraint is a major complication;} it leads to involved functional classes needed to establish strong duality and the existence of a solution to the dual. 
\ed{We note that demonstrating strong duality with non-local constraints is new not just for the economic literature but also to the broader mathematical context.}

\subsection*{Related literature}
Linear programs and their duals are ubiquitous in microeconomics and economic design \citep{vohra2011mechanism, bichler2017market}. \ed{The modern literature is increasingly interested in infinite-dimensional settings (corresponding to non-atomic type spaces) as they highlight geometric properties of solutions such as differentiability, convexity, and links to majorization.}
 Apart from multi-item auctions discussed below, infinite-dimensional linear programs and their duals naturally arise in various contexts, e.g.,  informationally or distributionally  robust auction  design~\citep{bergemann2016informationally, koccyiugit2020distributionally,suzdaltsev2020optimal},
information economics \citep{kolotilin2018optimal, dworczak2019simple, dizdar2020simple,arieli2021feasible}.  
Infinite-dimensional programs often have the structure similar to the Monge-Kantorovich optimal transportation, for example, in the context of sorting on the labor market~\citep{boerma2021sorting}, matching with transferable utility and principal-agent problems~\citep{chiappori2010hedonic}, econometrics~\citep{galichon2021survey}, optimal taxation~\citep{steinerberger2019tax}, strategic learning and forecasting~\citep{gensbittel2015extensions, arieli2021transport, guo2021costly}.
Other economic applications  of optimal transport can be found in \citep{figallikimmccann,McCannZhang} and are surveyed by~\cite{galichon2016optimal} and \cite{Carlier}. For non-linear economic problems, a dual approach sharing some similarity with optimal transportation duality was proposed by \cite{noldeke2018implementation}.
A comprehensive presentation of the mathematical theory of transportation can be found in the books by~\cite{santambrogio2015optimal} and \cite{villani2009optimal} and in surveys
\cite{BoKo, McCannGuill}. 

The continuous model of transportation developed by~\cite{beckmann1952continuous}
 is one of the classical economic models of transport networks that had considerable early popularity. It has not been used much in the recent economic literature with the exception of spatial equilibrium models of \cite{fajgelbaum2020optimal} and   \cite{allen2014trade}. \ed{Beckmann's problem has anticipated the dynamic perspective on optimal transportation playing an important role in the modern theory; see the discussion in in Appendix~\ref{app_Beckmann}.
		 This perspective is central for cutting-edge machine learning techniques such as the Wasserstein gradient flows \citep{cuturi2019,kolouri2017}.  Beckmann's problem depends on the difference between production and consumption distributions but not on the distributions per se which makes it similar to the transshipment problem, a version of the Monge-Kantorovich problem where the distributions are not fixed but their difference is \citep{rachev2006mass}. Beckmann's problem can be seen as a dynamic version of the transshipment problem  \citep{carlier2005variational}.}

\ed{For infinite-dimensional problems,
 the central questions become whether the duality gap is zero or not (strong versus weak duality) and whether primal and dual solutions exist.} Both strong duality and the existence are needed for complementary slackness conditions to hold.  In auction design, these questions have only been studied in the single-bidder case. \cite{daskalakis2017strong} established the connection to optimal transport, demonstrated the strong duality, and showed the existence;
their proofs were then simplified by \cite{kleiner2019strong}. 
For several bidders, \cite{giannakopoulos2018duality}  partially relaxed the incentive-compatibility constraint and got a weakly dual problem sharing some similarity with the maximal flow one. In contrast to our paper, they did not discuss the issue of existence as non-zero duality gap diminishes the importance of this question. \ed{\cite{cai2019duality} considered a general problem of Bayesian mechanism design with finite number of types and derived a strongly dual problem resembling the maximal flow one. 
As the problem is finite-dimensional, the existence questions become mute and the strong duality is a consequence of the standard linear programming duality.  A similar duality approach was outlined by~\cite{myerson2002incentive} who, however, focused on  bargaining applications and did not discuss auctions. None of these papers relied on multi-to-single-agent reduction; this simplified the derivation of the duals at the cost of getting high-dimensional problems for $B\geq 2$ agents. 
} 

Even for a single bidder, optimal multi-item auctions can be complex and require non-linear pricing of a continuum of fractional bundles; explicit answers are known in a few particular cases such as uniformly or exponentially distributed values~\citep{daskalakis2017strong}. A primal approach of \cite{haghpanah2021pure} based on virtual surplus maximization provides an alternative to optimal transportation technique of \cite{daskalakis2017strong} and, in some cases, pins down an optimal mechanism, e.g., it shows when pure bundling is optimal in the single-bidder case; see also~\cite[Chapter 8]{hartline2013mechanism}.
Instead of looking for optimal mechanisms the literature has mainly focused on either showing that a simple mechanism can guarantee a certain fraction of the optimal revenue or asking how well one can approximate the optimal mechanism withing a certain parametric class; see representative papers~\citep{hart2019better, babaioff2020simple} and~\citep{hart2017approximate,babaioff2021menu}. The only explicitly solved multi-item auction with several bidders 
assumes that bidders' valuations are binary \citep{yao2017dominant}.

\section{Model}\label{sec_model}

We work in the standard setting of Bayesian auction design with quasilinear bidders having i.i.d. additive  utilities over items. 

There is a set $\mathcal{B}=\{1,2,\ldots,B\}$ of $B\geq 1$ bidders and a set $\mathcal{I}=\{1,2,\ldots,I\}$ of $I\geq 1$ items. We assume that the items are divisible and normalize the total amount of each item to one unit. As usual, indivisible items can be made divisible by interpreting fractional amounts as probability shares.

Bidders treat the items as perfect substitutes and, hence,  bidders' preferences are modelled by additive utility functions quasi-linear in money. The utility function of a bidder $b\in \mathcal{B}$ receiving a bundle $p_b\in \R_+^\mathcal{I}$ of items for a price $t_b$ takes the form
$$\lb p_b, x_b   \rb-t_b,$$
where $\lb \cdot, \cdot \rb$ is the standard dot product in $\R^\mathcal{I}$
and the vector $x_b\in \R_+^\mathcal{I}$ specifies $b$'s maximal willingness to pay for each of the items. The vector $x_b$ can be seen as bidder $b$'s type and constitutes the bidder's private information. Each bidder's type $x_b$ belongs to the set of types\footnote{\ed{This assumption is without loss of generality as any bounded set of types can be made a subset of $[0,1]^\mathcal{I}$ by rescaling.}} $X=[0,1]^\mathcal{I}$.

We assume that the fraction of bidders of different types in the population is described by a  density $\rho$ positive on  $X$ and zero beyond. The bidders are chosen from this population independently and, hence, the types $x_b\in X$, $b\in\mathcal{B}$, are i.i.d. draws with the distribution $\mu$ where ${\dd}\mu(x_b)=\rho(x_b){\dd} x_b$.
The auctioneer and bidders know $\rho$ and each bidder observes the realization of her own type. 



A mechanism which we also refer to as auction is given by a collection of bundles  $P=(P_b(x))_{b\in\mathcal{B}}$ and transfers $T=(T_b(x))_{b\in\mathcal{B}}$ for each profile of types $x=(x_b)_{b\in\mathcal{B}}$.
Formally, a mechanism $(P,T)$ is a measurable map  $X^{\mathcal{B}}\to \R_+^{\mathcal{I}\times\mathcal{B}}\times \R^{\mathcal{B}}$:
	$$(x_b)_{b\in\mathcal{B}}\  \to \  \Big(P_b\big((x_b)_{b\in\mathcal{B}}\big),\ T_b\big((x_b)_{b\in\mathcal{B}}\big)\Big)_{b\in\mathcal{B}}.$$
Here $P_b$ is the bundle received by a bidder $b\in\mathcal{B}$ and $T_b$ is the amount of money she pays to the auctioneer. A mechanism is {feasible} if for any profile of types~$(x_b)_{b\in\mathcal{B}}$
\begin{equation}
    \label{eq_feasible_allocation}
   \sum_{b\in\mathcal{B}} P_{b,i}\big((x_b)_{b\in\mathcal{B}}\big)\leq 1\quad\mbox{for all items $i\in \mathcal{I}$,} 
\end{equation}   
i.e., the auctioneer has only one unit of each item to sell and so a mechanism cannot allocate more than one unit.   

The auctioneer aims to design an auction maximizing the expected revenue $\sum_{b\in\mathcal{B}} T_b$. 
Bidders' types are their private information and a bidder may misreport her type if this brings her higher utility. Similarly, participation is voluntary and bidders may decide not to take part in the auction if they do not expect this to be profitable.
Hence, providing incentives for truthful behavior and participation becomes design constraints. To formalize them,  compute the expected allocation and transfer faced by a bidder $b$ of a given type $x_b$ assuming that others report their types truthfully:
\begin{align}
    \overline{P}_b(x_b)&=\int_{X^{{\mathcal{B}}\setminus\{b\}}} P_b\big((x_b)_{b\in\mathcal{B}}\big)\cdot \left(\prod_{d\in {\mathcal{B}}\setminus\{b\}}  \rho(x_d)\right) {\dd} x_1\cdots{\dd} x_{b-1}{\dd} x_{b+1}\cdots {\dd} x_B,\\
    \overline{T}_b(x_b)&=\int_{X^{{\mathcal{B}}\setminus\{b\}}} T_b\big((x_b)_{b\in\mathcal{B}}\big) \cdot \left(\prod_{d\in {\mathcal{B}}\setminus\{b\}}  \rho(x_d)\right) {\dd} x_1\cdots{\dd} x_{b-1}{\dd} x_{b+1}\cdots {{\dd}} x_B.
\end{align}
Such one-bidder marginals $(\overline{P}_b,\overline{T}_b)$ of the original mechanism $(P,T)$ are known as its reduced forms or interim mechanisms. The reduced mechanism for a bidder~$b$ captures how her expected utility depends on her type and her report, i.e., all the information relevant to her: if her type is $x_b$ and she reports to be of type $x_b'$, while other bidders remain truthful, $b$'s expected utility takes the form $$\big\langle\overline{P}_b(x_b'),\,x_b\big\rangle-\overline{T}_b(x_b').$$

A mechanism  is called
 {Bayesian incentive-compatible} if truth-telling is a Bayesian equilibrium, i.e., no bidder $b$ has an incentive to misreport her values if others report truthfully. 
		Formally,
		\begin{equation}\label{eq_SP}
		\big\langle\overline{P}_b(x_b),\,x_b\big\rangle-\overline{T}_b(x_b)\geq \langle\overline{P}_b(x_b'),\,x_b\rangle-\overline{T}_b(x_b')
	\end{equation}
	for all  $x_b,x_b'\in X$ and $b\in\mathcal{B}.$

A mechanism is called {individually rational}  if no bidder wants to abstain from participation, i.e., nobody gets a negative expected utility. Formally, 
		\begin{equation}\label{eq_IR}
	\big\langle\overline{P}_b(x_b),\,x_b\big\rangle-\overline{T}_b(x_b)\geq 0
	\end{equation}
	for all  $x_b\in X$ and $b\in\mathcal{B}.$

The auctioneer's design problem takes the following form.

\smallskip
\noindent{\textbf{Auctioneer's problem:}} \emph{maximize  the expected revenue
\begin{equation}\label{eq_revenue}
\int_{X^{\mathcal{B}}} \left(\sum_{b\in\mathcal{B}} T_b\big((x_b)_{b\in\mathcal{B}}\big)\right)\cdot \left(\prod_{b\in\mathcal{B}}\rho(x_b)\right){\dd} x_1\cdots {\dd} x_B
\end{equation}
over {individually-rational Bayesian incentive-compatible feasible} mechanisms}~$(P,T)$.
\medskip

In the case of a single bidder $(B=1)$, the auctioneer's problem becomes the multi-item monopolist's problem.  
Note that for $B=1$, the reduced mechanism coincides with the original one, i.e., $\overline{P}_1\equiv P_1$ and~$\overline{T}_1\equiv T_1$.
In what follows, we will use the monopolist's problem as a benchmark and, in particular, connect our characterization to the one obtained by \cite{daskalakis2017strong}.

\section{Multi-bidder version of Rochet-Chon\'e representation}\label{sec_Rochet}

A common starting point for the analysis of the monopolist's problem is its equivalent representation derived in \cite{rochet1998ironing}. 
We first recall their insight in the single-bidder setting and then describe  its extension to the general case of $B\geq 1$ bidders. 

\subsection{
\ed{Monopolist's problem}\label{subsect_monopolist}}  \ed{Consider a one-bidder mechanism  $(P,T)$. With each such mechanism,} we can associate the interim utility function $u(x)=\big\langle P(x),x\big \rangle-T(x)$, i.e., the expected utility obtained by a bidder of type $x$. Following \cite{rochet1998ironing},
the monopolist's problem can be 
rewritten as a maximization over the utility function $u$ under some constraints. Bayesian incentive compatibility and individual rationality boil down to $u$ being a convex non-negative function. The allocation probabilities $P(x)$ are given by the gradient $\nabla u(x)$. Hence, $\langle \nabla u(x), x \rangle$ is the utility that the bidder derives from the allocated items. As the total utility is $u(x)$, the difference $\langle \nabla u(x), x \rangle-u(x)$ is the payment that goes to the monopolist. Consequently, the monopolist's problem reduces to maximizing
 \begin{equation}\label{eq_Rochet_monopolist}
\int_{X} \Big(\langle \nabla u(x), x\rangle -u(x)\Big) \rho(x) {\dd} x, 
\end{equation}
over convex  $u:\ X\to\R_+$ such that $\nabla u(x)\in [0,1]^\mathcal{I}$. The last condition originates from the requirement of feasibility: for each item $i$, the allocated amount
\begin{equation}\label{eq_allocation_via_utility}
P_{i}(x)=\frac{\partial u}{\partial{x_{i}}}(x)
\end{equation}
has to be between $0$ and $1$.
\subsection{
\ed{Auctioneer's problem}} 
Consider now the auction-design problem with $B\geq 1$ bidders. We show that this  problem can be reduced to an optimization problem that is similar to the monopolist's problem but the feasibility constraint $\frac{\partial u}{\partial x_{i}}(x)\leq 1$ on the gradient's values is replaced by a non-local majorization condition on the distribution of the gradient.
\begin{definition}[Majorization\footnote{Majorization is \ed{also known under the name of second-order stochastic dominance. Both are not} to be confused with a closely related notion of dominance with respect to the convex order also known as the Blackwell order, which corresponds to taking any convex $\varphi$, not necessarily non-decreasing. For probability measures, convex dominance implies that $\nu$ and $\nu'$ have the same mean, while for majorization,  the majorizing measure can have a higher mean, i.e., $\int t{\dd}\nu(t)\geq \int t{\dd} \nu'(t)$.}]\label{def_majorization}
For a pair of measures $\nu$ and $\nu'$, we say that  $\nu$ majorizes $\nu'$ if $\int \varphi{\dd}\nu\geq \int \varphi{\dd}\nu'$ for any convex non-decreasing function $\varphi$. 
 A random variable $\xi$ majorizes $\xi'$ if the distribution of $\xi$ majorizes that of $\xi'$. We write $\nu\succeq_{} \nu'$ and $\xi\succeq_{} \xi'$.
 \end{definition}
 Informally, majorization means that $\nu$ can be obtained from $\nu'$ by combining mean-preserving spreads with moving mass to higher values.\

As we will see, the auctioneer's problem with $B$ bidders is equivalent to the following one.

\smallskip 
\noindent\textbf{Multi-bidder Rochet-Chon\'e problem:}
\emph{maximize 
\begin{equation}\label{eq_Rochet_Chone_extension}
B\cdot\int_{X} \Big(\langle \nabla u(x), x\rangle -u(x)\Big) \rho(x) {\dd} x
\end{equation}
over convex non-decreasing functions  $u:\ X\to\R_+$  with $u(0)=0$ and such that for all $i\in \mathcal{I}$
\begin{equation}\label{eq_Border}
\frac{\partial u}{\partial x_i} (\chi)\preceq \xi^{B-1},
\end{equation}
where $\chi\in X$ is distributed with the density $\rho$ and $\xi$ is uniformly distributed on~$[0,1]$.}

\smallskip
Let us clarify the meaning of the condition~\eqref{eq_Border}. Each component of the gradient\footnote{\ed{We do not assume that the function $u$ is smooth and, hence, the partial derivative $\frac{\partial u}{\partial x_{i}} (x)$ may not exist for some $x$. Despite this fact, the  optimization problem~\eqref{eq_Rochet_Chone_extension} is well-defined since the gradient of a convex function exists almost everywhere and integration with respect to an absolutely continuous measure is not sensitive to the behavior of the integrand on sets of zero Lebesgue measure; see Appendix~\ref{sect_convex_analysis} for basics of convex analysis.}} $\frac{\partial u}{\partial x_{i}} (\chi)$ is treated there as a random variable by assuming that 
the argument $\chi\in X$ is random and distributed with the density $\rho$ and the distribution of this random variable must be majorized by the distribution of $\xi^{B-1}$, where $\xi$ is uniform on $[0,1]$. An equivalent way to write this condition is to assume that for any non-decreasing convex $\varphi$
\begin{equation}\label{eq_Border_expanded}
\int_X \varphi\left(\frac{\partial u}{\partial x_{i}} (x)\right)\rho(x) {\dd} x \leq \int_0^1 \varphi\left(z^{B-1}\right){\dd} z.
\end{equation}
\begin{proposition}\label{prop_Rochet}
The optimal revenue in the auctioneer's problem~\eqref{eq_revenue} and the value of the multi-bidder Rochet-Chon\'e problem~\eqref{eq_Rochet_Chone_extension} coincide and the optima in both problems are attained.
\end{proposition}
A proof of Proposition~\ref{prop_Rochet} is contained in Appendix~\ref{app_Rochet} and the key ideas are discussed below. 
  \ed{The proposition makes apparent the connection of the auctioneer's problem to  majorization theory.} \ed{The representation~\eqref{eq_Rochet_Chone_extension} is the starting point for the derivation of the dual in Section~\ref{sec_duality}. Combined with optimal-transportation insights,  it  leads to an algorithm  for computing optimal auctions  (Section~\ref{sec_algorithm}). The existence of optimal multi-item multi-bidder auctions has not been known and required new functional analytic arguments.}
\medskip

\ed{Proposition~\ref{prop_Rochet} allows one to treat auctions with a different number of bidders in a similar way.\footnote{Treating the number of bidders $B$ in~\eqref{eq_Rochet_Chone_extension} as a continuous parameter, one can even interpolate between auctions with different numbers of bidders.} However,} the single-bidder case is special. By plugging in  $\varphi(z)=\max\{0,\,z-1\}$ to~\eqref{eq_Border_expanded}, we see that majorization implies $\frac{\partial u}{\partial x_i} (x)\leq 1$ for any number of bidders. For one bidder, however, the reverse implication also holds as the right-hand side of~\eqref{eq_Border_expanded} is equal to $\varphi(1)$ and $\varphi$ is monotone. Consequently, the dominance condition on  the gradient's distribution boils down to the pointwise condition on the gradient's values and we obtain the classic Rochet-Chon\'e representation~\eqref{eq_Rochet_monopolist} used by~\cite{daskalakis2017strong}.
For $B>1$, the majorization constraint becomes non-local and restricts the distribution of the gradient rather than its pointwise values. As we will see in Section~\ref{sec_duality}, this non-locality is a complication compared to the single-bidder case. 
\medskip 

\ed{To obtain Proposition~\ref{prop_Rochet}, the maximization over individually-rational Bayesian incentive-compatible feasible mechanisms $(P,T)$ in the auctioneer's problem is replaced by the maximization over the corresponding reduced forms $(\overline{P}_b,\overline{T}_b)_{b\in\mathcal{B}}$. By a symmetrization argument, all these one-bidder mechanisms are the same without loss of generality. Thus the auctioneer's problem reduces to maximization of $B$ times the revenue of a Bayesian incentive-compatible  individually-rational one-bidder mechanism~$(\overline{P},\overline{T})$.
 However, not every single-agent mechanism is a reduced form of a feasible $B$-bidder mechanism $(P,T)$ and so we get an extra feasibility constraint
on $(\overline{P},\overline{T})$ originating from the feasibility constraint on $(P,T)$. 

The novelty is in how we handle this feasibility constraint. 
The first characterization of feasible reduced-form mechanisms was proved by~\cite{border1991implementation} but we 
rely on an extension of a less known alternative  characterization by
\cite{hart2015implementation} formulated in terms of majorization: 
a single-bidder mechanism $(\overline{P},\overline{T})$ is a reduced form of some feasible symmetric $B$-bidder mechanism $(P,T)$ if and only if,\footnote{\ed{The upper bound in~\eqref{eq_Hart} corresponds to the reduced form of a mechanism $(P,T)$ allocating each item $i$ to the bidder $b\in\mathcal{B}$ with the highest $x_{i}$. In other words, any reduced form is majorized by the reduced form of the efficient allocation rule.}}
for all items $i\in \mathcal{I}$,
\begin{equation}\label{eq_Hart}
\overline{P}_i(\chi)\preceq \xi^{B-1},
\end{equation}
where $\chi$ is distributed with the density $\rho$ and $\xi$ is uniformly distributed on $[0,1]$.
 
\cite{hart2015implementation} proved this result for $I=1$ item while a version of it derived by \cite{kleiner2021extreme} allows for multiple items but requires one-dimensional types. We need the result for the general setting with $I\geq 2$ items and show that the same dominance condition has to be applied to each of the components of $\overline{P}=(\overline{P}_{i})_{i\in \mathcal{I}}$. The intuition is that the  original feasibility constraint for $(P,T)$ restricts the allocation of each item separately and the constraint for $(\overline{P}, \overline{T})$ inherits this property 

Representing single-agent mechanisms by utility functions  as in the classic Rochet-Chon\'e formula, we obtain the equivalence between the auctioneer's problem  and~\eqref{eq_Rochet_Chone_extension}. This equivalence} allows us to construct a solution to one based on a solution to the other. Hence, to show that the optima are attained, it is enough to demonstrate that the optimum is attained in~\eqref{eq_Rochet_Chone_extension}. This follows from a compactness argument. The set of feasible $u$ is compact 
and the objective is continuous in the \ed{$\sup$-norm topology} of the space of continuous functions. Hence, the optimal $u$ exists since a continuous functional attains its maximal value on a compact set. 
A subtle point is the choice of topology. One might think that the \ed{$\sup$-norm topology}  is too weak to control the gradient and preserve the condition~\eqref{eq_Border} on the gradient's distribution. Indeed,  differentiability is too fine to be preserved by the \ed{$\sup$-norm topology}. However, thanks to the fact that feasible $u$ are convex, the local property of differentiability
can be replaced with a lower bound by an appropriate affine function (see the definition of subdifferential in Appendix~\ref{sect_convex_analysis}) which is respected by \ed{$\sup$-norm} limits.

\section{Duality}\label{sec_duality}
In Section~\ref{sec_Rochet}, we saw that the auctioneer's problem can be reduced to the multi-bidder Rochet-Chon\'e problem,  
which is a
convex program. In this section, we show that for any number of bidders, the dual to this program is a version of Beckmann's transportation problem \citep{beckmann1952continuous}.

\medskip

\ed{In Beckmann's problem,} we are given a cost function $\Phi$, densities of production $\pi_p(x)$ and consumption $\pi_c(x)$ of a certain commodity at every geographical location $x\in X$, where $X$ is a subset of an Euclidean space, and a weight-function $\rho$ on $X$. The goal is to find a transportation flow having the minimal cost and compensating \ed{supply-demand imbalance.}
 The direction and intensity of the flow are represented by a vector field $c=c(x)$. \ed{For a region $A\subset X$, the difference between the weighted outflow and inflow is given by $\int_{\partial A}\langle c(x),n(x)\rangle \rho(x){\dd} s(x)$, where $\partial A$ is the boundary of $A$, the vector $n(x)$ is the outward-pointing unit normal at $x\in \partial A$, and ${\dd} s(x)$ is the element of boundary volume. The compensation of imbalances means that this difference between must be equal to the difference  between supply and demand in this region. Recall that the divergence of a vector field $f=f(x)$ is defined by $\div[f](x)=\sum_i \frac{\partial}{\partial x_i} f(x)$. 
 The compensation
boils down} to the following identity:\footnote{The intuition is as follows. \ed{Consider a ``nice'' region $A\subset X$ (infinitesimally-small cubes are enough).
By the Gauss theorem, $\int_{\partial A}\langle c(x),n(x)\rangle \rho(x){\dd} s(x)=\int_A {\div}[\rho\cdot c]{\dd} x$.} We end up with the condition $\int_A {\div}[\rho\cdot c]{\dd} x=\int_A (\pi_p-\pi_c){\dd} x$ which holds for any $A$ and thus the integrands must be equal.} ${\div}[\rho\cdot c](x)+\pi(x)=0$, \ed{where $\pi=\pi_c-\pi_p$.} 
Beckmann's problem is to minimize the total weighted cost $\int_X \Phi (c) \rho(x){\dd} x$ over all such vector-fields.

In the application to the auctioneer's problem, the set of geographical locations $X$ will coincide with the set of types $X=[0,1]^{\mathcal{I}}$ and the weight $\rho$ will be the density of types' distribution. The supply-demand imbalance $\pi$ will be given by a signed measure which may have singularities. Accordingly, we need to allow the divergence to become a measure as well. To explain the intuition behind the formal definition, for a moment assume that $\rho$ is smooth and equals zero on the boundary of $X$. Then, using the Gauss theorem or just integrating by parts, we obtain that 
\begin{equation}\label{eq_divergence_by_parts_function}
\int_X \langle \nabla u(x),c(x) \rangle \rho(x) {\dd} x= -\int_X  u(x)\cdot \div[\rho\cdot c] {\dd} x
\end{equation}
for any smooth function~$u$  (there is no term corresponding to the contribution of the boundary of $X$ as we assumed that $\rho$ vanishes there).
This formula suggests the formal definition. For a vector field $c$ and weight $\rho$, the 
$\rho$-divergence ${\div}_{\rho}[c]$ is a measure on $X$ such that the integration-by-parts relation 
\begin{equation}\label{eq_rho_divergence_measures}
\int_X \langle \nabla u(x), c(x) \rangle \rho(x) {\dd} x= -\int_X  u(x) {\dd} \left({\div}_{\rho}[c]\right)(x)
\end{equation}
holds for any smooth $u$.  In general, the contribution of the boundary cannot be neglected and so ${\div}_{\rho}[c]$ may have boundary singularities even for smooth $c$ and\footnote{A similar use of measure-valued derivatives can be found in \citep{Ambrosio2000FunctionsOB}.}~$\rho$.
 
\medskip
\noindent\textbf{Beckmann's problem.} \emph{The set of geographical locations is $X=[0,1]^\mathcal{I}$. Spacial imbalance of production and consumption is given by a signed measure $\pi$ on $X$ \ed{such that $\pi(X)=0$, i.e., the total demand is equal to the total supply.} 
Given a convex cost function $\Phi\colon  \R^{\mathcal{I}}\to \R\cup\{+\infty\}$ and a density $\rho\colon X\to \R_+$, the goal is to minimize the cost $\int_X \Phi(c(x))\cdot\rho(x){\dd} x$ over continuously differentiable vector fields $c\colon X\to \R^{\mathcal{I}}$ such that $\div_\rho[c]+\pi=0$.
The value of Beckmann's problem is denoted by
\begin{equation}\label{eq_Beckmann_value}
 \B_\rho\left(\pi,\Phi\right)=\inf_{c\colon \,\div_\rho[c]+\pi=0}  \int_X \Phi(c(x))\cdot\rho(x){\dd} x. 
\end{equation}
If there are no smooth $c$ such that $\div_\rho[c]+\pi=0$, i.e., the minimization is over an empty set, we assume that $\B_\rho\left(\pi,\Phi\right)=+\infty$.
}

\medskip

We now connect Beckmann's problem to auctions. For this purpose, we make the imbalance $\pi$ a free parameter satisfying a majorization constraint. To describe this constraint, consider the revenue objective in the Rochet-Chon\'e problem~\eqref{eq_Rochet_Chone_extension} and get rid of derivatives via integration by parts
\begin{equation}\label{eq_transform_measure}
   \int_{X} \Big(\langle \nabla u(x), x\rangle -u(x)\Big) \rho(x) {\dd} x= -u(0)+\int_X u(x){\dd} m(x), 
\end{equation}
where $m$ is a signed measure 
such that this identity holds for any smooth~$u$. 
We consider the following  majorization constraint on\footnote{The definition of majorization (Definition~\ref{def_majorization}) is applicable to multidimensional signed measures. In particular,~\eqref{eq_pi_majorization} means that $\int u{\dd}\pi\geq \int u{\dd} m$ for any convex non-decreasing $u$ on $X$.}$^{,}$\footnote{\ed{A similar constraint appears in the single-bidder result by~\cite{daskalakis2017strong} who refer to $m$ as the transform measure.}} $\pi$: 
\begin{equation}\label{eq_pi_majorization}
\pi \succeq_{} m.
\end{equation}

To define the cost function $\Phi$, consider a collection $(\varphi_i)_{i\in \mathcal{I}}$  of non-decreasing convex functions on $\R_+$ with $\varphi_i(0)=0$. Let $\varphi_i^*$ be the Legendre transform of $\varphi_i$, i.e., $\varphi_i^*(y)=\sup_x \langle x,y\rangle-\varphi_i(x)$; see Appendix~\ref{sect_convex_analysis}. The cost function $\Phi$ is separable and takes the following form
\begin{equation}\label{eq_Phi}
\Phi(c)=\sum_{i\in\mathcal{I}}\varphi_i^*(|c_i|).
\end{equation}
We note that the higher is $\varphi_i$, the lower is $\varphi_i^*$ and so is the cost $\Phi$.
\begin{theorem}\label{th_vector_fields_inf}
    In the auctioneer's problem~\eqref{eq_revenue} with $|\mathcal{B}|=B\geq 1$ bidders, $|\mathcal{I}|=I\geq 1$ items, and bidders' types distributed on $X=[0,1]^{\mathcal{I}}$ with positive density $\rho$, the optimal revenue coincides with 
\begin{equation}\label{eq_vector_fields_inf}
B\cdot \inf_{
  \footnotesize{\begin{array}{c}
       (\varphi_{i})_{i\in\mathcal{I}},\\  
       \pi\succeq_{} m
  \end{array}}}\left[
  \B_\rho\Big(\pi,\,\Phi\Big)+ \sum_{i\in\mathcal{I}}\int_0^1\varphi_{i}\left(z^{B-1}\right){\dd} z\right],
 \end{equation}   
where $\Phi$ is given by~\eqref{eq_Phi} and $\varphi_{i}\colon \R_+\to \R_+\cup\{+\infty\}$ are non-decreasing convex functions with $\varphi_{i}(0)=0$ for each item $i\in\mathcal{I}$. 
\end{theorem}
Theorem~\ref{th_vector_fields_inf} is a particular case of a more general duality result (Theorem~\ref{th_vector_fields_inf_appendix}) proved in Appendix~\ref{app_Dual}. 
The proof goes in two steps. First, we prove a partial duality result (Theorem~\ref{duality-theorem-maxinf}) internalizing the majorization constraint. It can be interpreted as the equivalence between the auctioneer's problem and the monopolist's problem with adversarial production costs.
We derive a novel a priori bound on the solutions of the latter problem (Proposition~\ref{ae-exist-lip}) with a clear economic interpretation:  the monopolist can guarantee a non-negative revenue not only ex-ante but ex-post.
Then, with this a priori bound, we deduce the complete duality. A byproduct of the proof is that one can assume that the vector field $c$ in Beckmann's problem from~\eqref{eq_vector_fields_inf} has non-negative components. 

\ed{For one item, the optimal vector field in the dual problem coincides with ironed virtual valuation function; see Section~\ref{sec_examples} for details.  In general, the field extends the concept of ironed virtual values to the multi-item case. The functions $\varphi_i$ are shadow prices for the feasibility constraint faced by the auctioneer: increasing the probability that an item is allocated to high types unavoidably decreases this probability for low types. In the context of the partial dual problem, $\varphi_i$ are interpreted as  production costs chosen by the adversary and faced by the monopolist.} 

\medskip
Let us see why the minimization problem~\eqref{eq_vector_fields_inf} is well-defined, i.e., why we minimize over a non-empty set. We need to demonstrate that there is always $\pi\succeq_{} m$ such that $\div_\rho[c]+\pi=0$ for some smooth vector field $c$, and so Beckmann's problem has a finite value. It turns out that we can always  take $\pi=-\div_\rho[x]$ and $c(x)=x$. Let us demonstrate that the majorization condition $\int_X u(x){\dd}\pi(x)\geq\int_X u(x){\dd} m(x)$ holds. We rewrite both sides by the definitions of the divergence and $m$ and get 
$$\int_X\langle x,\nabla u(x)\rangle\rho(x){\dd} x\geq u(0)+ \int_X \left(\langle x,\nabla u(x)\rangle -u(x)\right)\rho(x){\dd} x.$$
The dot-product terms cancel out and we end up with an equivalent inequality $\int_X u(x)\rho(x){\dd} x\geq u(0)$ that holds for any non-decreasing $u$. We conclude that the problem~\eqref{eq_vector_fields_inf} has a finite value. Moreover, we obtain that the  auctioneer's optimal revenue is bounded from above by
\begin{equation}
\label{rhozetB}
B \cdot \inf_{
  \footnotesize{\begin{array}{c}
       (\varphi_{i})_{i\in\mathcal{I}}
  \end{array}}}
 \sum_{i\in \mathcal{I}}\left( \int_X  \varphi_i^*(x_i) \rho(x){\dd} x+ \int_0^1 \varphi_{i}\left(z^{B-1}\right){\dd} z\right).
\end{equation}
In this upper bound, the minimization splits into a family of $I$ identical one-dimensional minimization problems, one for each item $i\in I$. They can be solved explicitly and the resulting bound corresponds to full surplus extraction; see Appendix~\ref{app_full_surplus}.


\medskip

\subsection{Weak duality and complementary slackness}
Strong duality results such as Theorem~\ref{th_vector_fields_inf} can be seen as a combination of two statements: that the value of the primal problem is at most the value of the dual (weak duality) and that the gap between the values is zero. While the weak duality is always an easy part of the proof, this part is insightful as it explains the form of the dual and leads to complementary slackness conditions.

 Let us see why the weak duality holds, i.e., why the optimal revenue is bounded from above by~\eqref{th_vector_fields_inf}. We know that the optimal revenue equals to $B\cdot \int_X \left(\langle x,\nabla u(x)\rangle -u(x)\right)\rho(x){\dd} x$ for some convex non-decreasing function $u$ with $u(0)=0$ and such that the constraint~\eqref{eq_Border_expanded} by \cite{hart2015implementation} is satisfied (Proposition~\ref{prop_Rochet}).
Hence, the optimal revenue does not exceed
\begin{multline}\label{eq_payoff_function}
B\cdot \Bigg[\int_X \left(\langle x,\nabla u(x)\rangle -u(x)\right)\rho(x){\dd} x\\
+ \sum_{i\in\mathcal{I}}\left(\int_0^1\varphi_{i}\left(z^{B-1}\right){\dd} z- \int_X \varphi_i\left(\frac{\partial u}{\partial x_{i}} (x)\right)\rho(x) {\dd} x\right)\Bigg]
\end{multline}
for any  non-decreasing convex functions $\varphi_i$ on $\R_+$ with $\varphi_i(0)=0$  (each term in the sum is non-negative by the constraint of \cite{hart2015implementation}). The first integral can be rewritten as follows 
\begin{multline}\label{eq_first_integral_bound}
 \int_X \left(\langle x,\nabla u(x)\rangle -u(x)\right)\rho(x){\dd} x=\int_X u(x){\dd} m(x) \\
 \leq\int_X u(x){\dd}\pi(x)=\int_X\langle\nabla u(x),c(x)\rangle \rho(x){\dd} x,   
\end{multline}
 where $m$ is the transform measure from~\eqref{eq_transform_measure}, $\pi$ is an arbitrary measure such that $\pi\succeq_{} m$ and $c$ is any vector field such that $\div_\rho[c]+\pi=0$. The first equality holds by the definition of the transform measure, the inequality holds thanks to convexity of $u$, and the last equality is by the definition of divergence~\eqref{eq_rho_divergence_measures}. The Fenchel inequality (inequality~\eqref{eq_Fenchel_inequality} in Appendix~\ref{sect_convex_analysis}) applied to $\psi_i(t)=\varphi_i(|t|)$ implies the following bound on the last integrand
\begin{equation}\label{eq_Fenchel_application}
\langle \nabla u(x),\, c(x)  \rangle\leq \sum_{i\in\mathcal{I}}\varphi_{i}^*\big(\big|c_{i}(x)\big|\big)+\sum_{i\in\mathcal{I}}\varphi_{i}\left(\frac{\partial u}{\partial x_{i}}(x)\right),
\end{equation}
where we used that $\psi_i^*(t)=\varphi_i^*(|t|)$ and non-negativity of $u$'s partial derivatives. 
Replacing the first summand in~\eqref{eq_payoff_function} by the resulting upper bound,  we see that the terms with  partial derivatives of $u$ cancel out and the revenue is bounded from above by
$$B\cdot \left[\int_X \left(\sum_{i\in\mathcal{I}}\varphi_{i}^*\big(|c_{i}(x)|\big)\right)\rho(x){\dd} x+ \sum_{i\in\mathcal{I}}\int_0^1\varphi_{i}\left(z^{B-1}\right){\dd} z\right]$$
for all convex $\varphi_i$ with $\varphi_i(0)=0$, all measures $\pi\succeq_{} m$, and smooth vector fields $c$ such that $\div_\rho[c]+\pi=0$. Taking infimum over all such $\varphi_i$, $\pi$, and $c$, we conclude that the optimal revenue cannot exceed the right-hand side of~\eqref{eq_vector_fields_inf} thus establishing the weak duality.
\medskip

Complementary slackness conditions are a byproduct of the above computation. Let $u^\opt$, $\varphi_i^\opt$, $\pi^\opt$, and $c^\opt$ be the optima in the primal Rochet-Chon\'e problem~\eqref{eq_Rochet_Chone_extension}, the dual problem~\eqref{eq_vector_fields_inf}, and  internal Beckmann's problem, respectively. We know that $u^\opt$ exists by Proposition~\ref{prop_Rochet} and the existence of the rest of the optima is discussed below. For now, we assume that all of them exist.
Under this assumption, the only way the value of the primal problem can be equal to the value of the dual~\eqref{eq_vector_fields_inf} is if each inequality in the derivation of the weak duality holds as equality at $u^\opt$, $\varphi_i^\opt$, $\pi^\opt$, and $c^\opt$. Namely, each term in the sum from~\eqref{eq_payoff_function} must be zero, and the inequality in~\eqref{eq_first_integral_bound} together with the Fenchel inequalities used to derive~\eqref{eq_Fenchel_application} must all be equalities.  
These observations, combined with the complementary slackness condition for the Fenchel inequality (see Appendix~\ref{sect_convex_analysis}), lead to the following corollary.  
\begin{corollary}[Complementary slackness]\label{cor_slackness}
Optimal $u^\opt$, functions $\varphi_i^\opt$, measure $\pi^\opt$, and vector field $c^\opt$ satisfy the following family of conditions:
\begin{align}
\int_{X}  \varphi_{i}^\opt\left(\frac{\partial u^\opt}{\partial x_{i}}(x) \right) \rho(x) {\dd} x&=\int_0^1\varphi_{i}^\opt\left(z^{B-1}\right){\dd} z
\label{eq_slackness_adversary} \\
\int_X u^\opt(x){\dd} m(x) &= \int_X u^\opt(x) {\dd} \pi^\opt(x) \label{eq_slackness}\\
c_{i}^\opt(x) &\in  \partial \varphi_{i}^\opt\left(\frac{\partial u^\opt}{\partial x_{i}}(x)\right)\label{eq_vector_field_slackness}
\end{align}
In the last condition, $\partial$ denotes the subdifferential~\eqref{eq_subdifferential} and the inclusion holds for $\rho$-almost all $x\in X$.
\end{corollary}
\ed{Complementary slackness conditions have the following structural implications. Applying the Jensen inequality to~\eqref{eq_slackness_adversary} and taking into account convexity of $\varphi_i^\opt$, we see that the distribution of $\frac{\partial u^\opt}{\partial x_{i}}(\chi)$ with $\chi\sim \rho$ can differ from that of $\xi^{B-1}$ with $\xi\sim\mathrm{Uniform}([0,1])$ only over those regions where $\varphi_i^\opt$ is flat. Similarly,  condition~\eqref{eq_slackness} implies that $\pi^\opt$ can differ from the transform measure $m$ only where $u^\opt$ is flat. From the last condition, we obtain that $c_i^\opt$ is non-negative and non-decreasing.}

\medskip
\subsection{Existence} Whether the optima exist or not may seem a technical peculiarity. The importance of this question is justified by the complementary slackness conditions (Corollary~\ref{cor_slackness}) \ed{which hold only} if both primal and dual problems attain their optima. 

We know that the optimal value of the Rochet-Chon\'e problem~\eqref{eq_Rochet_Chone_extension} is attained at some $u^\opt$. It turns out that the family of optimal functions $\varphi_i^\opt$ in the dual problem~\eqref{eq_vector_fields_inf} also always exists and corresponds to an optimal strategy of an adversary in the auxiliary monopolist's problem with adversarial production costs discussed in Appendix~\ref{sec_adversary}.

We note that Beckmann's problem is prone to absence of an optimal smooth vector field $c^\opt$ even for standard cost functions such as $\Phi(c)=\|c\|$. A workaround is to allow for generalized vector fields by replacing a smooth vector field $c$ by a vector measure $\varsigma$. Then the optimal vector measure $\varsigma$ is known to exist provided that the supply-demand imbalance  $\pi$ is absolutely continuous and, moreover, $\varsigma$  itself turns out to be absolutely continuous
\citep[Theorem 4.16]{santambrogio2015optimal}.
\ed{In our setting, the transform measure $m$ typically has singularities on the boundary of $X$ inherited by $\pi\succeq_{} m$.}

To guarantee existence, we allow for generalized vector fields given by vector measures allowing for singular components. 
The divergence of a vector measure may not be a measure anymore and can only be defined in the space of generalized functions \citep{Ambrosio2000FunctionsOB}.  As $\pi=-\div_\rho[c]$, following this approach we would need to allow $\pi$ to become a generalized function as well. We avoid this complication by reformulating the constraint on the vector field bypassing $\pi$.
  
Consider  the set $\mathcal{C}^\mes$ of  non-negative vector measures
 $\varsigma=(\varsigma_{i})_{i\in\mathcal{I}}$ satisfying the following condition 
\begin{equation}\label{eq_condition_on_c_mes}
\int_X \bigl(\langle \nabla u(x),\, x \rangle - u(x) \bigr)\cdot \rho(x) {\dd} x \leq \sum_{i\in\mathcal{I}}\int_X \frac{\partial u}{\partial x_{i}}(x){\dd} \varsigma_{i}(x)
\end{equation}
for any smooth non-decreasing convex $u: X\to \R_+$ with $u(0)=0$. 
By the Lebesgue decomposition theorem, each  $\varsigma_{i}$ can be represented as the sum of the component that is absolutely continuous with respect to $\rho(x){\dd} x$ and the singular one. We get
\begin{equation}\label{eq_Lebesgue_decomposition}
{\dd}\varsigma_{i} = c_{i}(x)\cdot \rho(x){\dd} x + {\dd}\varsigma_{i}^{\sing}(x).    
\end{equation}
If the singular component is absent and $c=(c_i)_{i\in \mathcal{I}}$ is smooth, we can define 
$\pi=-\div_\rho[c]$ and see that the condition~\eqref{eq_condition_on_c_mes} is equivalent to the familiar majorization condition~\eqref{eq_pi_majorization} on $\pi$.

The following extension of Theorem~\ref{th_vector_fields_inf} guarantees that the optimum in the dual is attained. It is proved in Appendix~\ref{app_Dual}.
\begin{theorem}[Extended dual]\label{th_vector_fields_min}
The optimal revenue in the auctioneer's problem~\eqref{eq_revenue} coincides with 
\begin{equation}\label{eq_vector_fields_min}
B\cdot \min_{
  \footnotesize{\begin{array}{c}
       (\varphi_{i})_{i\in\mathcal{I}},\\  \varsigma \in \mathcal{C}^\mes 
  \end{array}}}
 \sum_{i\in \mathcal{I}}\left(\varsigma_{i}^\sing(X)+\int_X  \varphi^*_{i}\big(c_{i}(x)\big) \rho(x){\dd} x +\int_0^1\varphi_{i}\left(z^{B-1}\right){\dd} z\right)
 \end{equation} 
 and the minimum 
 is attained. Here $B$ is the number of bidders, $c_{i}$ and $\varsigma_{i}^\sing$ are given by~\eqref{eq_Lebesgue_decomposition}, and $\varphi_{i}$ are non-decreasing convex functions with $\varphi_{i}(0)=0$ for each item $i\in\mathcal{I}$. 
\end{theorem}
Note that  the objectives in Theorems~\ref{th_vector_fields_inf} and~\ref{th_vector_fields_min} 
match one another except for the fact that some mass in the extended dual can be transferred from the vector field $c$ to the singular component of the vector measure. This additional flexibility turns out to be crucial for  the existence of the optimum. 

One may think that the 
appearance of singular measures is an artifact of a particular proof technique and that singularities do not appear at least in nice examples. This intuition turns out to be wrong and singular measures happen to reflect the essence of the problem. In Appendix~\ref{app_examples}, we solve the dual problem explicitly for two uniform items and $B=1$ bidder and see that, even in this simplest case, 
there are singularities on the boundary of the set of types $X$. 

Theorem~\ref{th_vector_fields_min} allows us to write down the complementary slackness conditions without making an extra assumption that the optima exist. 
\begin{corollary}[Extended complementary slackness]\label{cor_slackness_extended}
 Consider optimal $u^\opt,$ $(\varphi_{i}^\opt)_{i\in\mathcal{I}},$ and ${\varsigma}^\opt$ and decompose  ${\varsigma}^\opt$ into absolutely-continuous and singular components as in~\eqref{eq_Lebesgue_decomposition}. Then all the previously discussed complementary slackness conditions~\eqref{eq_slackness},~\eqref{eq_vector_field_slackness}, and~\eqref{eq_slackness_adversary} hold. Moreover, there is one more condition:
 \begin{equation}\label{eq_slackness_singular}
\frac{\partial u^\opt}{\partial x_{i}}(x)=1
 \end{equation}
for $\varsigma_{i}^{\opt,\,\sing}$-almost all $x$. In particular, $u^\opt$  has a partial derivative with respect to $x_{i}$ for $\varsigma_{i}^{\opt,\,\sing}$-almost all $x$.
\end{corollary}
\ed{The majorization constraint and convexity of $u^\opt$ imply that $\frac{\partial u^\opt}{\partial x_{i}}$ is non decreasing and $\frac{\partial u^\opt}{\partial x_{i}}\leq 1$. We conclude that the singular component of the vector field can only be supported on the northeast boundary of $X$.}

\section{Examples, applications, and simulations}\label{sec_examples}
\ed{We start by discussing the connection between Beckmann's dual problem and that by \cite{daskalakis2017strong} in the one-bidder case and demonstrate  that the latter can be deduced from the former. Then we consider the classic case of one-item multi-bidder auctions and show that the dual vector field in Beckmann's problem is given by ironed virtual valuations. Building on this insight, we obtain that, in the multi-item case with independent values, running separate auctions for each of the items is never optimal.  Finally, we explore the structure of optimal multi-item multi-bidder auctions via numerical simulations.}
\subsection{One bidder: relation to \texorpdfstring{\cite{daskalakis2017strong}}{[Daskalakis et al., 2017]}}

In Theorem~\ref{th_vector_fields_inf}, we saw that the dual to the auctioneer's problem is given by Beckmann's transportation problem for any number of bidders $B\geq 1$. For $B=1$ bidder, \cite{daskalakis2017strong} derived another dual taking a form of the Monge-Kantorovich optimal transportation problem~\eqref{eq_Kantorovich}. 
It is not surprising that the duals for $B\geq 2$ bidders and $B=1$ bidder do not share any similarity as the feasibility constraint for several bidders becomes non-local and so the approach of \cite{daskalakis2017strong} is not applicable.
Here we focus on the case of $B=1$ bidder, where both approaches can be used and so the lack of similarity between the two duals may seem surprising.


It turns out that the dual from Theorem~\ref{th_vector_fields_inf} can be simplified in the single-bidder case. Indeed,  $z^{B-1}\equiv 1$ for $B=1$ bidder and so the second integral reduces to $\int_0^1 \varphi_{i}\left(z^{B-1}\right){\dd} z=\varphi_i(1)$. We obtain that the value of the auctioneer's problem is equal to
$$
\inf_{
  \footnotesize{\begin{array}{c}
       (\varphi_{i})_{i\in\mathcal{I}},\\  
       \pi\succeq_{} m
  \end{array}}}\left[
  \B_\rho\Big(\pi,\,\Phi\Big)+ \sum_{i\in\mathcal{I}} \varphi_{i}(1)\right]$$
 with $\Phi(c)=\sum_i \varphi_i^*(|c_i|)$.
 This expression can be further simplified. The lower  the cost function $\Phi$ in Beckmann's problem is, the lower is its value. By increasing $\varphi_i$ pointwise, we decrease its conjugate $\varphi_i^*$. Hence, keeping $\varphi_i(1)$ fixed, the best choice given the requirements of convexity and $\varphi_{i}(0)=0$ is the linear function: 
 $ \varphi_{i}(t) = \varphi_{i}(1)\cdot t$ on $[0,1]$ and $\varphi_{i}(t) = + \infty$ for   $t >1$. Optimization over $\varphi_{i}(1)$ gives $\varphi_{i}(1) =0$ and thus the conjugate $\varphi^{*}_{i}(t) = t$ for all $t$. 
 
 We obtain Beckmann's problem with the cost function given by $l^1$-norm $\Phi(c)=\|c\|_1=\sum_{i\in\mathcal{I}}|c_i|$.  Importantly, this cost function is $1$-homogeneous. Beckmann's problem with a $1$-homogeneous $\Phi$ has a peculiar property: its value does not depend on the density $\rho$ provided that it is smooth and positive, i.e., $\B_\rho\Big(\pi,\,\Phi\Big)=\B_{1}\Big(\pi,\,\Phi\Big)$, where in the second problem the density is equal to~$1$. 
 This property holds, since for any feasible vector field $c$ in the second problem,  $c'=\rho\cdot  c$ is a feasible vector field in the first problem with the same value. 
 \begin{corollary}
 \label{1bidderbeckmann}
 For $B=1$ bidder whose type is distributed according to a smooth positive density $\rho$, the optimal revenue of the auctioneer~\eqref{eq_revenue}  is equal to
  \begin{equation}
 \inf_{
  \footnotesize{\begin{array}{c}
       \pi\succeq_{} m
  \end{array}}}
  \B_1\Big(\pi,\, \|\cdot\|_1\Big),
 \end{equation} 
 where the cost function is given by $l^1$-norm $\|c\|_1=\sum_{i\in\mathcal{I}}|c_i|$.
 \end{corollary}
Beckmann's problem with the Lebesgue reference measure and the 
cost function~$\|\cdot\|_1$ is an exception where the Beckmann's problem is known to be connected to the Monge-Kantorovich one.\footnote{More generally, there is a connection between Beckmann's problem and congested optimal transportation problems of Monge-Kantorovich type; see the discussion in  Appendix~\ref{app_Beckmann}.} The so-called Beckmann's duality states that, for any $\pi$,
$$\B_1\Big(\pi,\, \|\cdot\|_1\Big)=
\min_{\gamma:\ \gamma_1=\pi_c, \ \gamma_2=\pi_p}\int \|x-y\|_1{\dd}\gamma(x,y),$$
where $\pi_c$ and $\pi_p$ are the positive and the negative parts of $\pi$, respectively, and the minimum is taken over positive measures $\gamma$ on $X\times X$ with marginals $\pi_c$ and $\pi_p$ \cite[Section 4.2]{santambrogio2015optimal}. Combining this identity with Corollary~\ref{1bidderbeckmann}, we obtain the dual problem in the form of  \cite{daskalakis2017strong}.
 \begin{corollary}[\cite{daskalakis2017strong}]
 \label{cor_Daskalakis_from_Beckmann}
  For $B=1$ bidder whose type is distributed according to a smooth positive density $\rho$, the optimal revenue of the auctioneer~\eqref{eq_revenue}  is equal to
  \begin{equation}\label{eq_Daskalakis_result}
 \inf_{
  \footnotesize{\begin{array}{c}
       \pi\succeq_{} m 
  \end{array}}} \min_{\gamma:\ \gamma_1=\pi_c, \ \gamma_2=\pi_p}\int \|x-y\|_1{\dd}\gamma(x,y).
 \end{equation} 
 \end{corollary}

\ed{
\subsection{One item: optimal vector fields as virtual valuations}\label{sec_one_item}
Consider one item and several bidders with values having a smooth strictly positive density $\rho$ on $[0,1]$. In this case, the vector field $c=c(x)$ from Theorem~\ref{th_vector_fields_inf} is a scalar related to virtual valuations \citep{myerson1981optimal}.

The condition that $\div_\rho[c]+\pi=0$ for some $\pi\succeq m$ in Beckmann's problem boils down to the requirement that 
\begin{equation}\label{eq_one_item_feasibility}
\int \big(x\cdot u'(x)-u(x)\big)\rho(x){\dd} x \leq \int u'(x) c(x) \rho(x) {\dd} x 
\end{equation}
for any convex monotone $u$ with $u(0)=0$.
Integrating by parts,  the left-hand side 
can be rewritten as $\int V(x)u'(x)\rho(x){\dd} x$, where $V(x)=x-\frac{1-F(x)}{\rho(x)}$ is 
the virtual valuation function and $F(t)=\int_0^t \rho(x){\dd} x$. As $u'$ is an arbitrary monotone function, we conclude that $c$ satisfies~\eqref{eq_one_item_feasibility} if and only if 
\begin{equation}\label{eq_virtual_lower_bound}
\int_t^1 V(x)\rho(x){\dd} x \leq \int_t^1 c(x)\rho(x){\dd} x\quad \mbox{for any $t$}.
\end{equation}
It turns out that this lower bound determines optimal field $c^\opt$ if we take into account the restrictions imposed by the complementary slackness condition~\eqref{eq_vector_field_slackness}, namely, $c^\opt$ is non-negative and non-decreasing.  Under the standard assumption that $V(x)$ is non-decreasing, we get
$$c^\opt(x)=\max\{0,\,V(x)\}.$$ 
More generally, we obtain the following result proved in Appendix~\ref{app_examples}.
\begin{proposition}\label{prop_ironed}
For one item and several bidders, the optimal vector field  $c^\opt(x)$ is equal to the ironed virtual valuation function.
\end{proposition}
Let us see why the result holds in the case of non-decreasing virtual valuations. For this purpose, we first guess  $u^\opt$ and $\varphi^\opt$ using the complementary slackness conditions. Assuming that~\eqref{eq_slackness_adversary} is binding for high types, we conclude that $(u^\opt)'(x)=\left(F(x)\right)^{n-1}$ for $x\geq x_0$ and some $x_0<1$. After integrating by parts, condition~\eqref{eq_slackness} becomes 
$$\int V(x) (u^\opt)'(x)  \rho(x){\dd} x= \int c^\opt(x) (u^\opt)'(x)\rho(x){\dd} x,$$
where we took into account that $(c^\opt)'(x)+\pi^\opt(x)=0$. This suggest defining $x_0$ by $V(x_0)=0$. Thus $$u^\opt(t)=\left\{\begin{array}{cc}\int_{x_0}^t \left(F(x)\right)^{n-1}{\dd} x,  & t\geq x_0\\ 0, &\mbox{otherwise}\end{array}\right.$$ 
Now, we can define $\varphi^\opt$ using~\eqref{eq_vector_field_slackness}: 
$$\varphi^\opt(t)=\left\{\begin{array}{cc}\int_{\left(F(x_0)\right)^{n-1}}^t  V\Big(\big(F^{n-1}\big)^{-1}(x)\Big){\dd} x, & t\geq \left(F(x_0)\right)^{n-1}\\
0, &\mbox{otherwise}\end{array}\right.$$
The functions $u^\opt$ and $\varphi^\opt$ are convex and the triplet $(u^\opt,c^\opt,\varphi^\opt)$ satisfies the complementary slackness conditions by the construction. This verifies all the guesses made along the way and implies optimality of $u^\opt,c^\opt,$ and $\varphi^\opt$ in the primal and dual problems. 

A straightforward modification of this guess-and-verify approach is applicable to non-monotone $V$ and implies Proposition~\ref{prop_ironed}. By such an explicit construction, we obtain that for one item, the optimum in the dual problem is attained even if singular vector fields from Theorem~\ref{th_vector_fields_min} are not allowed.

In Appendix~\ref{app_examples}, we prove a stronger statement: even if the vector field is allowed to have singularities, the optimum is attained on the classical solution described above. This general result is needed for the application that we discuss next.

\subsection{Several bidders with independent items: suboptimality of selling separately}

Consider $B\geq 2$ bidders whose values over $I\geq 2$ items are distributed with density $\rho(x)=\rho_1(x_1)\cdot\ldots\cdot \rho_I(x_I)$, i.e., values for different items are independent. For simplicity, assume that each $\rho_i$ is continuous and strictly positive on $[0,1]$.

A natural idea is to sell independent items separately by running optimal Myersonian auction for each of them. \cite{jehiel2007mixed} showed that this idea never leads to optimal revenue. We demonstrate how this observation becomes a simple corollary of duality and the just solved case of one item. Indeed, if selling separately was optimal, then the optimal vector field $c=c^\opt(x_1,\ldots, x_I)$ would have the following structure:
$$c(x_1,\ldots, x_I)=\Big(c_1^\opt(x_1),\ldots, c_I^\opt(x_I)\Big),$$
where $c_i^\opt(x_i)$ corresponds to the optimal one-item mechanism for distribution~$\rho_i(x_i)$.

To demonstrate that such $c$ cannot be optimal in the dual problem, it is enough to show that it is infeasible, namely, $\pi$ defined by $\div_\rho[c]+\pi=0$ does not dominate the transform measure $m$ from~\eqref{eq_transform_measure}. Equivalently, we need to demonstrate that there is a convex monotone function $u=u(x_1,
\ldots, x_i)$ with $u(0)=0$ such that
\begin{equation}\label{eq_c_feasibility_selling_separately}
\int_X \bigl(\langle \nabla u(x),\, x \rangle - u(x) \bigr)\cdot \rho(x) {\dd} x \leq \int_X \langle \nabla u(x),\, c(x)\rangle \rho(x){\dd} x.
\end{equation}
From Section~\ref{sec_one_item}, we know that $c_i^\opt$ are given by ironed virtual valuation functions. Since the highest types are never ironed, there is a constant $a<1$ such that, on the interval  $[a,1]$, the components $c_i^\opt$ coincide with the corresponding virtual value functions $V_i$.

Consider a convex monotone function $u$ that is non-zero in the region $\min_i x_i\geq a$ only. For example, we can take $u(x)=\max\big \{0,\, \sum_i x_i- I\cdot a\big\}$. Thus
\begin{align}
\int_X \frac{\partial u}{\partial x_i}(x) c_i^\opt(x_i) \rho(x){\dd} x&=
\int_X \frac{\partial u}{\partial x_i}(x) V_i(x_i) \rho(x){\dd} x=\\
&=\int_X \left(\frac{\partial}{\partial x_i} u(x)\cdot x_i -u(x)\right) \rho(x){\dd} x,
\end{align}
where the second identity is obtained via integration by parts. Thus
$$\int_X \langle \nabla u(x),\, c(x)\rangle \rho(x){\dd} x= \int_X \bigl(\langle \nabla u(x),\, x \rangle - I\cdot u(x) \bigr)\cdot \rho(x) {\dd} x.$$
Since $I\geq 2 $, this expression is less than the left-hand side of~\eqref{eq_c_feasibility_selling_separately}. We conclude that $c$ is infeasible in the dual problem and thus selling separately cannot be optimal for $I\geq 2$ items. The detailed argument can be found in Appendix~\ref{app_examples}.
}

\subsection{Several bidders and several items: optimal auctions via simulations}\label{sec_algorithm}

\ed{The guess-and-verify approach illustrated in the one-item case and also applicable in the one-bidder case (see Appendix~\ref{app_examples}) can also be applied to the multi-bidder multi-item setting. The starting point for this approach is an explicit guess about the optimal mechanism. 
In this section, we explore the case of $B\geq 2$ bidders with i.i.d. uniformly distributed values over $I=2$ items using numerical simulations. The algorithmic insights are discussed below. The simulations indicate a complicated structure of the optimal mechanism and suggest that the optimal auction may not admit a closed-form solution even in this benchmark setting.


For two bidders,} the solution to the primal problem of revenue maximization is shown in Figure~\ref{fig_optimal_primal}. This figure depicts the probability to receive the first item as a function of bidder's values $(x_1,x_2)$, i.e.,  the optimal reduced allocation rule  $\overline{P}_1^\opt(x_1,x_2)=\frac{\partial}{\partial x_1} u^\opt (x_1,x_2)$. The probability for the second item can be obtained by symmetry: $\overline{P}^\opt_2(x_1,x_2)=\overline{P}^\opt_1(x_2,x_1)$.
The discontinuity that we see in Figure~\ref{fig_optimal_primal} correspond to the multi-dimensional reserve price: the minimal $x_1$ to receive a non-zero portion of the first item non-linearly depends on $x_2$ unless $x_2$ is high enough.
\begin{figure}
    \centering
    \includegraphics[width=5cm]{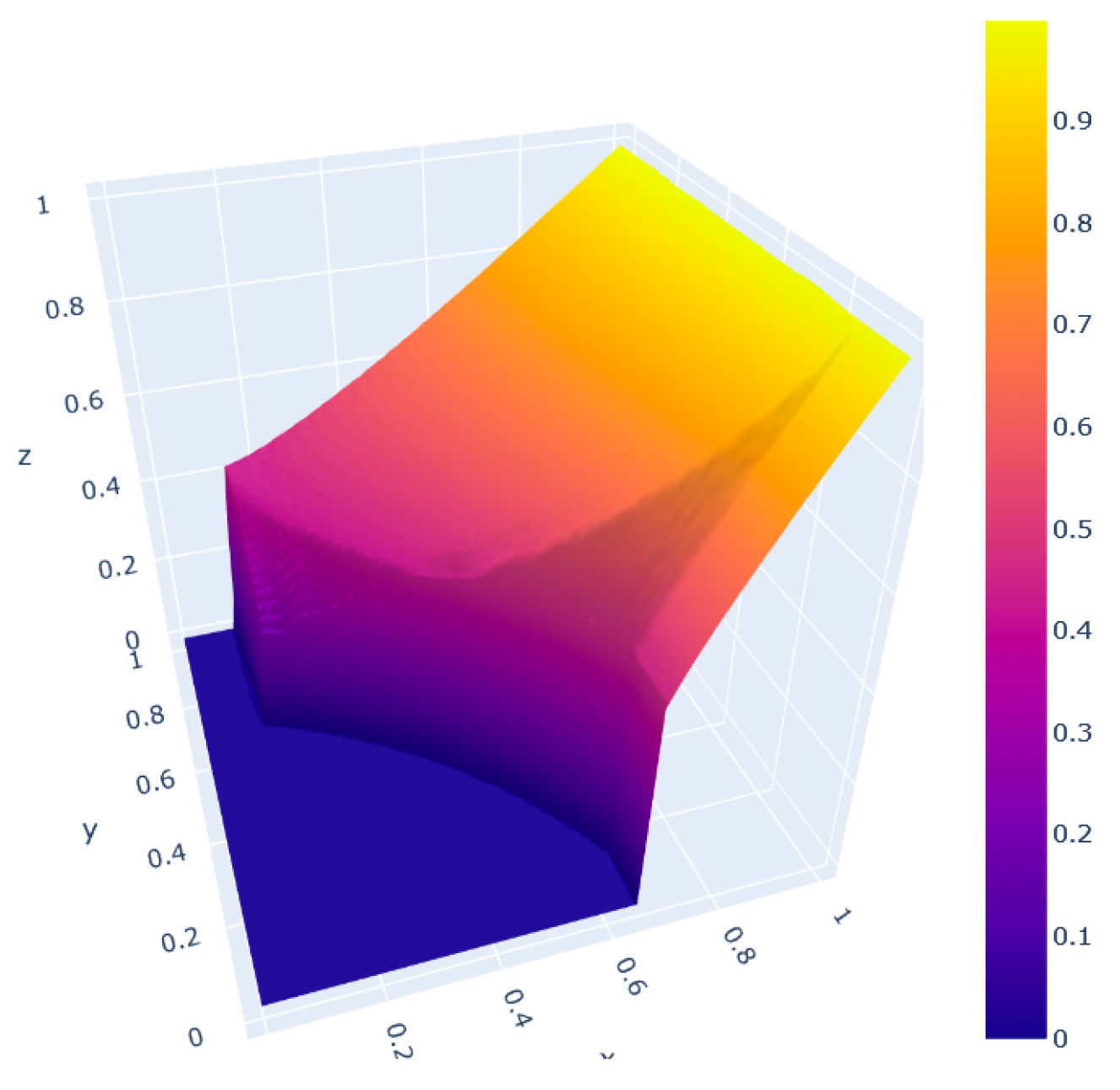}\hskip 2 cm
    \includegraphics[width=5cm]{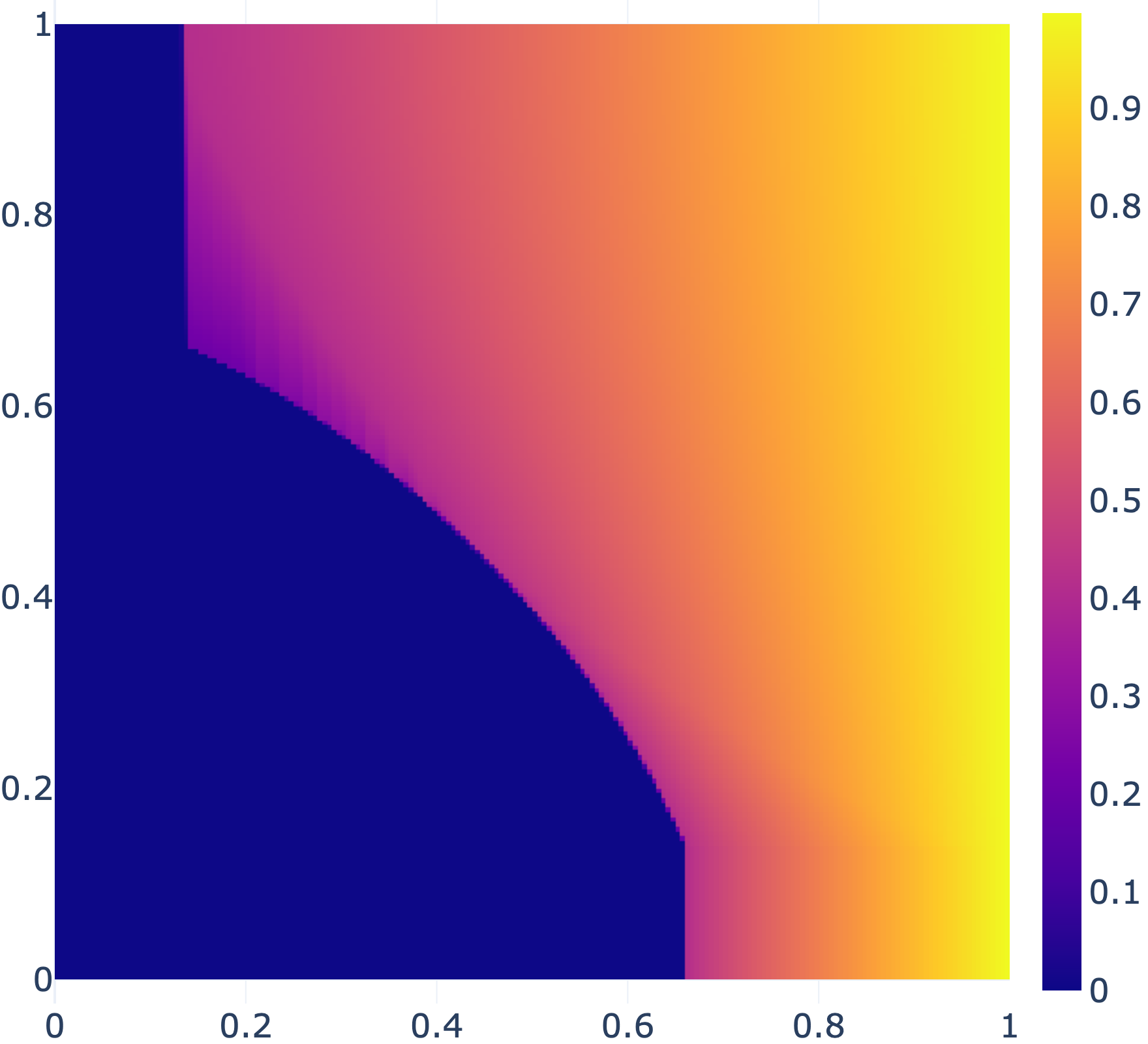}
    \caption{The probability to receive the first item as a function of bidder's values $(x_1,x_2)$ in the optimal  $2$-bidder $2$-item auction with i.i.d.  values uniform on $[0,1]$. }
    \label{fig_optimal_primal}
\end{figure}
The solution to the dual problem is shown in Figure~\ref{fig_dual}. The contour plot demonstrates the first component $c_1^\opt$ of the optimal vector field~$c^\opt$; the second component can be obtained by $c_2^\opt(x_1,x_2)=c_1^\opt(x_2,x_1)$.  By the complementary slackness condition~\eqref{eq_vector_field_slackness}, we have $c_{i}^\opt(x) \in  \partial \varphi_{i}^\opt\left(\frac{\partial u^\opt}{\partial x_{i}}(x)\right)$ and so one could expect that  the vector field inherits the discontinuity of $\partial u$. The optimal vector field turns out to be continuous because the optimal $\varphi_i$ are zero in the  discontinuity region. 
 \begin{figure}
    \centering
    \includegraphics[width=6.5cm]{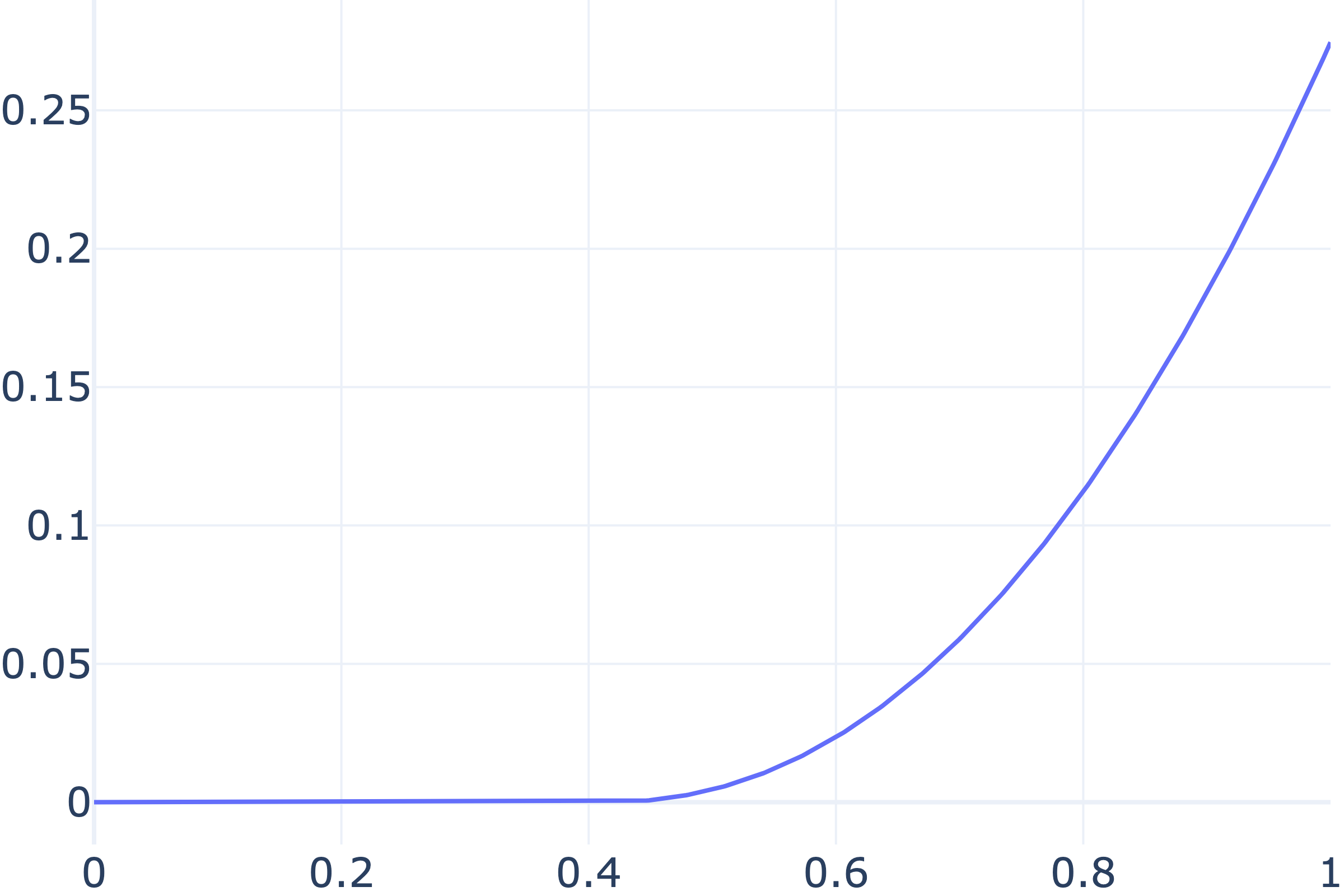}\hskip 1 cm
    \includegraphics[width=5cm]{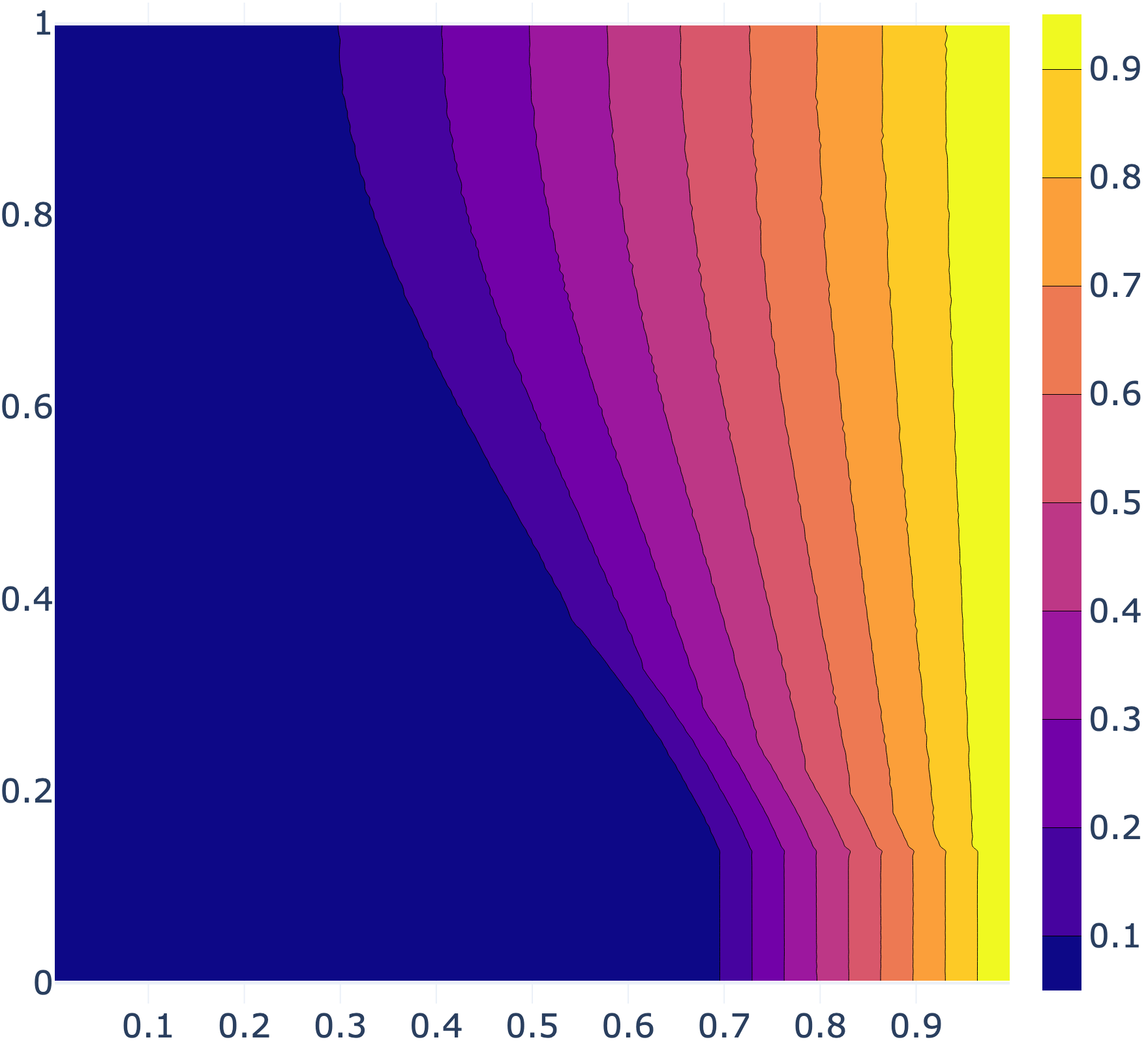}
    \caption{The optimal solution to the dual problem: functions $\varphi_1=\varphi_2$ (left) and a contour plot  of the first component of the  vector field $c=(c_1,c_2)$ from Beckmann's problem (right). }
    \label{fig_dual}
\end{figure}
\ed{None of $\overline{P}_i^\opt, c_i$ and $\varphi_i$
seem to be given by elementary functions in any of the regions: even those parts that look linear or quadratic are, in fact, not.}

 \ed{For $B\geq 2$ bidders, we computed how the optimal revenue depends on~$B$.} Figure~\ref{fig_revenue} depicts this dependence. Naturally, the optimal revenue is bounded from below by the revenue obtained from selling the items separately using Myerson's optimal auction and, from above, by the revenue that the auctioneer would get if she could extract the full surplus.\footnote{Revenue of Myerson's  auction run for each item  separately is  $2B\int_{0.5}^1 (2x-1)x^{B-1}{\dd} x$ while the full surplus is  $2\left(1-\frac{1}{B+1}\right)$.} We see that the advantage from using the optimal auction is substantial for small number of bidders and it is maximal for $B=2$ where the optimal mechanism increases the revenue by $5.7\%$. For large number of bidders, the use of optimal auction is not justified as selling the items separately leads to almost full surplus extraction.\footnote{\ed{Selling separately via optimal posted price mechanisms as well as using the optimal posted price mechanism for the grand bundle} extract $1-O\left(\frac{1}{B}\right)$ fraction of the full surplus, as $B\to\infty$.}
 \begin{figure}
    \centering
    \includegraphics[width=8cm]{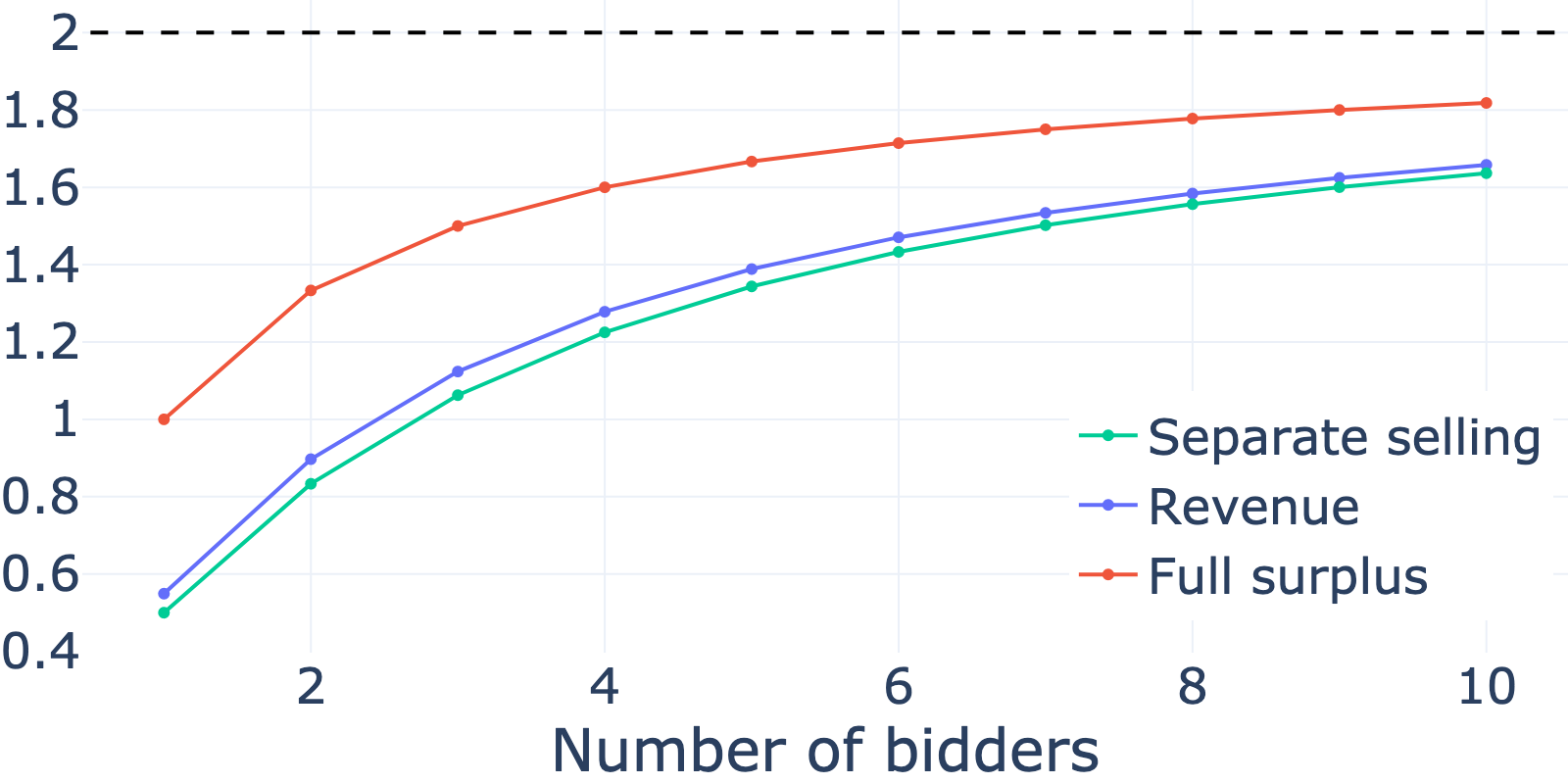}
    \caption{Revenue as a function of the number of bidders $B$ for two items with i.i.d. values uniform on $[0,1]$. Graphs from bottom to top: selling separately (light-green), selling optimally (blue), full surplus extraction (red), limit for $B\to\infty$ (the dashed line).}
    \label{fig_revenue}
\end{figure}

\subsubsection{Algorithm} Here we describe high-level ideas behind the algorithm. The detailed description and proofs can be found in Appendix~\ref{app_numerical}. 
As discussed in the introduction, finding a solution numerically is far from being straightforward: although the auctioneer's problem is a linear program in a functional space, any reasonable discretization of it cannot be handled by modern LP solvers because of the curse of dimensionality. Indeed, if an agent can have $n$ different values for each of $I$ items, then the mechanism should specify an allocation and transfers for each of $\big(n^I\big)^B$ profiles of types which becomes computationally intractable already for two items, $n=100$, and $B=2$ agents or for $n=10$, and $B=4$ agents.

We escape the curse of dimensionality by dealing with the Rochet-Chon\'e problem~\eqref{eq_Rochet_Chone_extension} which is equivalent to the auctioneer's problem by Proposition~\ref{prop_Rochet}. For $n$ points in the discretization, the dimension of the Rochet-Chon\'e problem is constant in the number of bidders $B$. This observation lies at the heart of algorithmic multi-to-single-agent reduction  proposed (but not implemented) by~\cite{cai2012algorithmic} and~\cite{alaei2019efficient}.\footnote{\ed{Comparing Figure~\ref{fig_optimal_primal} to plots obtained  by~\cite{dutting2019optimal} who did not rely multi-to-single-agent reduction, we see that even their advanced neural-network approach cannot overcome the curse of dimensionality and, as a result, is prone to smoothing artefacts.}} 
The reduction in the dimensionality comes at the cost of complexity of the feasibility constraint:
 the classic form of this constraint by~\cite{border1991implementation}
 leads to exponentially many inequalities and the two papers propose distinct ad~hoc constructions reducing this number to polynomial. 
 
 We rely on feasibility constraint in the  majorization form~\eqref{eq_Border} derived by \cite{hart2015implementation}. It is convex but non-linear. A natural linearization is suggested by the relation between majorization and martingales well-known to economists working on information design \citep{blackwell1951comp}.
 We use this relation in the following form:
 a measure  $\nu$ on $[0,1]$ majorizes $\nu'$ if and only if there is a distribution $\gamma$ on $[0,1]^2$ with marginals $\nu$ on $y$ and $\nu'$ on $x$ and such that $\int y{\dd}\gamma(y\mid x)\geq x$ for $\gamma$-almost all $x$, where $\gamma(y\mid x)$ denotes the conditional measure on $y$ given $x$ \cite[Theorem 4.A.5]{shaked2007stochastic}.\footnote{Equivalently, there is a supermartingale $(\xi,\eta)$ such that $\xi$ is distributed according to $\nu'$ and $\eta$, according to $\nu$.} 
 
 Considering $(u,\gamma)$ as unknowns, we obtain a linear optimization problem equivalent to~\eqref{eq_Rochet_Chone_extension}. Discretization of this problem leads to a number of constraints polynomial in~$n$. In Appendix~\ref{app_numerical}, \ed{relying on duality,}  we demonstrate that the values of the discretized problems are guaranteed to converge to the true value as the discretization becomes finer and finer.     
 Our approach is inspired by~\cite{ekeland2010algorithm} and, to the best of our knowledge, we are the first to obtain such approximation guarantees in multi-item auction design. To speed up the computation in practice, we adapt insights from~\cite{oberman2013numerical} to handle the incentive-compatibility constraint; see Appendix~\ref{app_numerical} for details.
 
\ed{The algorithm was implemented in Python using the LP solver from Gurobi library.
 Simulations were run on Amazon EC2 instance m6i.16xlarge  with 64 vCPUs with 3rd generation Intel Xeon Scalable cores and 256 GB of memory. For $n=200\times 200$ points in discretization, $I=2$ items and $B=2$ bidders, the computation required 83 minutes of real time and 20 hours of user time.}

\ifdefined\EC
\bibliographystyle{ACM-Reference-Format}
\bibliography{auctions}
\else
\bibliography{main}
\fi

\appendix

\section{Convex analysis basics}\label{sect_convex_analysis} Throughout the paper, we consider convex functions on $[0,1]$,  $X$, $\R$, or $\R^\mathcal{I}$ taking values in $\R\cup\{+\infty\}$. Here we briefly remind the reader some important facts and definitions.

The {subdifferential} of a convex function $f$ is defined by
\begin{equation}\label{eq_subdifferential}
\partial f(x) = \{\tau: f(y) \ge f(x) + \langle \tau,\, y - x \rangle, \ \ \ \forall y\}.
\end{equation}
Partial derivatives of $f(x)$ at $x\in \R^d$ (if exist) are denoted by 
$$f_{x_i}(x)=\frac{\partial f}{\partial x_i}(x).$$
The gradient $\nabla f(x)$ is the vector of partial derivatives
$$\nabla f(x)=\left(f_{x_i}(x)\right)_{i=1,\ldots,d}.$$
If $f$ is differentiable at $x$, then the subdifferential $\partial f(x)$ consists of just one element: $\partial f(x)=\{\nabla f(x)\}$. By the Alexandrov theorem, a convex function is twice differentiable except for a set of zero Lebesgue measure; see \cite[Theorem 14.1]{villani2009optimal}. In particular, the gradient $\nabla f(x)$ is defined almost everywhere and hence the integrals of the gradient with respect to an absolutely continuous measure are well-defined even for non-smooth $f$.

 The {Legendre transform} also known as Fenchel's conjugate of a convex function~$f$ is a convex function given by
 \begin{equation*}\label{eq_Fenchel_conjugate}
f^*(y) = \sup_x \Bigl(\langle x, y \rangle - f(x) \Bigr).
\end{equation*}
We will widely use the Fenchel inequality
\begin{equation}\label{eq_Fenchel_inequality}
f(x) + f^*(y) \ge \langle x, y \rangle
\end{equation}
and the corresponding ``complementary slackness'' condition  taking the following form: $f(x) + f^*(y) = \langle x, y \rangle
$ if and only if $y \in \partial f(x)$ and $x \in \partial f^*(y)$.

\section{Proof of Proposition~\ref{prop_Rochet}  (Rochet-Chon\'e representation of  auctioneer's problem)}\label{app_Rochet}

Recall that  the auctioneer's problem~\eqref{eq_revenue} is to maximize the revenue
\begin{equation}\label{eq_revenue_appendix}
\int_{X^{\mathcal{B}}} \left(\sum_{b\in\mathcal{B}} T_b\big((x_b)_{b\in\mathcal{B}}\big)\right)\cdot \left(\prod_{b\in\mathcal{B}}\rho(x_b)\right){\dd} x_1\cdots {\dd} x_B
\end{equation}
over individually-rational Bayesian incentive-compatible feasible mechanisms. The multi-bidder Rochet-Chon\'e problem~\eqref{eq_Rochet_Chone_extension} is to maximize 
\begin{equation}\label{eq_Rochet_appendix}
B\cdot\int_{X} \Big(\langle \nabla u(x), x\rangle -u(x)\Big) \rho(x) {\dd} x
\end{equation}
over convex non-decreasing functions  $u:\ X\to\R_+$  with $u(0)=0$ and such that, for all $i\in\mathcal{I}$,
\begin{equation}\label{eq_Border_appendix}
\frac{\partial u}{\partial x_{i}} (\chi)\preceq \xi^{B-1},
\end{equation}
where $\chi\in X$ is distributed with the density $\rho$ and $\xi$ is uniformly distributed on $[0,1]$. 
Our goal is to prove that the values of the two optimization problems coincide and both maxima are attained. The proof relies on a sequence of lemmas.

It will be convenient to work with a version of the Rochet-Chon\'e problem where the constraint $u(0)=0$ is relaxed (the requirements that $u$ is non-decreasing and takes only non-negative values remain).
\begin{lemma}\label{lm_relaxing_u0}
The constraint $u(0)=0$ in the Rochet-Chon\'e problem~\eqref{eq_Rochet_appendix} can be relaxed without affecting the value and whether the optimum is attained or not.
\end{lemma}
\begin{proof}
It is enough to show that for any feasible $u$ with $u(0)>0$, there is a feasible $\tilde{u}$ with $\tilde{u}(0)=0$ and the same or higher value of the objective. Defining $\tilde{u}(x)=u(x)-u(0)$ completes the proof.
\end{proof}
Let us demonstrate that, for any feasible solution to~\eqref{eq_revenue_appendix},  there is a feasible solution to~\eqref{eq_Rochet_appendix} with the relaxed constraint $u(0)=0$ and vice versa. This will imply that the two problems have the same values and, moreover, the optima are attained or not attained simultaneously.
\begin{lemma}\label{lm_constructing_u_from_PT}
For any  individually-rational Bayesian incentive-compatible feasible mechanism $(P,T)$ from~\eqref{eq_revenue_appendix}, there exists a function $u$ satisfying all the constraints of the Rochet-Chon\'e problem~\eqref{eq_Rochet_appendix} except for, possibly, $u(0)=0$ and such that the revenue of $(P,T)$ is equal to the value of~\eqref{eq_Rochet_appendix} at $u$.
\end{lemma}
\begin{proof}
 Consider the symmetrization of the mechanism $(P,T)$ over all permutations of bidders:
 \begin{align*}
    P_{b,i}^{\mathrm{sym}}\left((x_k)_{k\in\mathcal{B}}\right)&=\frac{1}{|N|!}\sum_{\sigma\in S_\mathcal{B}} P_{\sigma(b),i}\left((x_{\sigma(b)})_{b\in\mathcal{B}}\right)\\ T_{i}^{\mathrm{sym}}\left((x_k)_{k\in\mathcal{B}}\right)&=\frac{1}{|N|!}\sum_{\sigma\in S_\mathcal{B}} T_{\sigma(b)}\left((x_{\sigma(b)})_{b\in\mathcal{B}}\right), 
 \end{align*}
where $S_\mathcal{B}$ denotes the set of all permutations $\sigma$ of the set of bidders $\mathcal{B}$. The symmetrization $(P^{\mathrm{sym}},T^{\mathrm{sym}})$ results in the same revenue and inherits all the properties of $(P,T)$. By symmetry, all the bidders contribute equally to the revenue and so the revenue can be rewritten as 
$$B\cdot \int_{X^{\mathcal{B}}}  T_1^{\mathrm{sym}}\big((x_b)_{b\in\mathcal{B}}\big)\cdot \left(\prod_{b\in\mathcal{B}}\rho(x_b)\right){\dd} x_1\cdots {\dd} x_B= B\cdot \int_{X}  \overline{T_1^{\mathrm{sym}}}\big(x\big)\cdot \rho(x){\dd} x,$$
where $\big(\overline{P}_b^{\mathrm{sym}},\, \overline{T}_b^{\mathrm{sym}}\big)$ denotes bidder $b$'s reduced mechanism  (reduced mechanisms are the same for all the bidders by symmetry).
 Define $u(x)$ as the average utility of a bidder of type $x\in X$ in $(P^{\mathrm{sym}},T^{\mathrm{sym}})$:
$$u(x)= \Big\langle\overline{P_1^{\mathrm{sym}}}(x),\,x\Big\rangle-\overline{T_1^{\mathrm{sym}}}(x).$$
By the definition of incentive compatibility,
$$\big\langle\overline{P_b^{\mathrm{sym}}}(y),\,y\big\rangle-\overline{T_b^{\mathrm{sym}}}(y)\geq \langle\overline{P_b^{\mathrm{sym}}}(x),\,y\rangle-\overline{T_b^{\mathrm{sym}}}(x).$$
Thus
\begin{equation}\label{eq_linear_lower_bound}
u(y)\geq u(x)+\big\langle\overline{P_1^{\mathrm{sym}}}(x),\  y-x\big\rangle.
\end{equation}
We conclude that $u(y)=\max_{x\in X} \left( u(x)+\big\langle\overline{P_1^{\mathrm{sym}}}(x),\ y-x\big\rangle\right)$ and, hence, $u$ is a convex function as the pointwise maximum of a family of affine functions. Comparing~\eqref{eq_linear_lower_bound} to the definition of the subdifferential of a convex function~\eqref{eq_subdifferential}, we see that $\overline{P_1^{\mathrm{sym}}}(x)$ belongs to the subdifferential $\partial u(x)$. For Lebesgue-almost all $x$, the  gradient $\nabla f(x)$ of a convex function $f$ is well-defined and the subdifferential $\partial f(x)$ coincides with the singleton $\{\nabla f(x)\}$. Therefore,
\begin{equation}\label{eq_u_and_gradient-appendix}
\overline{P_1^{\mathrm{sym}}}(x)=\nabla u(x)
\end{equation}
for almost all $x$. By the definition of $u$, we can express $\overline{T_1^{\mathrm{sym}}}(x)$ as follows:
$$\overline{T_1^{\mathrm{sym}}}(x)=\Big\langle\overline{P_1^{\mathrm{sym}}}(x),\,x\Big\rangle-u(x)=\Big\langle\nabla u(x),\,x\Big\rangle-u(x),$$
where the second equality holds almost everywhere. Thus
$$B\cdot \int_{X}  \overline{T_1^{\mathrm{sym}}}\big(x\big)\cdot \rho(x){\dd} x=B\cdot\int_{X} \Big(\langle \nabla u(x), x\rangle -u(x)\Big) \rho(x) {\dd} x,$$
i.e., $u$ gives the same value to~\eqref{eq_Rochet_appendix} as $(P,T)$ to~\eqref{eq_revenue_appendix}. We already know that $u$ is convex. It remains  to check that $u$ is non-negative, monotone, and that it satisfies the majorization constraint~\eqref{eq_Border_appendix}. Non-negativity is immediate since, by the definition, $u\geq 0$ is equivalent to individual rationality of $(P^{\mathrm{sym}},T^{\mathrm{sym}})$. By~\eqref{eq_u_and_gradient-appendix}, $u$ is a convex function with the gradient  having non-negative components almost everywhere. Hence, $u$ is non-decreasing.

To check~\eqref{eq_Border_appendix}, note that $P^{\mathrm{sym}}$ can be seen as a family of ${I}$ allocation rules $P^{\mathrm{sym},\,i}=(P_{b,i})_{b\in\mathcal{B}}$, one for each item $i\in\mathcal{I}$. The reduced allocation $\overline{P_b^{\mathrm{sym},\,i}}\colon X\to \R_+$ for this one-item rule is equal to the corresponding component of $\overline{P_b^{\mathrm{sym}}}$.

\cite{hart2015implementation} showed that a function $f\colon  X\to [0,1]$ coincides with a reduced form $\overline{Q}_b$ of some bidder-symmetric feasible one-item mechanism $(Q,S)$ if and only if $f(\chi)\preceq_{} \xi^{B-1},$ where $\chi\in X$ is distributed with the density $\rho$ and $\xi$ is uniformly distributed on $[0,1]$. 

Applying this characterization, we conclude that  $$\overline{P_b^{\mathrm{sym},\,i}}(\chi)\preceq_{} \xi^{B-1}.$$
 Since $\overline{P_b^{\mathrm{sym},\,i}}$ is equal to $\overline{P_{b,i}^{\mathrm{sym}}}$ and the latter coincides with $\frac{\partial u}{\partial x_{i}}$ by~\eqref{eq_u_and_gradient-appendix}, we obtain the desired condition~\eqref{eq_Border_appendix}. To summarize, for any $(P,T)$, we constructed $u$ giving the same value to the  Rochet-Chon\'e problem and satisfying all its constraints (without $u(0)=0$ which was shown to be redundant).
 \end{proof}
 Now we show how to construct $(P,T)$ starting from~$u$.
 \begin{lemma}\label{lm_constructing_PT_from_u}
 For any $u$ satisfying the constraints of the Rochet-Chon\'e problem~\eqref{eq_Rochet_appendix} except for, possibly, $u(0)=0$, there exists an individually-rational Bayesian incentive-compatible feasible mechanism $(P,T)$ such that its revenue~\eqref{eq_revenue_appendix} is equal to the value of~\eqref{eq_Rochet_appendix} at $u$.
\end{lemma}
 \begin{proof}
The proof reverses the construction used to prove Lemma~\ref{lm_constructing_u_from_PT}.
Consider a function  $f^{i}$ equal to the component of $u$'s gradient corresponding to an item $i\in\mathcal{I}$, i.e.,   $f^{i}=\frac{\partial u}{\partial x_{i}}$. We assume that $f^{i}$ is defined for all $x\in X$: whenever the gradient is not well-defined, we select  $f^{i}$  arbitrarily so that the vector $f=(f^{i}(x))_{i\in\mathcal{I}}$ belongs to the subdifferential $\partial u(x)$.
The function $f^{i}$  is non-negative as $u$ is monotone and $f(\chi)$  is majorized by $\xi^{B-1}$ since  $\frac{\partial u}{\partial x_{i}}(\chi)$ is. Thus, by the theorem of \cite{hart2015implementation}, there exists a feasible one-item allocation $P^{i} \colon X^{\mathcal{B}}\to \R_+^{\mathcal{B}}$  such that $\overline{P_b^{i}}=f^{i}(x)$ for any bidder~$b$. 

Define the mechanism $(P,T)$ as follows.  The items are allocated by applying $P^{i}$ to each $i\in\mathcal{I}$, i.e., $P_{b,i}=P_b^{i}$. The transfers $T$ are given  by $$T_b((x_b)_{b\in\mathcal{B}})=\big\langle f(x_b), x_b\big\rangle -u(x_b).$$
Thus $(P,T)$ is feasible and the reduced mechanisms satisfy
\begin{equation}\label{eq_PT_from_u}
\overline{P}_b(x)=f(x)\quad \mbox{and}\quad  \overline{T}_b(x)=\big\langle f(x), x\big\rangle -u(x)
\end{equation}
for any bidder $b$. As $f=\nabla u$ almost everywhere, the second identity in~\eqref{eq_PT_from_u} implies that the revenue of $(P,T)$ coincides with the value of~\eqref{eq_Rochet_appendix} at $u$. It remains to check that $(P,T)$ is individually rational and Bayesian incentive-compatible. Individual rationality reads as  $\langle\overline{P}_b(x),x\rangle-\overline{T}_b(x)\geq 0$. By~\eqref{eq_PT_from_u}, the left-hand side equals $u(x)$ and so individual rationality follows from non-negativity of $u$. To show incentive-compatibility, recall that $f(x)$ is an element of the subdifferential of $u$ and so
$$u(x')\geq u(x)+\big\langle f(x),\, x'-x\big\rangle.$$
By~\eqref{eq_PT_from_u}, this inequality rewrites as
$$\big\langle\overline{P}_b(x_b),\,x_b\big\rangle-\overline{T}_b(x_b)\geq \langle\overline{P}_b(x_b'),\,x_b\rangle-\overline{T}_b(x_b'),$$
which is exactly the condition of incentive-compatibility for $(P,T)$. Thus $(P,T)$ is an individually-rational Bayesian incentive-compatible feasible mechanism with revenue equal to the  value of  the Rochet-Chon\'e objective at $u$.
\end{proof}
The above lemmas imply that the values of problems~\eqref{eq_revenue_appendix} and~\eqref{eq_Rochet_appendix} coincide. To prove that the optima are attained we need the following pair of lemmas.

Let $\mathcal{U}_{{\lip}, 1}$ be the set of all convex non-decreasing functions $u\colon X\to \R_+$ with $u(0)=0$ satisfying $1$-Lipschitz condition $|u(x)-u(y)|\leq \sum_{i\in\mathcal{I}}|x_{i}-y_{i}'|$ and endowed with the topology of the set of continuous functions.
\begin{lemma}\label{lm_Rochet_continuity}
The Rochet-Chon\'e objective
\begin{equation}\label{eq_Rochet_objective_appendix}
B\cdot\int_{X} \Big(\langle \nabla u(x), x\rangle -u(x)\Big) \rho(x) {\dd} x
\end{equation}
is a continuous functional over the set $\mathcal{U}_{{\lip}, 1}$.
\end{lemma}
\begin{proof}
Let $u^{(n)} \to u$ be a uniformly convergent sequence of functions from $\mathcal{U}_{{\lip}, 1}$. Any limiting point $y$ of any sequence $\{y^{(n)}\}$ such that $y^{(n)} \in \partial u^{(n)}(x)$, belongs to $\partial u(x)$. Indeed, for every $z$ one has
$$
u^{(n)}(z) \ge u^{(n)}(x) + \langle y^{(n)}, z - x \rangle
$$ 
by definition of the subdifferential. From the convergence $u^{(n)} \to u$ and $y^{(n)} \to y$ one gets
$$
u(z) \ge u(x) + \langle y, z - x \rangle
$$ 
for all $z$, hence, $y \in \partial u(x)$.
Since the subdifferential $\partial u^{(n)}(x)$ coincides with the gradient $\{ \nabla u^{(n)}(x)\}$ for all $n$ and almost all $x$, we get that
$\nabla u^{(n)}(x)$ converges to $\nabla u(x)$ almost everywhere. 

Thus the convergence of $u^{(n)} \to u$ in the topology of continuous functions implies the convergence of integrands in~\eqref{eq_Rochet_objective_appendix} almost everywhere. To deduce the continuity of the functional, we need to show that taking the limit commutes with the integration. This follows from the Lebesgue dominated convergence theorem. To apply this theorem, it remains to show that, in addition to convergence almost everywhere, the sequence of integrands is bounded. Since $u^{(n)}$ is a convergent sequence of continuous functions,  $\sup_{n\in \N ,x\in X}|u^{(n)}(x)|<\infty$  and $\langle \nabla u^{(n)}(x), x\rangle$ is bounded by ${I}$ thanks to the $1$-Lipschitz property. We obtain boundedness of the sequence of integrands and conclude that the Rochet-Chon\'e objective is continuous.
\end{proof}
The next lemma shows that the feasible set in the Rochet-Chon\'e problem is a compact subset of $\mathcal{U}_{{\lip}, 1}$.
\begin{lemma}\label{lm_compactness}
The set of convex non-decreasing functions  $u:\ X\to\R_+$ with $u(0)=0$ satisfying the majorization condition~\eqref{eq_Border_appendix} is a compact subset of $\mathcal{U}_{{\lip}, 1}$.
\end{lemma}
\begin{proof}
Since the upper bound in~\eqref{eq_Border_appendix} is a random variable taking values in $[0,1]$, we see that the gradient of a function $u$ from the statement of the lemma takes values in $[0,1]^{\mathcal{I}}$ and thus such $u$ belongs to $\mathcal{U}_{{\lip}, 1}$. 

To prove the compactness of the set of such $u$, note that the set $\mathcal{U}_{{\lip}, 1}$ is a
set of uniformly bounded uniformly equicontinuous functions. Hence, any sequence of functions from $\mathcal{U}_{{\lip}, 1}$ contains a convergent subsequence. Thus, to prove compactness of a subset of $\mathcal{U}_{{\lip}, 1}$, it is enough to check closedness of this subset.
If $u^{(n)}\in \mathcal{U}_{{\lip}, 1}$ is a sequence of functions converging uniformly to some $u$, we know that 
their gradients $\nabla u^{(n)}$ converge to  $\nabla u$ almost surely (see the proof of Lemma~\ref{lm_Rochet_continuity}). As the gradients are bounded, their distributions converge weakly. Therefore, if $u^{(n)}$ satisfy the majorization condition~\eqref{eq_Border_appendix}, it is also satisfied by the limit $u$. We obtain closedness and thus compactness.
\end{proof}
Now the proof of Proposition~\ref{prop_Rochet} is almost immediate.
\begin{proof}[Proof of Proposition~\ref{prop_Rochet}]
By Lemma~\ref{lm_relaxing_u0}, the value of the Rochet-Chon\'e problem does not change if we relax the constraint $u(0)=0$. Lemma~\ref{lm_constructing_u_from_PT} implies that the value of the Rochet-Chon\'e problem with the relaxed constraint is at least the value of the auctioneer's problem, while Lemma~\ref{lm_constructing_PT_from_u} gives the opposite inequality. Thus the values of the Rochet-Chon\'e and the auctioneer's problems are equal.

By Lemmas~\ref{lm_Rochet_continuity} and~\ref{lm_compactness}, the Rochet-Chon\'e problem can be seen as maximization of a continuous functional over a compact set. Therefore, this problem attains its optimum, i.e., the optimal $u$ exists. By Lemma~\ref{lm_constructing_PT_from_u}, we can find a mechanism $(P,T)$ such that the auctioneer's revenue is the same as the value of the Rochet-Chon\'e objective. Thus the optimum in the auctioneer's problem is also attained, i.e., the optimal auction exists as well.
\end{proof}

\section{Duality and proofs}\label{app_Dual}


In this section, we prove Theorems~\ref{th_vector_fields_inf} and Theorem~\ref{th_vector_fields_min} establishing the strong dual to the auctioneer's problem. The proof is split into two big parts. First, we derive a partial dual problem internalizing the feasibility constraint~\eqref{eq_Border_appendix} of \cite{hart2015implementation}. This problem is interpreted as a problem of a monopolist facing adversarial production costs; a result which may be of independent interest. In terms of this problem, we formulate a novel a priori bound on solutions, our main technical tool. Next, relying on this tool, we proceed with proving the theorems.

By Proposition~\ref{prop_Rochet}, we know that the auctioneer's problem is equivalent to the multi-bidder Rochet-Chon\'e problem where the distribution of $u$'s gradient is majorized by a particular distribution depending on the number of bidders. As our arguments do not depend on the exact form of the dominating distribution, in this section we allow for general majorizing distributions and, consequently, the results of this section extend Theorems~\ref{th_vector_fields_inf} and Theorem~\ref{th_vector_fields_min} to general majorization.

Let us describe the generalized Rochet-Chon\'e problem and introduce some useful notation along the way. Recall that $\mathcal{I}$ is the  set of $I\geq 1$ items and $X=[0,1]^\mathcal{I}$ is the set of bidders' types endowed with a density $\rho$. We will denote the corresponding distribution by $\mu$ so that
$${\dd}\mu(x)=\rho(x){\dd} x$$
and assume that $\rho$ is strictly positive on $X$.

For a convex function $u$ on $X$, its gradient is well-defined for almost all $x$; see Appendix~\ref{sect_convex_analysis}. For the gradient's component $\frac{\partial}{\partial x_i} u(x)$ we will sometimes use compact notation $u_{x_i}(x)$. 
We denote by $\nu_i$ the distribution of the gradient's $i$'th component $u_{x_i}(\chi)$ assuming that $\chi$ has distribution $\mu$. 

\smallskip 
\noindent\textbf{Rochet-Chon\'e problem with general majorization:}
\emph{given an absolutely continuous probability measure $\mu$ on $X$ and a collection of probability measures $(\eta_i)_{i\in \mathcal{I}}$ on $\R_+$,  maximize 
\begin{equation}\label{eq_Rochet_Chone_general_majorization}
\int_{X} \Big(\langle \nabla u(x), x\rangle -u(x)\Big) {\dd} \mu(x)
\end{equation}
over convex non-decreasing functions  $u\colon X\to\R_+$  with $u(0)=0$ and such that for all $i\in \mathcal{I}$
\begin{equation}\label{eq_Border_general_majorization}
\nu_i\preceq \eta_i,
\end{equation}
where $\nu_i$  is the distribution of $ u_{x_i}$.}

\smallskip

If all $\eta_i$ are the same and coincide with the distribution of $\xi^{B-1}$ with $\xi$ uniform on $[0,1]$, then the problem~\eqref{eq_Rochet_Chone_general_majorization} coincides with the multibidder Rochet-Chon\'e problem~\eqref{eq_Rochet_Chone_extension} up to a factor $B$ in the objective. By Proposition~\ref{prop_Rochet}, for such choice of $\eta_i$, the value of~\eqref{eq_Rochet_Chone_general_majorization} is equal to $\frac{1}{B}$ of the optimal revenue in the auctioneer's 
problem with $B$ bidders.

\subsection{Auctioneer's problem as monopolist's problem with adversarial production costs}\label{sec_adversary}
Consider a monopolist selling $I\geq 1$ items $i\in \mathcal{I}$ to one buyer whose type $x$ is distributed according to some measure $\mu$ on $X=[0,1]^\mathcal{I}$ with density $\rho$. In contrast to the single-bidder setting considered in Sections~\ref{sec_model} and~\ref{sec_Rochet}, these items have not yet been produced and so deciding on the amount to produce is a part of the monopolist's problem. We assume that the production costs are separable across items and, for each item $i\in \mathcal{I}$, are given by a convex non-decreasing function $\varphi_i$. The presence of the production costs $\sum_{i\in \mathcal{I}} \varphi_i\big(P_i\big)$ replace the feasibility constraint $P_i\leq 1$ of the monopolist's problem considered  in Section~\ref{subsect_monopolist}. That model corresponds to a particular case of $\varphi_i$ equal to $0$ on $[0,1]$ and $+\infty$ outside.  

\smallskip 
\noindent\textbf{Monopolist's problem with production costs.} \emph{For each item $i\in \mathcal{I}$, convex non-decreasing production costs $\varphi_{i}:\ [0,\infty)\to \R_+\cup\{+\infty\}$ with $\varphi_{i}(0)=0$ are given. The monopolist aims to maximize the total revenue consisting of the buyer's payment minus the production costs
\begin{equation}\label{eq_nonlinear_monopolist}
\Phi\big(u,(\varphi_i)_{i\in I} \big) = \int_{X} \Big(\langle \nabla u(x), x\rangle -u(x)\Big) {\dd} \mu(x)-\sum_{i\in \mathcal{I}}\int_{X}  \varphi_{i}\left(\frac{\partial u}{\partial x_{i}} (x)\right) {\dd} \mu(x)
\end{equation}
over convex non-decreasing functions $u\colon X\to\R_+$ with\footnote{One can show that this problem is equivalent to maximization of $\int_X(T(x)-\sum_i\varphi_i(P_i)){\dd} \mu(x)$  over  individually-rational Bayesian incentive-compatible mechanisms $(P,T):\ X\to \R_+^\mathcal{I}\times \R$ (the argument repeats the proof of Proposition~\ref{prop_Rochet}). We do not rely on this equivalence.} $u(0)=0$.}
\smallskip 

Let $\rev^\opt\big[(\varphi_i)_{i\in \mathcal{I}}\big]$ be the value of the problem~\eqref{eq_nonlinear_monopolist}, i.e., the maximal revenue the monopolist can achieve. Since the zero mechanism corresponding to $u\equiv 0$ is  feasible, the maximal revenue is non-negative, however, it may be infinite, e.g., if the costs are zero and so the monopolist has an incentive to increase production infinitely.

Consider an adversary who aims to minimize the monopolist's revenue by selecting the production costs but is 
penalized for choosing high costs. The adversary's objective is to minimize
\begin{equation}\label{eq_adversary}
\rev^\opt\big[(\varphi_i)_{i\in \mathcal{I}}\big]+\sum_{i\in \mathcal{I}}\int\varphi_i(z){\dd}\eta_i(z)
\end{equation}
for some given measures $\eta_i$.
\begin{theorem}\label{duality-theorem-maxinf}
Let $\eta_i$ be probability measures on $[0,1]$ such that $\eta_i([t, 1]) > 0$ for any~$t<1$.
Then the following assertions hold:
\begin{itemize}
        \item The value of the Rochet-Chon\'e problem with general majorization~\eqref{eq_Rochet_Chone_general_majorization}
       coincides with the optimal value achieved by the adversary in the minimization problem~\eqref{eq_adversary}. 
       
       In particular, if all $\eta_i$ are equal to the the distribution of $\xi^{B-1}$ with $\xi$ uniformly distributed on $[0,1]$, the value of~\eqref{eq_adversary} coincides with $\frac{1}{B}$ fraction of the auctioneer's optimal revenue~\eqref{eq_revenue} for $B$ bidders. 
       \item The optimum in~\eqref{eq_adversary} is attained, i.e., the adversary has an optimal strategy given by lower semicontinuous functions $(\varphi_{i}^\opt)_{i\in \mathcal{I}}$.
\end{itemize}
\end{theorem}
Let us formulate the result paying attention to functional classes to which $u$ and $\varphi_i$ belong. 

Denote by $ \mathcal{U}_{{\lip}, K}$  the set of non-decreasing convex functions $u$ on $X$ that have $u(0)=0$ and are $K$-Lipschitz  in the $l^{1}$-norm, i.e. $|u(x)-u(x')|\leq K\sum_i |x_i-x_i'|$. Note that monotonicity and $K$-Lipschitz properties together are equivalent to the following inequality on partial derivatives
 $$
 0 \le u_{x_i} \le K, \ \ \forall i \in \mathcal{I},
 $$
that must hold almost everywhere in $X$. For a probability measure $\eta$ on $[0,1]$, denote by 
$\mathcal{U}_{\eta,+\infty}$ the set of convex non-decreasing lower semicontinuous functions  $\varphi\colon \R_+\to \R_+\cup\{+\infty\}$ such that  $\varphi(0)=0$, the integral $\int_0^1 \varphi(z){\dd}\eta(z)<\infty$, and  $\varphi(z)=+\infty$ for $z>1$.

Formally, we prove the following identity
\begin{multline}\label{eq:max=min_max}
 \max_{u \in \mathcal{U}_{{\lip}, 1}, \nu_i \preceq_{} \eta_i  }  \int \left( \langle x, \nabla u \rangle - u(x) \right) d \mu= 
 \\
 =
  \min_{\varphi_i \in \mathcal{U}_{\eta_i, +\infty}}
 \max_{u \in \mathcal{U}_{{\lip}, 1}} \left[\Phi(u, \varphi_i) + \sum_{i\in\mathcal{I}}\int_0^1\varphi_i(x) \,{\dd}\eta_i(x)\right].
\end{multline}

\medskip 

The proof of Theorem~\ref{duality-theorem-maxinf} is contained in the next subsection. The high-level idea is to apply a functional minimax theorem to the Lagrangian internalizing the majorization constraint~\eqref{eq_Border_general_majorization}. 
Indeed, interpret $\Phi\Big(u,\,\big(\varphi_{i}\big)_{i\in\mathcal{I}}\Big)+\sum_{i\in \mathcal{I}}\int\varphi_i(z){\dd}\eta_i(z)$ as the payoff function in a zero-sum game. The maximizer selects $u$, while the minimizer picks $(\varphi_{i})_{i\in\mathcal{I}}$. 
The minimizer can infinitely penalize the maximizer for a violation of the majorization constraint~\eqref{eq_Border_general_majorization}. On the other hand, if the constraint is not violated, the best the minimizer can do is to select $\varphi_{i}\equiv 0$ on $[0,1]$ for all $i$ making the payoff equal to the objective of the Rochet-Chon\'e problem with general majorization~\eqref{eq_Rochet_Chone_general_majorization}. We conclude that the $\max\inf$-value of the game coincides with the value of the Rochet-Chon\'e problem~\eqref{eq_Rochet_Chone_general_majorization}. Similarly, one can show that $\inf\max$-value is the optimal value of the adversary's problem~\eqref{eq_adversary}. Next we apply the following functional minimax theorem 
which can be found in \cite[Theorem 2.4.1]{AdamsHedberg}.
 \begin{theorem}
 \label{thm:minimax}
 Let  $X, Y$  be convex subsets of linear topological spaces. We assume, in addition, that $X$ is a compact Hausdorff space. Let $f \colon X \times Y \to (-\infty,+\infty]$ 
 be a function  that is lower semicontinuous in $x$ for every  $y \in Y$, convex in $x$, and concave in  $y$. Then
 $$
 \min_{x \in X} \sup_{y \in Y} f(x,y) = \sup_{y \in Y} \min_{x \in X} f(x,y).
 $$
 \end{theorem}
By this theorem, we conclude that the $\max\inf$ and $\inf\max$ values coincide. This gives us the first item of Theorem~\ref{duality-theorem-maxinf}. We note that, in contrast to typical game-theoretic derivations of dual problems, the payoff function $\Phi$ is not affine in the strategy $u$ of the maximizer. However, $\Phi$ is convex in $u$ which is enough for Theorem~\ref{thm:minimax}.

This gives the result with infimum over $\varphi_i$ instead of minimum. Proving that the minimum is attained is the most difficult part of the proof as the set of minimizer's strategies is not compact and so we cannot use the standard compactness arguments.

\subsection{Proof of Theorem~\ref{duality-theorem-maxinf}}\label{subsec_partial_dual_proof}

In addition to $\mathcal{U}_{{\lip}, K}$  and 
$\mathcal{U}_{\eta,+\infty}$ defined above, we will need the following functional spaces:
\begin{itemize}
  \item
 $\mathcal{U}^p$, $p \ge 1$, is the set of non-decreasing convex functions $u\colon X\to \R_+$ with $u(0)=0$ and such that $\int_X|\nabla u|^p{\dd} \mu <\infty$, i.e., the gradient of $u$ belongs to~$L^p(\mu)$.
 \item 
 $\mathcal{U}_{\mathbb{R}_+}$ 
 is the set of convex non-decreasing lower semicontinuous functions  $\varphi\colon  \R_+ \to \R_+\cup\{+\infty\}$ with  $\varphi(0)=0$ such that there exists $t_0 \in \R_+$
 with 
 $$
 \varphi(t_0) > t_0.
 $$
Note that lower semicontinuity withing this class simply means that
 $$
 \lim_{s \to t -} \varphi(s) = \varphi(t),
 $$
 where $t = \inf\{s: \varphi(s) = +\infty\}$.
  \item $\mathcal{U}_{[0,1]}^{+\infty}$ is the set of all functions $\varphi\in \mathcal{U}_{\mathbb{R}_+}$ such that $\varphi_i(t)$ is finite for $t<1$ and equal to~$+\infty$ for $t>1$.
\end{itemize}

The following two simple lemmas provide compactness and continuity properties needed for the proof of Theorem~\ref{duality-theorem-maxinf}.

\begin{lemma} For any given $K \ge 0$,
the set $\mathcal{U}_{{\lip}, K}$ is compact in the uniform convergence topology.
\end{lemma}
\begin{proof}
The compactness follows from the Arzel\'a--Ascoli theorem and the obvious fact that convexity and monotonicity are preserved under uniform convergence.
\end{proof}

The following lemma extends the continuity of the objective obtained in Lemma~\ref{lm_Rochet_continuity} in the presence of functions $\varphi_i$.
\begin{lemma}
\label{uppercontin}
For any tuple of $\varphi_i \in \mathcal{U}_{\mathbb{R}_+}$
the functional
$\Phi(u, \varphi_i)$
is upper semicontinuous in $u$ in the uniform convergence topology on   $\mathcal{U}_{{\lip}, K}$ for every $K>0$. 

If, in addition,  $\varphi_i$ do not take value $+\infty$ and are continuous, then 
$\Phi(u, \varphi_i) $
is continuous in $u$ in the uniform convergence topology on   $\mathcal{U}_{{\lip}, K}$ for every $K>0$.
\end{lemma}
\begin{proof}
As it was demonstrated in the proof of Lemma~\ref{lm_Rochet_continuity}, 
if $u^{(n)} \to u$ is a uniformly convergent sequence of Lipschitz convex non-decreasing functions on ${X}$, then the gradients 
$\nabla u^{(n)}$ converge to   $\nabla u(x)$ almost everywhere. Thus, by the Fatou lemma (Theorem 11.20 in \cite{aliprantis2006infinite}), it is sufficient to check that
$$
\overline{\lim_{n}}  \left[\langle x, \nabla u^{(n)}(x) \rangle - u^{(n)}(x)  - \sum_{i\in\mathcal{I}} \varphi_{i}(u^{(n)}_{x_i}(x)) \right]
\le   \langle x, \nabla u(x) \rangle - u(x)  - \sum_{i\in\mathcal{I}} \varphi_{i}(u_{x_i}(x)) 
$$
$\mu$-a.e. and $\langle x, \nabla u^{(n)}(x) \rangle - u^{(n)}(x)  - \sum_{i\in\mathcal{I}} \varphi_{i}(u^{(n)}_{x_i}(x)) \le C$ for some $C$.
The first inequality follows immediately from $\mu$-a.e. convergence and lower semicontinuity of $\varphi_i$.
Next, since $u^{(n)}$, $u^{(n)}_{x_i}$ are nonnegative, one has
$$
 \langle x, \nabla u^{(n)}(x) \rangle - u^{(n)}(x)  - \sum_{i\in\mathcal{I}} \varphi_{i}(u^{(n)}_{x_i}(x)) \le \langle x, \nabla u^{(n)}(x)\rangle  \le {I}K.
$$

The second statement of the lemma follows from the Lebesgue dominated convergence theorem applied as in Lemma~\ref{lm_Rochet_continuity}.
\end{proof}

Now we are ready to prove Theorem~\ref{duality-theorem-maxinf}.
\begin{proof}[Proof of Theorem~\ref{duality-theorem-maxinf}.] \ \\
 \indent \textbf{Step 1.} 
The optimum in the left-hand side of~\eqref{eq:max=min_max} is attained.
Indeed, by Proposition~\ref{prop_Rochet}, we know that the value of the Rochet-Chon\'e problem attains its value and this problem coincides with the left-hand side of~\eqref{eq:max=min_max}. The proposition establishes the result for a particular choice of majorizing measures $\nu_i$ originating from the condition by \cite{hart2015implementation}, but the extension to arbitrary $\nu_i$ is straightforward.

\textbf{Step 2.} Let us rewrite our problem  in the minimax form and apply Theorem~\ref{thm:minimax}:
 \begin{align*}
&\max_{u \in \mathcal{U}^1, \nu_i \le \eta_i} \int \left(\langle x, \nabla u \rangle - u(x) \right)\,{\dd}\mu= \\
&= \max_{u \in \mathcal{U}_{{\lip}, 1}} \inf_{\varphi_i \in \mathcal{U}_{\eta_i, +\infty}}\left[\int \Big(\langle x, \nabla u \rangle - u(x)  - \sum_{i\in\mathcal{I}}\varphi_i(u_{x_i})\Big)\,{\dd}\mu + \sum_{i\in\mathcal{I}}\int_0^1\varphi_i(x) \,{\dd}\eta_i( x)\right]=\\
&=\inf_{\varphi_i \in \mathcal{U}_{\eta_i, +\infty}} \max_{u \in \mathcal{U}_{{\lip}, 1}} \left[\int \Big(\langle x, \nabla u \rangle - u(x)  - \sum_{i\in\mathcal{I}}\varphi_i(u_{x_i})\Big)\,{\dd}\mu + \sum_{i\in\mathcal{I}}\int_0^1\varphi_i(x) \,{\dd}\eta_i( x)\right]=\\
&=\inf_{\varphi_i \in \mathcal{U}_{\eta_i, +\infty}} \max_{u \in \mathcal{U}_{{\lip}, 1}} \left[\Phi(u, \varphi_i) + \sum_{i\in\mathcal{I}}\int_0^1\varphi_i(x) \,{\dd}\eta_i( x)\right].
\end{align*}
 
The first equality is obvious, while the second one follows from the minimax principle. Here we use compactness of  $\mathcal{U}_{{\lip},1}$, linearity in $\varphi_i$, concavity in $u$ (follows from convexity of $\varphi_i$), the upper semicontinuity was established in Theorem 2.

\textbf{Step 3.} We construct such a family of functions $(\varphi_i)$ that the infimum in the right-hand side of~\eqref{eq:max=min_max} is prospectively reached on them. Consider a sequence of tuples of functions $\{(\varphi_i^{(n)})_{i\in\mathcal{I}}\}_n \subset \mathcal{U}_{\eta_i, +\infty}$ such that
\begin{multline}\label{eq:lim_max_is_max}
\lim_{n \to \infty}\max_{u \in \mathcal{U}_{{\lip}, 1}} \left[\Phi(u, \varphi_i^{(n)}) + \sum_{i\in\mathcal{I}}\int_0^1\varphi_i^{(n)}(x) \,{\dd}\eta_i( x)\right]= \\
= \max_{u \in \mathcal{U}^1, \nu_i \preceq_{} \eta_i} \int (\langle x, \nabla u \rangle - u(x))\,{\dd}\mu.
\end{multline}

Denote by $M$ the optimal value of the objective function
\[
M = \max_{u \in \mathcal{U}^1, \nu_i \preceq_{} \eta_i} \int (\langle x, \nabla u \rangle - u(x))\,{\dd}\mu.
\]
We may assume that for all $n$ we have
\[
2M \ge \max_{u \in \mathcal{U}_{{\lip}, 1}} \left[\Phi(u, \varphi_i^{(n)}) + \sum_{i\in\mathcal{I}}\int_0^1\varphi_i^{(n)}(x) \,{\dd}\eta_i( x)\right].
\]
Since $\max_{u \in \mathcal{U}_{{\lip}, 1}} \Phi(u, \varphi_i^{(n)}) \ge 0$ and $\varphi_i^{(n)}(x) \ge 0$ for all $x$, we conclude that $\int_0^1 \varphi_i^{(n)}(x)\,{\dd}\eta_i \le 2M$ for all ${i\in\mathcal{I}}$ and for all $n$. All the functions $\varphi_i^{(n)}$ are non-negative and non-decreasing on $[0, 1]$; therefore, for every $t \in [0, 1)$,
\[
\int_0^1 \varphi_i^{(n)}(x)\,{\dd}\eta_i \ge \int_t^1 \varphi_i^{(n)}(x)\,{\dd}\eta_i \ge \varphi_i^{(n)}(t) \cdot \eta_i([t, 1]).
\]
Thus $\varphi_i^{(n)}(t) \le M_t$ for all $t \in [0, 1)$ and for all $n$, where $M_t = 2M / \eta_i([t, 1])$.

For every $t \in [0, 1)$, the sequence $\{\varphi_i^{(n)}\}$ is uniformly bounded on $[0, t]$ by the constant $M_t$. So, applying Helly's principle and passing to subsequences countably many number of times, we can assume that there exists a tuple of functions $(\varphi_i)_{i\in\mathcal{I}}$ defined on $[0, 1)$ such that $\varphi_i^{(n)} \to \varphi_i$ pointwise on $[0, 1)$. 

Each of the functions $\varphi_i$ is  non-negative, non-decreasing, and convex. In particular, $\lim_{t \to 1-} \varphi_i(t)$ is well defined. We extend the definition of $\varphi_i$ on $[0, +\infty)$ as follows: define $\varphi_i(1)$ as $\lim_{t \to 1-} \varphi_i(t) \in \mathbb{R} \cup \{+\infty\}$, and define $\varphi_i(x)$ at every $x \in (1, +\infty)$ to be equal to $+\infty$. The constructed function is lower-semicontinuous. Besides, $\lim_{n \to \infty} \varphi_i^{(n)}(x) = \varphi_i(x)$ for all $x \in [0, 1)$, and one can easily check that $\liminf_{n \to \infty} \varphi_i^{(n)}(1) \ge \varphi_i(1)$. 

We only need to check that $\varphi_i \in L^1(\eta_i)$ to prove that $\varphi_i \in \mathcal{U}_{\eta_i, +\infty}$. For every $i$, the sequence of integrals $\{\int_0^1 \varphi_i^{(n)}(x)\,{\dd}\eta_i\}_n$ is bounded. Passing to subsequences, we may additionally assume (and we will use it in the following part of the proof) that each of this sequences converges. By the Fatou lemma,
\begin{equation}\label{eq:int_of_phi_converges}
2M \ge \lim_{n \to \infty} \int_0^1 \varphi_i^{(n)}(x)\,{\dd}\eta_i \ge \int_0^1 \liminf_{n \to \infty}
\varphi_i^{(n)}(x)\,{\dd}\eta_i \ge \int_0^1 \varphi_i(x)\,{\dd}\eta_i.
\end{equation}
Hence, $\varphi_i \in L^1(\eta_i)$. Thus $\varphi_i \in \mathcal{U}_{\eta_i, +\infty}$.

\textbf{Step 4.}  We claim that for all $u \in \mathcal{U}_{{\lip}, 1}$ the following inequality holds:
\[
\Phi(u, \varphi_i) + \sum_{i\in\mathcal{I}}\int_0^1\varphi_i(x) \,{\dd}\eta_i( x) \le M.
\]

Fix a function $u \in \mathcal{U}_{{\lip}, 1}$. For any $\varepsilon \in (0, 1)$, consider a function $u^\varepsilon = (1 - \varepsilon) \cdot u \in \mathcal{U}_{{\lip}, 1 - \varepsilon}$. For all $x \in {X}$ and for all ${i\in\mathcal{I}}$, the value $u^\varepsilon_{x_i}(x)$ is not greater than $1 - \varepsilon$. So, for any ${i\in\mathcal{I}}$, the sequence of functions $\{\varphi_i^{(n)}(u^\varepsilon_{x_i}(x))\}_n$ converges to $\varphi_i(u^{\varepsilon}_{x_i}(x))$ pointwise almost everywhere. In addition, the inequality $0 \le \varphi_i^{(n)}(u^\varepsilon_{x_i}(x)) \le M_{1 - \varepsilon}$ holds for almost all $x \in {X}$ and for all $n$; therefore, it follows from Lebesgue's dominated convergence theorem that
\[
\lim_{n \to \infty}\int\varphi_i^{(n)}(u^{\varepsilon}_{x_i})\,{\dd}\mu = \int\varphi_i(u^{\varepsilon}_{x_i})\,{\dd}\mu.
\]
Combining this with the fact that $\lim_{n \to \infty} \int_0^1 \varphi_i^{(n)}(x)\,{\dd}\eta_i \ge  \int_0^1 \varphi_i(x)\,{\dd}\eta_i$, we conclude that
\[
\lim_{n \to \infty}\left[\Phi\left(u^\varepsilon, \varphi_i^{(n)}\right) + \sum_{i\in\mathcal{I}}\int_0^1\varphi_i^{(n)}(x) \,{\dd}\eta_i( x) \right] \ge \Phi(u^\varepsilon, \varphi_i) + \sum_{i\in\mathcal{I}}\int_0^1\varphi_i(x) \,{\dd}\eta_i( x).
\]
In particular,
\begin{align*}
M &= \lim_{n \to \infty}\max_{v \in \mathcal{U}_{{\lip}, 1}}\left[\Phi\left(v, \varphi_i^{(n)}\right) + \sum_{i\in\mathcal{I}}\int_0^1\varphi_i^{(n)}(x) \,{\dd}\eta_i( x)\right] \ge \\
&\ge \lim_{n \to \infty}\left[\Phi\left(u^\varepsilon, \varphi_i^{(n)}\right) + \sum_{i\in\mathcal{I}}\int_0^1\varphi_i^{(n)}(x) \,{\dd}\eta_i( x) \right] \ge \\
&\ge \Phi(u^\varepsilon, \varphi_i) + \sum_{i\in\mathcal{I}}\int_0^1\varphi_i(x) \,{\dd}\eta_i( x).
\end{align*}

Let $\varepsilon_n = \frac{1}{n}$. For every ${i\in\mathcal{I}}$, the sequence $\{\varphi_i(u_{x_i}^{\varepsilon_n}(x))\}_n$ is an increasing sequence of non-negative functions that converges to $\varphi_i(u_{x_i}(x))$ pointwise. So, by the Beppo Levi's lemma (Theorem 11.18 in \cite{aliprantis2006infinite}) we have
\[
\lim_{n \to \infty}\int \varphi_i(u_{x_i}^{\varepsilon_n}(x))\,{\dd}\mu( x) = \int \varphi_i(u_{x_i}(x))\,{\dd}\mu(x).
\]
Thus, for all $u \in \mathcal{U}_{{\lip}, 1}$, we have
\[
\Phi(u, \varphi_i) + \sum_{i\in\mathcal{I}}\int_0^1\varphi_i(x) \,{\dd}\eta_i(x) = \lim_{n \to \infty}\left(\Phi(u^{\varepsilon_n}, \varphi_i) + \sum_{i\in\mathcal{I}}\int_0^1\varphi_i(x) \,{\dd}\eta_i( x)\right) \le M.
\]

Since the last inequality holds for all $u \in \mathcal{U}_{{\lip}, 1}$, we conclude that
\[
\max_{u \in \mathcal{U}_{{\lip}, 1}}\left[\Phi(u, \varphi_i) + \sum_{i\in\mathcal{I}}\int_0^1\varphi_i(x) \,{\dd}\eta_i(x)\right] \le M.
\]
Thus it follows from the definition of $M$ that the equality holds and the minimum in~\eqref{eq:max=min_max} is reached on a family of functions sequence of functions $(\varphi_i)_{i\in\mathcal{I}}$.
\end{proof}



\subsection{Tools to approach complete duality: main a priori estimate and its corollaries}\label{sec_apriori}

In the next section, we discuss the complete duality results. The main insight allowing us to cope with non-compactness of the problem is an a priori bound on a solution. This bound has a clear economic interpretation in the context of the monopolist's problem with production~\eqref{eq_nonlinear_monopolist}. Here we discuss this bound and its implications.

Informally, the bound is as follows. It states that in the optimal mechanism, the monopolist never gets a negative revenue ex-post, i.e., 
$\langle \nabla u^\opt(x), x\rangle -u^\opt(x)-\sum_{i\in\mathcal{I}}\varphi_{i}( u_{x_i}^\opt(x))$ is non-negative.\footnote{Formulated in terms of monopolist's mechanism $(P^\opt,T^\opt)$, this inequality means 
$T^\opt(x)-\sum_{i\in\mathcal{I}}  \varphi_{i}\big(P_{i}^\opt(x)\big)\geq 0.$} This observation is not elementary as one could possibly expect that by  serving those costumers who bring negative profit, the monopolist could extract higher rent from the rest of the population.

We will rely on the notation for functional classes introduced in Section~\ref{subsec_partial_dual_proof}.
\begin{proposition}{\bf (Main a priori estimate).}
\label{ae-exist-lip}
Let $(\varphi_i)_{i \in I}$ be a collection of functions from $\mathcal{U}_{\mathbb{R}_+}$. Then for every function $u \in \mathcal{U}^1$, there exists a non-decreasing convex function $\tilde{u}$ with $\tilde{u}(0) = 0$ such that
\begin{equation}
\label{main-ae}
 \langle x, \nabla \tilde{u}(x) \rangle - \tilde{u}(x)  - \sum_{i\in\mathcal{I}} \varphi_{i}(\tilde{u}_{x_i}(x))  
 \ge \max\left\{ \langle x, \nabla u(x) \rangle - u(x)  - \sum_{i\in\mathcal{I}} \varphi_{i}(u_{x_i}(x)),\, 0\right\}
\end{equation}
for all $x \in X$. In particular, this implies (see Proposition~\ref{prop_Lipschitz_bound}) that for any function $u^{\mathrm{opt}} \in \mathcal{U}^1$ maximizing the functional $\Phi(u, \varphi_i)$ over $u\in \mathcal{U}^1$, the inequality
\begin{equation}
\label{mainapest}
 \langle x, \nabla {u^\opt}(x) \rangle - {u^\opt}(x)  - \sum_{i\in\mathcal{I}} \varphi_{i} \left(u_{x_i}^\opt (x)\right) \ge 0
\end{equation}
holds almost everywhere.
\end{proposition}
The main a priori estimate is used in  Proposition~\ref{prop_Lipschitz_bound} to show that, for a wide class of functions $\varphi_i$, the functional $\Phi(u,\varphi_i)$ attains its maximum on a Lipshitz function $u$. This fact will be used in approximation Lemma \ref{prop:approx_U_with_Q} which, together with  Proposition \ref{thm:max_is_min_in_Q}, help us to justify the minimax principle (Proposition \ref{thm:max_is_min_in_Q}) needed to prove complete duality (Theorem~\ref{th_vector_fields_inf_appendix}).
\begin{proof}[Proof of Proposition~\ref{ae-exist-lip}]
Consider the Legendre transform of $u$
$$
u^*(y) = \sup_{x} (\langle x, y \rangle - u(x)),
$$
assuming $u(x) = +\infty$  if $x \notin {X}$. Next we define
$$
v(y) = \max\left\{u^*(y), \sum_{i\in\mathcal{I}} \varphi_i(y_i)\right\}.
$$
Note that $v$ is a lower semicontinuous convex function and $v(0)=0$. Set  
$$
\tilde{u} = v^* = \left[\max\left\{u^*, \sum_{i\in\mathcal{I}} \varphi_i\right\}\right]^*.
$$
Then, by the Fenchel--Moreau theorem \citep{rockafellar2015convex},
$$(\tilde{u})^* = \max\left\{u^*, \sum_{i\in\mathcal{I}} \varphi_i\right\} $$ and, for every point $x$ where $\nabla \tilde{u}(x)$ exists, one has
\begin{equation}
\label{main-ae1}
\langle x, \nabla \tilde{u}(x) \rangle - \tilde{u}(x) - \sum_{i\in\mathcal{I}} \varphi_i(\tilde{u}_{x_i})
= \left( (\tilde{u})^*  -  \sum_{i\in\mathcal{I}} \varphi_i  \right)(\nabla \tilde{u}) \ge 0.
\end{equation}

Consider a point $x$, such that $\nabla u(x)$ exists and  satisfies $\langle x, \nabla u(x) \rangle - u(x)  - \sum_{i\in\mathcal{I}} \varphi_{i}(u_{x_i}(x)) \ge 0$. Equivalently,
$( u^* - \sum_{i\in\mathcal{I}} \varphi_i)(\nabla u(x)) \ge 0$. It follows from the theorem about the subdifferential of a maximum
of convex functions (Dubovitsky-Milyutin theorem; see Theorem 3.50 in \cite{beck2017first})
that $\partial v(y) = \partial \left[\max\left\{u^*, \sum_{i\in\mathcal{I}} \varphi_i\right\}\right](y)$ contains $\partial u^*(y)$ if $u^* \ge  \sum_{i\in\mathcal{I}} \varphi_i$.
Hence, if $x$ satisfies $( u^* - \sum_{i\in\mathcal{I}} \varphi_i)(\nabla u(x)) \ge 0$, then
$$
v(\nabla u(x)) = u^*(\nabla(x))\quad  { \rm and} \quad  x \in \partial v(\nabla u(x)).
$$
This implies
$v(\nabla u(x)) + v^*(x) = \langle \nabla u(x), x \rangle$, hence
$$
\tilde{u}(x) = v^*(x)  = \langle \nabla u(x), x \rangle - u^*(\nabla u(x)) = u(x)
$$
and from the inclusion $x \in \partial v(\nabla u(x))$ we get $\nabla u(x) \in \partial v^*(x) = \partial \tilde{u}(x)$.
In particular,  if $\nabla{\tilde{u}}$ exists, then $\nabla u(x) = \nabla \tilde{u}(x)$ and
\begin{equation}
\label{main-ae2}
\langle x, \nabla \tilde{u}(x) \rangle - \tilde{u}(x) - \sum_{i\in\mathcal{I}} \varphi_i(\tilde{u}_{x_i})
= \langle x, \nabla {u}(x) \rangle - {u}(x) - \sum_{i\in\mathcal{I}} \varphi_i({u}_{x_i}).
\end{equation}
The desired inequality  (\ref{main-ae}) follows from (\ref{main-ae1}) and (\ref{main-ae2}).
\end{proof}

With the help of the main a priori estimate, we obtain the following a priori bound on the regularity of the optimum.

\begin{proposition}\label{prop_Lipschitz_bound}
Fix a collection of functions $\varphi_i \in \mathcal{U}_{\mathbb{R}_+}$, $i\in \mathcal{I}$, and consider numbers $M_i$ such that
$$
\varphi_i(M_i) > M_i,
$$
which exist by the definition of the class $\mathcal{U}_{\mathbb{R}_+}$.
Then there exists a number $L$ depending on  $M_i$ and  $\varphi_i(M_i)$ such that $\Phi(u, \varphi_i) $
attains its maximum  on $\mathcal{U}^1$ at a function ${u^\opt}$ that belongs to $\mathcal{U}_{{\lip}, L}$.
\end{proposition}
\begin{proof}
By Proposition~\ref{ae-exist-lip}, any function $v \in \mathcal{U}^1$ can be replaced with a convex non-decreasing function $u$ with $u(0) = 0$ such that
\begin{align}
    &\Phi(u, \varphi_i) \ge \Phi(v, \varphi_i)\label{eq:phi_u_greater_phi_v}
\intertext{and}
&\langle x, \nabla u(x) \rangle - u(x) - \sum_{i\in \mathcal{I}} \varphi_i( u_{x_i}(x)) \ge 0 \label{eq:nonnegative_revenue}
\end{align}
for almost all  $x$. Moreover, if the function $v$ does not satisfy inequality~\eqref{eq:nonnegative_revenue}, then~\eqref{eq:phi_u_greater_phi_v} is strict. So, it is enough to check the existence of $L$ depending on $\varphi_i$, such that every Lipschitz function $u$ satisfying this inequality
is  $L$-Lipschitz. 

Indeed, since $u(x) \ge 0$, $u_{x_i}(x) \ge 0$ and $x_i \le 1$ for all ${i\in\mathcal{I}}$, 
assumption (\ref{eq:nonnegative_revenue}) implies:
\begin{align}
\begin{split}\label{eq:nonneg_rev_in_b_i}
u_{x_1} + u_{x_2} + \cdots + u_{x_{I}}   - &\sum_{i\in \mathcal{I}} \varphi_i(u_{x_i})\ge \\
\ge \langle x, \nabla u(x) \rangle - u(x) - &\sum_{i \in \mathcal{I}} \varphi_i(u_{x_i}) \ge 0.
\end{split}
\end{align}

For all $i\in \mathcal{I}$, consider function 
$\psi_i(x_i) =   x_i - \varphi_i(x_i)$. This function is concave and $0 = \psi_i(0) > \psi_i(M_i)$. Hence,  $\psi_i$ is decreasing on  $[M_i, +\infty)$
and its maximum is reached on $[0, M_i]$. Note that $\psi_i(x_i) \le x_i$ for all $x_i$, hence 
\[
\max_{x \ge 0} \psi_i(x_i) = \max_{0 \le x_i \le M_i} \psi_i(x_i) \le  M_i.
\]

Inequality (\ref{eq:nonneg_rev_in_b_i}) can be rewritten in the following form: 
\[
\sum_{i\in \mathcal{I}} \psi_i(u_{x_i}(x_i)) \ge 0
\]
for almost all  $x \in {X}$. Hence, for all ${i\in\mathcal{I}}$ and almost all  $x \in {X}$,
\begin{equation}\label{eq:psi_i_ineq}
\psi_i(u_{x_i}(x_i)) \ge -\sum_{j \ne i}\psi_j(u_{x_j}(x_j)) \ge - \sum_{j \ne i} M_j. 
\end{equation}

Concavity of  $\psi_i$ implies that all  $x_i \ge M_i$ satisfy inequality
\[
\frac{\psi_i(x_i) - \psi_i(M_i)}{x_i - M_i} \le \psi'_i(M_i) \le \frac{\psi_i(M_i) - \psi_i(0)}{M_i - 0} \quad \Leftrightarrow \quad \psi_i(x_i) \le \frac{x}{M_i}\psi_i(M_i).
\]
Hence, if
 \[
x_i > \max\left\{M_i, -\frac{\sum_{j \ne i} M_j \cdot M_i}{\psi_i(M_i)}\right\} = \widehat{M}_i,
\]
then  $\psi_i(x_i) < - \sum_{i\in\mathcal{I}} M_i$.

Hence, inequality (\ref{eq:psi_i_ineq}) implies that $u_{x_i}(x) \le \widehat{M}_i$ for almost all $x$. Thus $u$ is  $L$-Lipschitz with  $L = \max\left\{\widehat{M}_i\right\}$. 

It remains to show that $\Phi(u, \varphi_i)$  attains its maximum  on $\mathcal{U}_{{\lip},L}$.

According to estimate  (\ref{mainapest}) we can restrict ourselves to the set of functions  $u \in \mathcal{U}^1$ satisfying
\begin{align*}
\langle x, \nabla u(x) \rangle - u(x) - \sum_{i\in \mathcal{I}} \varphi_i(u_{x_i}(x)) \ge 0 
\end{align*}
for almost all  $x \in {X}$. We showed that all such functions  $u$ belong to $\mathcal{U}_{{\lip}, L}$.
This set is compact in uniform convergence topology and  $\Phi(u, \varphi_i)$ upper semicontinuous on $\mathcal{U}_{{\lip}, L}$. 
Hence, it reaches its maximum on this set.
\end{proof}

To prove complete duality, we will need the following weak form of  partial duality.    
 The goal is to represent  the value in the $\inf\max $ form so that we can  apply the miminax theorem and obtain the $\max \inf$ representation, which is done in  Proposition~\ref{thm:max_is_min_in_Q}.
 The subtlety is that, to apply the minimax theorem, compactness of one of the spaces is required and  so we need to choose carefully a dense minimization subspace $\mathcal{Q}$ in the set of convex one-dimensional functions.  In what follows,
\begin{equation}\label{eq_Q}
 \mathcal{Q} \subset \mathcal{U}_{\mathbb{R}_+}    
\end{equation}
  denotes the set of all increasing, convex functions   $\varphi\colon [0, +\infty)\to\R_+$ that equal zero at the origin and have bounded derivatives. 

\begin{proposition}\label{cor:max=max_inf}
Under the assumptions of Theorem~\ref{duality-theorem-maxinf}, the following partial duality equation holds:
\begin{equation*}
     \max_{u \in \mathcal{U}_{{\lip}, 1}, \nu_i \preceq_{} \eta_i  }  \int \left( \langle x, \nabla u \rangle - u(x) \right) d \mu  = \inf_{\varphi_i \in \mathcal{Q}} \max_{u \in \mathcal{U}^1} \left[\Phi(u, \varphi_i) + \sum_{i\in\mathcal{I}}\int_0^1\varphi_i(x) \,{\dd}\eta_i( x)\right].
\end{equation*}
\end{proposition}
The key part of the proof is the following lemma.
\begin{lemma}\label{prop:approx_U_with_Q}
For any family of functions $(\varphi_i)_{i\in\mathcal{I}} \subset \mathcal{U}_{[0, 1]}^{+\infty}$ there exist increasing sequences of functions $\{\varphi_i^{(n)}\}_n \subset \mathcal{Q}$, ${i\in\mathcal{I}}$,  such that each sequence $\{\varphi_i^{(n)}\}_n$ converges to $\varphi_i$ pointwise on $[0, 1]$ and that
\[
\max_{u \in \mathcal{U}_{{\lip}, 1}}\Phi(u, \varphi_i) = \lim_{n \to \infty}\max_{u \in \mathcal{U}^1}\Phi\left(u, \varphi_i^{(n)}\right).
\]
\end{lemma}

\begin{proof}[Proof of Lemma~\ref{prop:approx_U_with_Q}]
For every $n$, denote by $t_{n, i}$ such a point on the interval $[0, 1]$ that $n \in \partial \varphi_i(t_{n, i})$. Such a point exists since $(-\infty, +\infty) = \cup_{t \in [0, 1]} \partial \varphi_i(t)$. Denote by $\varphi_i^{(n)}$ the following function:
\[
\varphi_i^{(n)}(t) = \begin{cases}
\varphi_i(t), &\text{if } t \in [0, t_{n, i}],\\
\varphi_i(t_{n, i}) + n(t - t_{n, i}), &\text{otherwise.}
\end{cases}
\]

The function $\varphi_i^{(n)}$ is convex; therefore, $\varphi_i^{(n)} \in \mathcal{Q}$. Besides, $\varphi_i^{(n)}(x) \le \varphi_i(x)$ for $x \in [0, 1]$, and $\varphi_i^{(n)}$ coincides with $\varphi_i$ on the interval $[0, t_n]$. For each $i$, the sequence of points $\{t_{n, i}\}_n$ is monotonically increasing and converges to $1$; therefore, each sequence $\{\varphi_i^{(n)}\}_n$ is increasing and converges to $\varphi_i$ pointwise on $[0, 1)$. 
Finally, since pointwise supremum of lower semicontinuous functions is lower semicontinuous, we conclude that
\[
\lim_{n \to \infty}\varphi_i^{(n)}(1) = \lim_{t \to 1}\lim_{n \to \infty}\varphi_i^{(n)}(t) = \lim_{t \to 1}\varphi_i(t) = \varphi_i(1).
\]

Consider any function $v \in \mathcal{U}_{{\lip}, 1}$. For each ${i\in\mathcal{I}}$, the sequence of non-negative functions $\{\varphi_i^{(n)}(v_{x_i}(x))\}_n$ is monotonically increasing and converges to $\varphi_i(v_{x_i}(x))$ pointwise almost everywhere; therefore, by the monotone convergence theorem
\[
\lim_{n \to \infty}\int \varphi_i^{(n)}(v_{x_i})\,{\dd}\mu = \int \varphi_i(v_{x_i})\,{\dd}\mu.
\]
So,
\[
\Phi(v, \varphi_i)  = \lim_{n \to \infty}\Phi\left(v, \varphi_i^{(n)}\right) \le \lim_{n \to \infty}\max_{u \in \mathcal{U}^1}\Phi\left(u, \varphi_i^{(n)}\right).
\]
Since the last inequality holds for all $v \in \mathcal{U}_{{\lip}, 1}$, we conclude that
\begin{equation}\label{eq:max_le_lim_max}
    \max_{u \in \mathcal{U}_{{\lip}, 1}}\Phi(u, \varphi_i) \le \lim_{n \to \infty}\max_{u \in \mathcal{U}^1}\Phi\left(u, \varphi_i^{(n)}\right).
\end{equation}

Let $u^{(n)}$ be a maximizer of the functional $\Phi(\cdot,\, \varphi_i^{(n)})$. For every $n \ge 4{I}$, we  have $\varphi_i^{(n)}\left(1 + \frac{4{I}}{n}\right) \ge 4{I}$ and $1 + \frac{4{I}}{n} \le 2$. So, if we denote by $M_i$ the number $1 + \frac{4{I}}{n}$, by Proposition~\ref{prop_Lipschitz_bound}, we have
\[
u_{x_i} \le \max\left\{M_i,\, \frac{({I} - 1) \cdot M_i^2}{\varphi_i(M_i) - M_i}\right\} \le \max\left\{1 + \frac{4{I}}{n},\, \frac{4({I} - 1)}{4{I} - 2}\right\} = 1 + \frac{4{I}}{n}.
\]
So, $u^{(n)} \in \mathcal{U}_{{\lip}, 1 + 4{I} / n} \subset \mathcal{U}_{{\lip}, 2}$ for all $n \ge 4{I}$.

Passing to a subsequence, one can assume that $u^{(n)} \to u^{(0)}$ uniformly. Since $u_{x_i}^{(n)}(x) \to u_{x_i}^{(0)}(x)$ pointwise for almost all $x$, and $u_{x_i}^{(n)} \le 1 + \frac{4{I}}{n} \to 1$, we conclude that $\overline{u}^{(0)} \in \mathcal{U}_{{\lip}, 1}$. By the dominated convergence theorem,

\begin{equation}\label{eq:lim_int_un_is_int_u}
\lim_{n \to \infty} \int \left[\langle x, \nabla u^{(n)}(x) \rangle - u^{(n)}(x)\right]\,{\dd}\mu \to \int \left[\langle x, \nabla u^{(0)}(x) \rangle - u^{(0)}(x)\right]\,{\dd}\mu.
\end{equation}

We claim that for almost all $x \in {X}$ and for every ${i\in\mathcal{I}}$ we have
\[
\liminf_{n \to \infty}\varphi_i^{(n)}(u^{(n)}_{x_i}(x)) \ge \varphi_i(u^{(0)}_{x_i}(x)).
\]
Let $t_i = \sup\{t \colon \varphi_i(t) \le 2{I} + 1\}$. First, since every $\varphi_i^{(n)}$ is a non-decreasing function, \[
\varphi_i^{(n)}(u^{(n)}_{x_i}(x)) \ge \varphi_i^{(n)}\left(\min(u^{(n)}_{x_i}(x), t_i)\right).
\]

Next, let us check that
\[
\liminf_{n \to \infty} \varphi_i^{(n)}\left(\min(u^{(n)}_{x_i}(x), t_i)\right) \ge \liminf_{n \to \infty} \varphi_i\left(\min(u^{(n)}_{x_i}(x), t_i)\right).
\]
Indeed, if $t_i < 1$, then $t_{n, i} \ge t_i$ for all large enough $n$. Therefore, \[\varphi_i^{(n)}\left(\min(u^{(n)}_{x_i}(x), t_i)\right) =  \varphi_i\left(\min(u^{(n)}_{x_i}(x), t_i)\right)
\] for all large enough $n$. Otherwise, suppose that $t_i = 1$. Then $\varphi_i(1) \le 2{I} + 1 < +\infty$; therefore, by the lower semicontinuity of $\varphi_i$ for any $\varepsilon > 0$ there exists a point $p_i < 1$ such that $\varphi_i(p_i) \ge \varphi_i(1) - \varepsilon$. Then for all $n$ such that $t_{n, i} \ge p_i$ the inequality $\varphi_i^{(n)}(x) - \varphi_i(x) \ge -\varepsilon$ holds for all $x \in [0, 1]$. Indeed, if $x \le t_{n, i}$, then $\varphi_i^{(n)}(x) = \varphi_i(x)$. Otherwise, $x \ge t_{n, i} \ge p_i$; therefore,
\[
\varphi_i^{(n)}(x) \ge \varphi_i^{(n)}(p_i) = \varphi_i(p_i) \ge \varphi_i(x) - \varepsilon.
\]
Thus, in the case $t_i = 1$, the inequality
\[
\liminf_{n \to \infty} \varphi_i^{(n)}\left(\min(u^{(n)}_{x_i}(x), t_i)\right) \ge \liminf_{n \to \infty} \varphi_i\left(\min(u^{(n)}_{x_i}(x), t_i)\right) - \varepsilon
\]
holds for all $\varepsilon > 0$. Letting $\varepsilon$ tend to 0, we obtain the desired one.

Finally, we check that $u^{(0)}_{x_i}(x) = \lim_{n \to \infty} u_{x_i}^{(n)}(x) \le t_i$. If $t_i = 1$, the inequality holds since $u^{(0)} \in \mathcal{U}_{{\lip}, 1}$. Suppose that $t_i < 1$ and $u^{(0)}_{x_i}(x) > t_i$. Then for all large enough $n$ we have $u^{(n)}_{x_i}(x) \ge t_i$ and $t_{n, i} \ge t_i$. In this case,
\[
\varphi_i^{(n)}(u^{(n)}_{x_i}(x)) \ge \varphi_i^{(n)}(t_i) = \varphi_i(t_i) = 2{I} + 1.
\]
On the other hand, by Proposition~\ref{ae-exist-lip}, the inequality 
\[
\varphi_i^{(n)}(u^{(n)}_{x_i}(x)) \le \langle x, \nabla u^{(n)}(x)\rangle
\]
holds for almost all $x \in {X}$. For all $n \ge 4{I}$ we have $u^{(n)} \in \mathcal{U}_{{\lip}, 2}$; therefore, for almost all $x$ we have
\[
\varphi_i^{(n)}(u^{(n)}_{x_i}(x)) \le 2{I},
\]
which contradicts the previous inequality.

Thus
\begin{align*}
\liminf_{n \to \infty} \varphi_i^{(n)}\left(u^{(n)}_{x_i}(x)\right) &\ge \liminf_{n \to \infty} \varphi_i^{(n)}\left(\min(u^{(n)}_{x_i}(x), t_i)\right) \ge \\ 
&\ge \liminf_{n \to \infty} \varphi_i\left(\min(u^{(n)}_{x_i}(x), t_i)\right) =\\
&=\varphi_i\left(\liminf_{n \to \infty}\min(u^{(n)}_{x_i}(x), t_i)\right) = \varphi_i\left(u_{x_i}^{(0)}(x)\right).
\end{align*}
Therefore, it follows from Fatou's lemma that
\[
\liminf_{n \to \infty}\int \varphi_i^{(n)}\left(u^{(n)}_{x_i}(x)\right)\,{\dd}\mu \ge \int \liminf_{n \to \infty}\varphi_i^{(n)}\left(u^{(n)}_{x_i}(x)\right)\,{\dd}\mu \ge \int \varphi_i\left(u^{(0)}_{x_i}(x)\right)\,{\dd}\mu.
\]
So, combining it with~\eqref{eq:lim_int_un_is_int_u} we conclude that
\[
\Phi\left(u^{(0)}, \varphi_i\right) \ge \limsup_{n \to \infty}\Phi\left(u^{(n)}, \varphi_i^{(n)}\right) = \lim_{n \to \infty}\max_{u \in \mathcal{U}^1} \Phi\left(u, \varphi_i^{(n)}\right).
\]
Comparing it to~\eqref{eq:max_le_lim_max}, we conclude that the equality holds and this completes the proof of the statement.
\end{proof}
\begin{proof}[Proof of Proposition~\ref{cor:max=max_inf}]
 By the standard argument,
 \[
 \max_{u \in \mathcal{U}_{{\lip}, 1}, \nu_i \preceq_{} \eta_i  }  \int \left( \langle x, \nabla u \rangle - u(x) \right) d \mu \le \inf_{\varphi_i \in \mathcal{Q}} \max_{u \in \mathcal{U}^1} \left[\Phi(u, \varphi_i) + \sum_{i\in\mathcal{I}}\int_0^1\varphi_i(x) \,{\dd}\eta_i( x)\right].
 \]
 
 Let $(\varphi_i)_{i\in\mathcal{I}} \subset \mathcal{U}_{\eta_i, +\infty}$ be family of functions on which the minimum in the right-hand side of~\eqref{eq:max=min_max} is reached. By Lemma~\ref{prop:approx_U_with_Q}, there exist increasing sequences $\{\varphi_i^{(n)}\}_n \subset \mathcal{Q}$, $1 \le i \le {I}$, such that the sequence $\{\varphi_i^{(n)}\}_n$ converges to $\varphi_i$ pointwise on $[0, 1]$ and that
\[
\max_{u \in \mathcal{U}_{{\lip}, 1}}\Phi(u, \varphi_i) = \lim_{n \to \infty}\max_{u \in \mathcal{U}^1}\Phi\left(u, \varphi_i^{(n)}\right).
\]
By the monotone convergence theorem, we additionally have $\lim_{n \to \infty}\int \varphi_i^{(n)}\,{\dd}\eta_i = \int \varphi_i\,{\dd}\eta_i$. Thus
\begin{align*}
\max_{u \in \mathcal{U}_{{\lip}, 1}, \nu_i \preceq_{} \eta_i  }  &\int \left( \langle x, \nabla u \rangle - u(x) \right) d \mu= \\
&= \max_{u \in \mathcal{U}_{{\lip}, 1}} \left[\Phi(u, \varphi_i) + \sum_{i\in\mathcal{I}}\int_0^1\varphi_i(x) \,{\dd}\eta_i( x)\right]= \\
&= \lim_{n \to \infty}\max_{u \in \mathcal{U}_{{\lip}, 1}} \left[\Phi(u, \varphi_i^{(n)}) + \sum_{i\in\mathcal{I}}\int_0^1\varphi_i^{(n)}(x) \,{\dd}\eta_i( x)\right].
\end{align*}
\end{proof}

\subsection{Complete duality} 

Relying on the partial duality established in Sections~\ref{sec_adversary}
and the a priori estimate from Section~\ref{sec_apriori}, we are ready to prove complete duality for
the monopolist's problem with general majorization~\eqref{eq_Rochet_Chone_general_majorization}
extending Theorems~\ref{th_vector_fields_inf}
and~\ref{th_vector_fields_min}.

We will rely on notation introduced in Section~\ref{subsec_partial_dual_proof}. Denote by
   $
  \mathcal{C}
 $
 the set of smooth nonnegative (coordinate-wise) vector fields  $c =(c_i)_{i\in \mathcal{I}}$ such that
\begin{equation}\label{eq_C_definition}
\int (\langle x, \nabla u(x) \rangle - u(x)){\dd}\mu \le \int \langle c(x), \nabla u(x) \rangle {\dd}\mu
\end{equation}
 for all  $u \in \mathcal{U}^1$. The condition above is equivalent to the majorization constraint $-{\div}_{\rho}[c] \succeq m$, where ${\dd}\mu(x) = \rho(x)\,{\dd} x$ and $m$ is the transform measure. Note, in particular, that~$x \in \mathcal{C}$.
 \begin{theorem}\label{th_vector_fields_inf_appendix}
Let $\eta_i$ be probability measures on $[0,1]$ such that $\eta_i([t, 1]) > 0$ for any~$t>0$. Then the value\footnote{Recall that this value is defined by $\max_{u \in \mathcal{U}_{{\lip}, 1}, \nu_i \preceq_{} \eta_i  }  \int \left( \langle x, \nabla u \rangle - u(x) \right) d \mu$.} of the Rochet-Chon\'e problem with general majorization~(\ref{eq_Rochet_Chone_general_majorization}) is equal to
\begin{equation}\label{eq_duality_inf_appendix}
 \inf_{\varphi_i \in \mathcal{Q},\,c \in \mathcal{C}}\left(\sum_{i\in\mathcal{I}}\int \varphi_i^*(c_i)\,{\dd}\mu + \sum_{i\in\mathcal{I}}\int_0^1\varphi_i(x)\,{\dd}\eta_i( x)\right).
\end{equation}
\end{theorem}
 Let us check that Theorem~\ref{th_vector_fields_inf} is a corollary of Theorem~\ref{th_vector_fields_inf_appendix}.  
 \begin{proof}[Proof of Theorem~\ref{th_vector_fields_inf}]
 Fix all $\eta_i$ to coincide with the distribution of $\xi^{B-1}$ where $\xi$ is uniform on $[0,1]$. Then the value of the Rochet-Chon\'e problem with general majorization is $\frac{1}{B}$ fraction of the value of the corresponding auctioneer's problem with $B$ bidders (Proposition~\ref{prop_Rochet}).
 
 In Section~\ref{sec_duality}, we already demonstrated weak duality for the auctioneer's problem, i.e., we checked that $\frac{1}{B}$ fraction of the auctioneer's revenue cannot exceed
 \begin{equation}\label{eq_Beckmann_duality_rhs_appendix}
 \inf_{
  \footnotesize{\begin{array}{c}
       (\varphi_{i})_{i\in\mathcal{I}},\\  
       \pi\succeq_{} m
  \end{array}}}\left[
  \B_\rho\Big(\pi,\,\Phi\Big)+ \sum_{i\in\mathcal{I}}\int_0^1\varphi_{i}\left(z^{B-1}\right){\dd} z\right],
  \end{equation}
  where $\Phi(c)=\sum_i \varphi_i^*(|c_i|)$. Hence, to prove that this expression coincides with $\frac{1}{B}$ fraction of the optimal revenue, it is enough to demonstrate that it is bounded from above by~\eqref{eq_duality_inf_appendix}.
  
 Comparing~\eqref{eq_C_definition} to the definition of divergence~\eqref{eq_divergence_by_parts_function} and that of the transform measure~\eqref{eq_transform_measure}, we see that $\mathcal{C}$ consists of vector fields $c$ with non-negative components such that $-\div_\rho[c]\succeq m$ or, equivalently, there exists $\pi\succeq m$ such that $\div_\rho[c]+\pi=0$. Let $\mathcal{C}_\pm$ be the superset of $\mathcal{C}$ obtained by dropping the non-negativity condition. We get
$$\inf_{c \in \mathcal{C}_\pm}\sum_{i\in\mathcal{I}}\int \varphi_i^*(|c_i|)\,{\dd}\mu=
  \B_\rho\Big(\pi,\,\Phi\Big).$$
  Since $\mathcal{C}\subset \mathcal{C}_\pm$, we conclude that
  $$\inf_{c \in \mathcal{C}}\sum_{i\in\mathcal{I}}\int \varphi_i^*(c_i)\,{\dd}\mu=\inf_{c \in \mathcal{C}}\sum_{i\in\mathcal{I}}\int \varphi_i^*(|c_i|)\,{\dd}\mu\geq \inf_{
  \footnotesize{\begin{array}{c}
       (\varphi_{i})_{i\in\mathcal{I}},\\  
       \pi\succeq_{} m
  \end{array}}}
  \B_\rho\Big(\pi,\,\Phi\Big).$$
 Hence,~\eqref{eq_duality_inf_appendix} is an upper bound on~\eqref{eq_Beckmann_duality_rhs_appendix}. Thus~\eqref{eq_Beckmann_duality_rhs_appendix} is equal to $\frac{1}{B}$ fraction of the  auctioneer's optimal revenue.
 \end{proof}
 
 As a preliminary step to proving Theorem~\ref{th_vector_fields_inf_appendix}, we prove a complete duality result for the monopolist's problem with fixed production costs.  Denote by $\overline{\mathcal{C}}$ the set of bounded nonnegative vector fields  $c =(c_i)_{i\in \mathcal{I}}$, not necessary smooth, such that
\begin{equation*}
\int (\langle x, \nabla u(x) \rangle - u(x)){\dd}\mu \le \int \langle c(x), \nabla u(x) \rangle {\dd}\mu
\end{equation*}
 for all  $u \in \mathcal{U}^1$. Note that $\mathcal{C} \subset \overline{\mathcal{C}}$. Recall that $\mathcal{Q}$ is defined in~\eqref{eq_Q}.
\begin{proposition}\label{thm:max_is_min_in_Q}
For any family of functions $(\varphi_i)_{i\in\mathcal{I}} \subset \mathcal{Q}$,  the following relation holds
 \[
 \max_{u \in \mathcal{U}^1}  \Phi(u, \varphi_i) 
 =  \min_{c \in \overline{\mathcal{C}}}  \int \sum_{i\in\mathcal{I}} \varphi^*_{i}(c_i) {\dd}\mu.
 \]
Moreover, if all the functions $\varphi_i$ are continuously differentiable, then the vector field $c_i = \varphi_i'(\overline{u}_{x_i})$ solves the dual problem, where $\overline{u}$ is an optimal solution to the problem $\max_{u \in \mathcal{U}^1}  \Phi(u, \varphi_i)$.
 \end{proposition}
 
\begin{proof}
For every   $\varphi_i \in \mathcal{Q}$, one has
 $$
 \langle x, \nabla u \rangle - u(x)  - \sum_{i\in\mathcal{I}} \varphi_{i}(u_{x_i})\
 = \min_{c(x) \ge 0}
 \Bigl(  \langle x - c(x), \nabla u \rangle - u(x) +  \sum_{i\in\mathcal{I}} \varphi^*_{i}(c_i) \Bigr).
 $$
 The minimum is taken among of all nonnegative vector fields and it is attained at $c_i = \varphi'_i(u_{x_i})$.
In particular, $0 \le c_i \le \sup \varphi'_i$ (we apply here that  $\varphi_i \in \mathcal{Q}$, hence the derivatives $\varphi'_i$ are uniformly bounded).
 
 Thus for every $\varphi_i \in \mathcal{Q}$ we get 
  $$
 \langle x, \nabla u \rangle - u(x)  - \sum_{i\in\mathcal{I}} \varphi_{i}(u_{x_i})\
 = \min_{B(\varphi_i)}
 \Bigl(  \langle x - c(x), \nabla u \rangle - u(x) +  \sum_{i\in\mathcal{I}} \varphi^*_{i}(c_i) \Bigr),
 $$
where  $B(\varphi_i)$ is the set of non-negative vector fields  $c = (c_i)$ satisfying  $c_i(x) \le \sup \varphi'_i(x)$ for $\mu$-a.e. $x$.
Hence,
  \begin{align*}
  \max_{u \in \mathcal{U}^1}  \Phi(u, \varphi_i) &   =    \sup_{u \in \mathcal{U}^1} \min_{B(\varphi_i)} 
 \Bigl(\int \bigl(   \langle x - c(x), \nabla u \rangle - u(x) +  \sum_{i\in\mathcal{I}} \varphi^*_{i}(c_i) \bigr){\dd}\mu \Bigr),
 \end{align*}
 We apply the minimax principle and the fact that $B(\varphi_i)$ is a closed subset of a ball in  $L^{\infty}(\mu)$, endowed with the *-weak topology.    The Banach--Alaoglu theorem implies that $B(\varphi_i)$ is compact. Hence,
 \begin{align*}
 &
  \sup_{u \in \mathcal{U}^1} \min_{c \in B(\varphi_i)} 
 \Bigl(\int \bigl(   \langle x - c(x), \nabla u \rangle - u(x) +  \sum_{i\in\mathcal{I}} \varphi^*_{i}(c_i) \bigr){\dd}\mu \Bigr)=
 \\& =  \min_{c \in B(\varphi_i)}   \sup_{u \in \mathcal{U}^1}
 \Bigl(\int \bigl(   \langle x - c(x), \nabla u \rangle - u(x) +  \sum_{i\in\mathcal{I}} \varphi^*_{i}(c_i) \bigr){\dd}\mu \Bigr)
 \end{align*}
 Let us check that the minimax principle is applicable.
 Indeed, the convexity of the functional on  $B(\varphi_i)$ is obvious, it is sufficient to check the lower semicontinuity.
 Let us consider a sequence $c^{(n)} \in B(\varphi_i)$ such that $c^{(n)} \to c$ *-weakly in  $L^{\infty}(\mu)$ (in particular, weakly in $L^2(\mu)$).
 It is sufficient to show that  $\underline{\lim}_n \int \varphi^*_{i}(c^{(n)}_i){\dd}\mu \ge \int  \varphi^*_{i}(c_i){\dd}\mu$.
 
 Passing to  a subsequence (if necessary), which we denote again by  $c^{(n)}_i$, one can assume without loss of generality that
  $ \int \varphi^*_{i}(c^{(n)}_i){\dd}\mu $ has a limit and the sequence of  $\frac{1}{N} \sum_{n=1}^{N} c^{(n)}_i$ converges in $L^2(\mu)$
  and $\mu$-a.e.
 Applying convexity of $\varphi_i$, one gets
 \begin{align*}
 \lim_n \int \varphi^*_{i}(c^{(n)}_i){\dd}\mu &= 
 \lim_N  \frac{1}{N} \sum_{n=1}^{N} \int \varphi^*_{i}(c^{(n)}_i){\dd}\mu\ge
 \\
 &\ge \lim_N \int \varphi^*_{i}\Bigl(  \frac{1}{N}  \sum_{n=1}^{N} c^{(n)}_i \Bigr){\dd}\mu
\ge  \int  \varphi^*_{i}(c_i){\dd}\mu.
 \end{align*}
 In the last inequality we use convergence almost everywhere and the Fatou lemma.
 
The next step is obvious:
 \begin{multline*}
  \min_{c \in B(\varphi_i)}   \sup_{u \in \mathcal{U}^1}
 \left(\int \bigl(   \langle x - c(x), \nabla u \rangle - u(x) +  \sum_{i\in\mathcal{I}} \varphi^*_{i}(c_i) \bigr){\dd}\mu \right)=
 \\
 =   \min_{c \in  B(\varphi_i) \cap \overline{\mathcal{C}}}  
\sum_{i\in\mathcal{I}}\int    \varphi^*_{i}(c_i) {\dd}\mu.
 \end{multline*}
 Hence,
 \[
  \max_{u \in \mathcal{U}^1}  \Phi(u, \varphi_i)   =   \min_{c \in  B(\varphi_i) \cap \overline{\mathcal{C}}}  
 \Bigl(\sum_{i\in\mathcal{I}}\int    \varphi^*_{i}(c_i) {\dd}\mu \Bigr).
 \]
Clearly, $\min_{c \in  B(\varphi_i)\cap \overline{\mathcal{C}}}  
 \sum_{i\in\mathcal{I}}\int    \varphi^*_{i}(c_i) {\dd}\mu $ can be replaced with $\min_{c \in \overline{\mathcal{C}}}  
 \sum_{i\in\mathcal{I}}\int    \varphi^*_{i}(c_i) {\dd}\mu $, since, by the standard arguments,  
 $ \max_{u \in \mathcal{U}^1}  \Phi(u, \varphi_i) \le \min_{c \in \overline{\mathcal{C}}} \sum_{i\in\mathcal{I}}\int    \varphi^*_{i}(c_i) {\dd}\mu $.
 
 Now, assume that all the functions $\varphi_i$ are continuously differentiable. Let $c^{(0)}$ be an optimal solution to the dual problem $\min_{c \in\overline{\mathcal{C}}} \int \sum_{i\in\mathcal{I}} \varphi^*_{i}(c_i) {\dd}\mu$, and let $\overline{u}$ be an optimal solution to the problem $\max_{u \in \mathcal{U}^1}  \Phi(u, \varphi_i)$. The following sequence of inequalities holds:
 \begin{align*}
 &\int \left(\langle \nabla \overline{u}(x), x \rangle  - \overline{u}(x)  - \sum_{i \in I}\varphi_i(\overline{u}_{x_i}(x))\right)\,{\dd}\mu \le \\
 &\le \int \left(\langle \nabla \overline{u}(x), x \rangle  - \overline{u}(x)  - \sum_{i \in I}\overline{u}_{x_i}(x) \cdot c^{(0)}_i(x)\right)\,{\dd}\mu + \int \sum_{i \in I} \varphi^*\left(c^{(0)}_i\right)\,{\dd}\mu\le \\
 &\le\int \sum_{i \in I} \varphi^*\left(c^{(0)}_i\right)\,{\dd}\mu.
 \end{align*}
 The left-hand and right-hand sides of this inequality are equal. Therefore,
 \begin{align*}
 &\int \left(\langle \nabla \overline{u}(x), x \rangle  - \overline{u}(x)  - \sum_{i \in I}\overline{u}_{x_i}(x) \cdot c^{(0)}_i(x)\right)\,{\dd}\mu = 0\\
 \intertext{and}
 &\sum_{i \in I}\int \left( \varphi_i(\overline{u}_{x_i}) + \varphi^*_i\left(c^{(0)}_i\right) - \overline{u}_{x_i} \cdot c^{(0)}_i\right)\,{\dd}\mu = 0.
 \end{align*}
 Thus  $c^{(0)}_i(x) \in \partial \varphi_i(\overline{u}_{x_i}(x))$  for $\mu$-almost all $x$.
 Since the functions $\varphi_i$ are continuously differentiable, we conclude that $\partial \varphi_i(x_0) = \{\varphi'(x_0)\}$ for all $x_0 > 0$ and $\partial \varphi_i(x_0) = (-\infty, \varphi'(x_0))$ at the point $x_0 = 0$. Thus $c^{(0)}_i(x) = \varphi_i'(\overline{u}_{x_i}(x))$ for $\mu$-almost all such points $x$ that $\overline{u}_{x_i}(x) > 0$.
 
 Consider a vector field defined by the equation $\overline{c}_i(x) = \varphi_i'(\overline{u}_{x_i}(x))$ for all $x$. Since $\overline{c}_i(x) \ge c^{(0)}_i(x)$ for $\mu$-almost all $x$, we conclude easily that $\overline{c} \in \overline{\mathcal{C}}$. Since $\overline{c}_i(x) \in \partial \varphi_i(\overline{u}_{x_i}(x))$ for $\mu$-almost all $x$,
 \[
 \sum_{i \in I}\int \left( \varphi_i(\overline{u}_{x_i}) + \varphi^*_i(\overline{c}_i) - \overline{u}_{x_i} \cdot \overline{c}_i\right)\,{\dd}\mu = 0.
 \]
 In addition, $\overline{u}_{x_i} \cdot c^{(0)}_i = \overline{u}_{x_i} \cdot \overline{c}_i$ for $\mu$-almost all $x$; therefore,
 \begin{multline*}
 \int \left(\langle \nabla \overline{u}(x), x \rangle  - \overline{u}(x)  - \sum_{i \in I}\overline{u}_{x_i}(x) \cdot \overline{c}_i(x)\right)\,{\dd}\mu= \\
 = \int \left(\langle \nabla \overline{u}(x), x \rangle  - \overline{u}(x)  - \sum_{i \in I}\overline{u}_{x_i}(x) \cdot c^{(0)}_i(x)\right)\,{\dd}\mu = 0.
 \end{multline*}
 Thus we finally conclude that
 \begin{align*}
 \int \sum_{i \in I} \varphi^*(\overline{c}_i)\,{\dd}\mu &=\int\left(\langle \nabla \overline{u}(x), x \rangle  - \overline{u}(x)  - \sum_{i \in I}\varphi_i(\overline{u}_{x_i})\right)\,{\dd}\mu+\\
  &+ \sum_{i \in I}\int\left(\varphi_i(\overline{u}_{x_i}) + \varphi^*_i(\overline{c}_i) - \overline{u}_{x_i} \cdot \overline{c}_i\right)\,{\dd}\mu- \\
  &- \int \left(\langle \nabla \overline{u}(x), x \rangle  - \overline{u}(x)  - \sum_{i \in I}\overline{u}_{x_i}(x) \cdot \overline{c}_i(x)\right)\,{\dd}\mu=\\
  &=\int\left(\langle \nabla \overline{u}(x), x \rangle  - \overline{u}(x)  - \sum_{i \in I}\varphi_i(\overline{u}_{x_i})\right)\,{\dd}\mu.
 \end{align*}
 This equality means that $\overline{c} \in \argmin_{c \in \overline{\mathcal{C}}} \int \sum_{i\in\mathcal{I}} \varphi^*_{i}(c_i) {\dd}\mu$.
\end{proof}

Next, we extend the previous result to smooth vector fields.
 \begin{proposition}\label{prop:max_is_inf_over_smooth}
 For any family of strictly convex continuously differentiable functions $(\varphi_i)_{i\in\mathcal{I}} \subset \mathcal{Q}$, the following relation holds
  \[
 \max_{u \in \mathcal{U}^1}  \Phi(u, \varphi_i) 
 =  \inf_{c \in \mathcal{C}}  \int \sum_{i\in\mathcal{I}} \varphi^*_{i}(c_i) {\dd}\mu.
\]
 \end{proposition}
 \begin{proof}
 Let $\overline{u} \in \argmax_{u \in \mathcal{U}^1}  \Phi(u, \varphi_i)$, and let $(\overline{c}_i)_{i \in I}$ be the vector field defined by the formula $\overline{c}_i(x) = \varphi_i'(\overline{u}_{x_i}(x))$. By Proposition~\ref{thm:max_is_min_in_Q}, the vector field $\overline{c}$ is an optimal solution to the dual problem 
 \[
 \max_{c \in \overline{\mathcal{C}}}  \int \sum_{i\in\mathcal{I}} \varphi^*_{i}(c_i) {\dd}\mu.
 \]
 
 We claim that the function $\overline{u}_{x_i}$ is continuous at all the points $x_0$ such that $\overline{u}$ is differentiable at $x_0$. Indeed, consider any sequence of points $\{x_n\}_n$ converging to $x_0$. One can verify easily that if $p_n \in \partial \overline{u}(x_n)$ and if $\partial \overline{u}(x_0)$ contains only one point $\nabla \overline{u}(x)$, then the sequence $\{p_n\}_n$ converges to $\nabla \overline{u}(x_0)$. Hence, the sequence $\{\overline{u}_{x_i}(x_n)\}_n$ converges to $\overline{u}_{x_i}(x_0)$, and this implies the continuity of $\overline{u}_{x_i}$ at the point~$x_0$.
 
 In particular, this means that the function $\overline{c}_i = \varphi_i'(\overline{u}_{x_i})$ is continuous almost everywhere; therefore, the function
 \[
 \overline{c}_i^{\mathrm{up}}(x_0) = \limsup_{x \to x_0} \overline{c}_i(x)
 \]
 is upper semi-continuous and $\overline{c}_i =  \overline{c}_i^{\mathrm{up}}$ almost everywhere. Thus the vector field $\overline{c}^{\mathrm{up}} = \left( \overline{c}_i^{\mathrm{up}}\right)_{i \in I}$ belongs to $\overline{\mathcal{C}}$ and is an optimal solution to the dual problem.
 
Denote $L_i = \sup_x \varphi'(\overline{u}_{x_i}(x)) = \sup_x \overline{c}_i^{\mathrm{up}}(x)$. Since the function $\overline{u}_{x_i}$ is bounded by Proposition~\ref{prop_Lipschitz_bound} and the function $\varphi_i'$ is strictly increasing, we have $L_i < \sup_x \varphi_i'(x) = \widehat{L}_i$. Since the function $\overline{c}_i^{\mathrm{up}}$ is upper semi-continuous, it can be written  as the pointwise limit of a non-increasing family $\{c_i^{(n)}\}$ of smooth functions; moreover, we can require that $\sup_x c_i^{(n)} \le \left(L_i + \widehat{L}_i\right) / 2$ for all $n$.

Since $c_i^{(n)} \ge \overline{c}_i^{\mathrm{up}}$, we clearly have $c^{(n)} = \left(c_i^{(n)}\right)_{i \in I} \in \mathcal{C}$. Let us check that
\begin{equation}\label{eq:smooth_approx_integrals_converges}
\lim_{n \to \infty} \int \sum_{i\in\mathcal{I}} \varphi^*_{i}\big(c_i^{(n)}\big){\dd}\mu = \int \sum_{i\in\mathcal{I}} \varphi^*_{i}\big(\overline{c}^{\mathrm{up}}_i\big) {\dd}\mu.
\end{equation}
Indeed, the function $\varphi_i^*$ is continuous and non-decreasing on the interval $\big[0, (L_i + \widehat{L}_i) / 2\big]$; therefore, the sequence of functions $\varphi_i^*\big(c_i^{(n)}\big)$ is a non-increasing family that converges to $\varphi_i^*\big(\overline{c}_i^{\mathrm{up}}\big)$ pointwise. Then Beppo Levi's lemma implies~\eqref{eq:smooth_approx_integrals_converges}.

This implies that
\[
\inf_{c \in \mathcal{C}}  \int \sum_{i\in\mathcal{I}} \varphi^*_{i}(c_i) {\dd}\mu \le \lim_{n \to \infty} \int \sum_{i\in\mathcal{I}} \varphi^*_{i}\big(c_i^{(n)}\big){\dd}\mu = \int \sum_{i\in\mathcal{I}} \varphi^*_{i}\big(\overline{c}^{\mathrm{up}}_i\big) {\dd}\mu = \max_{u \in \mathcal{U}^1}  \Phi(u, \varphi_i).
\]
On the other hand, since $\mathcal{C} \subset \overline{\mathcal{C}}$, we  have
\[
\inf_{c \in \mathcal{C}}\int \sum_{i\in\mathcal{I}} \varphi^*_{i}(c_i) {\dd}\mu \ge \min_{c \in \overline{\mathcal{C}}}\int \sum_{i\in\mathcal{I}} \varphi^*_{i}(c_i) {\dd}\mu = \max_{u \in \mathcal{U}^1}  \Phi(u, \varphi_i).
\]
This implies the desired duality relation.
 \end{proof}
 
Now, we can prove Theorem~\ref{th_vector_fields_inf_appendix}.
 \begin{proof}[Proof of Theorem~\ref{th_vector_fields_inf_appendix}.]
 Denote by $\mathcal{Q}_{\mathrm{st}}$ the subset of functions $\varphi \in \mathcal{Q}$ such that $\varphi$ is continuously differentiable and strictly convex. We claim that
 \begin{multline}\label{eq:inf_Q_is_inf_Q_st}
 \inf_{\varphi_i \in \mathcal{Q}}\max_{u \in \mathcal{U}^1} \left[\Phi(u, \varphi_i) + \sum_{i\in\mathcal{I}}\int_0^1\varphi_i(x) \,{\dd}\eta_i( x)\right]= \\
 =  \inf_{\varphi_i \in \mathcal{Q_{\mathrm{st}}}}\max_{u \in \mathcal{U}^1} \left[\Phi(u, \varphi_i) + \sum_{i\in\mathcal{I}}\int_0^1\varphi_i(x) \,{\dd}\eta_i( x)\right].
 \end{multline}
 Since $\mathcal{Q}_{\mathrm{st}} \subset \mathcal{Q}$, we conclude that the left-hand side is not greater than the right-hand side. Let us check the opposite inequality.
 
 Consider any function $\varphi \in \mathcal{Q}$. Let $\{p^{(n)}(x)\}$ be a sequence of smooth non-negative kernel functions such that each function $p^{(n)}(x)$ is supported on the interval $[-1/n, 0]$ and $\int_{-\infty}^{+\infty}p^{(n)}(x)\,{\dd} x = 1$. Consider the function \[
 \varphi^{(n)}(x) = \left(\varphi * p^{(n)}\right)(x) = \int_0^{+\infty} \varphi(t) p^{(n)}(x-t)\,{\dd} t = \int_{0}^{\frac{1}{n}}\varphi(x + t) p^{(n)}(-t)\,{\dd} t.
 \]
 One can easily check that the function $ \varphi^{(n)}(x)$ is smooth, convex, non-negative, and non-decreasing for $x \ge 0$. Moreover, if $\varphi'(x) \le L$ for all $x$, then $\big(\varphi^{(n)}\big)'(x) \le L$ for all $x$, so $\varphi^{(n)}(x) - \varphi^{(n)}(0) \in \mathcal{Q}$. Finally, $\varphi(x) \le \varphi^{(n)}(x) \le \varphi\big(x + \frac{1}{n}\big) \le \varphi(x) + \frac{L}{n}$ for all $x \ge 0$. Denoting  \[
 \widehat{\varphi}^{(n)}(x)  = \left(\varphi^{(n)}(x) - \varphi^{(n)}(0)\right) + \frac{1}{n}\left(x + \exp(-x) - 1\right),
 \]
 we conclude that $\widehat{\varphi}^{(n)} \in \mathcal{Q}_{\mathrm{st}}$, that $\widehat{\varphi}^{(n)}(x) \ge \varphi(x) - \varphi^{(n)}(0) \ge \varphi(x) - \frac{L}{n}$ for all $x \ge 0$, and that \[
 \widehat{\varphi}^{(n)}(x) - \varphi(x) \le \frac{L + 1}{n}\quad\text{for all $x \in [0, 1]$.}
 \]
 
 Consider any family of functions $(\varphi_i)_{i \in I} \subset \mathcal{Q}$. For each $i \in I$, let $\{\varphi_i^{(n)}\}_n \subset \mathcal{Q}_{\mathrm{st}}$ be a sequence of functions such that $\varphi_i^{(n)}(x)  - \varphi_i(x) \ge -\frac{1}{n}$ for all $x \ge 0$ and that $\varphi_i^{(n)}(x) - \varphi_i(x) \le \frac{1}{n}$ for all $x \in [0, 1]$.
 Since $\varphi_i^{(n)}(x) \ge \varphi_i(x) - \frac{1}{n}$ for all $x \ge 0$, we have
 \[
     \Phi(u, \varphi^{(n)}_i) \le \Phi(u, \varphi_i) + \frac{I}{n}\quad\Rightarrow\quad\max_{u \in \mathcal{U}^1}  \Phi(u, \varphi^{(n)}_i) \le  \max_{u \in \mathcal{U}^1}  \Phi\left(u, \varphi_i\right) + \frac{I}{n}.
 \]
 In addition,
 \[
 \int_0^1 \varphi_i^{(n)}(x)\,{\dd}\eta_i(x) \le  \int_0^1 \varphi_i(x)\,{\dd}\eta_i(x) + \frac{1}{n}.
 \]
 Thus
 \begin{multline*}
 \max_{u \in \mathcal{U}^1} \left[\Phi\left(u, \varphi^{(n)}_i\right) + \sum_{i\in\mathcal{I}}\int_0^1\varphi^{(n)}_i(x) \,{\dd}\eta_i( x)\right]\le \\
 \le \max_{u \in \mathcal{U}^1} \left[\Phi(u, \varphi_i) + \sum_{i\in\mathcal{I}}\int_0^1\varphi_i(x) \,{\dd}\eta_i( x)\right] + \frac{2I}{n}.
 \end{multline*}
 
 In particular, for any family of functions $(\varphi_i)_{i \in I} \subset \mathcal{Q}$ we have
 \begin{align*}
     \inf_{\varphi_i \in \mathcal{Q_{\mathrm{st}}}}&\max_{u \in \mathcal{U}^1} \left[\Phi(u, \varphi_i) + \sum_{i\in\mathcal{I}}\int_0^1\varphi_i(x) \,{\dd}\eta_i( x)\right]\le \\
     &\le\liminf_{n \to \infty}\max_{u \in \mathcal{U}^1} \left[\Phi\left(u, \varphi^{(n)}_i\right) + \sum_{i\in\mathcal{I}}\int_0^1\varphi^{(n)}_i(x) \,{\dd}\eta_i( x)\right] \le\\
     &\le \max_{u \in \mathcal{U}^1} \left[\Phi(u, \varphi_i) + \sum_{i\in\mathcal{I}}\int_0^1\varphi_i(x) \,{\dd}\eta_i( x)\right].
 \end{align*}
 This implies the equality (\ref{eq:inf_Q_is_inf_Q_st}).
 
 Now, the desired result follows from the combination of Propositions~\ref{cor:max=max_inf} and~\ref{prop:max_is_inf_over_smooth}.
 \end{proof}
 \medskip

Next, we prove the duality theorem in the strong form $\max = \min$. To do it, we need an extension of the set of feasible vector fields $\mathcal{C}$.
We denote by $\mathcal{C}^{mes}$ the set of tuples of non-negative measures $(\varsigma_i)_{i\in\mathcal{I}}$ satisfying
\[
\int \left(\langle \nabla u(x), x \rangle  - u(x) \right)\,{\dd}\mu \le \sum_{i\in\mathcal{I}} \int u_{x_i}\,{\dd}\varsigma_i
\]
for every smooth $u \in \mathcal{U}^1$.

 \begin{theorem}\label{th_duality_min_appendix}
Under the assumptions of Theorem~\ref{th_vector_fields_inf_appendix}, the following identity holds:
\begin{multline*}
\max_{u \in \mathcal{U}^1, \nu_i \preceq_{} \eta_i} \int (\langle x, \nabla u \rangle - u(x))\,{\dd}\mu= \\
= \min_{\substack{\varsigma \in \mathcal{C}^{mes},\\
\varphi_i \in \mathcal{U}_{\eta_i, +\infty}}} \sum_{i\in\mathcal{I}} \left(\varsigma_i^{\mathrm{sing}}(X) + \int_0^1 \varphi_i(x)\,{\dd}\eta_i + \int \varphi_i^*(\varsigma_i^a(x))\,{\dd}\mu\right).
\end{multline*}
\end{theorem}
Note that Theorem~\ref{th_vector_fields_min} is a particular case of Theorem~\ref{th_duality_min_appendix} for all $\eta_i$ equal to the distribution of $\xi^{B-1}$ with $\xi$ uniform on $[0,1]$.
\smallskip

To prove Theorem~\ref{th_duality_min_appendix} we need several auxiliary results.

\begin{lemma}
The set $\mathcal{C}^{mes}$ is closed in the weak*-topology.
\end{lemma}
\begin{proof} Trivial, as we can additionally require the test function $u$ to be smooth.
\end{proof}
\begin{lemma}\label{prop:c_def_works_for_u_lip}
Let $(\varsigma_i)_{i\in\mathcal{I}} \in \mathcal{C}^{mes}$, and let $\varsigma_i = \varsigma_i^a(x)\,{\dd}\mu( x) + \varsigma_i^{\mathrm{sing}}$ be a decomposition of the component $\varsigma_i$ into an absolutely continuous and a singular part w.r.t. $\mu$. Then for any $u \in \mathcal{U}_{{\lip}, 1}$, the following inequality holds:
\[
\int (\langle \nabla u(x), x \rangle  - u(x))\,{\dd}\mu \le \sum_{i\in\mathcal{I}} \left(\varsigma_i^{\mathrm{sing}}(X) + \int u_{x_i}(x) \cdot \varsigma_i^a(x)\,{\dd}\mu\right).
\]
\end{lemma}
\begin{proof}
Let $\overline{u}$ be any non-negative convex function defined on the whole $\mathbb{R}^{\mathcal{I}}$ such that $0 \le \overline{u}_{x_i}(x) \le 1$ for all $x \in \mathbb{R}^{\mathcal{I}}$ and that $\overline{u}|_{{X}} = u$. {It can be defined, for instance, in the following way: 
$\overline{u}(x) = \sup_{\alpha} l_{\alpha}$, where $\{l_{\alpha}$\} is the set of  affine functions satisfying $l_{\alpha}|_{X} \le u$.}

Let $\{p_n\}$ be a sequence of Gaussian kernels converging to $\delta(0)$, and denote by $\overline{u}^{(n)}$ the convolution $\overline{u} * p_n$. One can easily check that $u^{(n)}(x)$ is a smooth non-negative convex function such that $0 \le \overline{u}^{(n)}_{x_i}(x) \le 1$ for all $x \in \mathbb{R}^{\mathcal{I}}$. Moreover, the sequence $\{\overline{u}^{(n)}\}$ converges uniformly to $u$ on ${X}$. Thus, denoting by $u^{(n)}(x)$ the function $\overline{u}^{(n)}(x) - \overline{u}^{(n)}(0)$, we conclude that $u^{(n)}$ is smooth, $u^{(n)} \in \mathcal{U}_{{\lip}, 1}$, and the sequence $\{u^{(n)}\}$ converges uniformly to $u$ on ${X}$.

Since $u^{(n)}_{x_i}(x) \le 1$ for all $x \in {X}$, the following inequality holds:
\begin{align}
\begin{split}\label{eq:ineq_on_un_and_c_tmp}
\int \big(\langle \nabla u^{(n)}(x), x \rangle  &- u^{(n)}(x) \big)\,{\dd}\mu \le\\
&\le \sum_{i\in\mathcal{I}} \left(\int u^{(n)}_{x_i}(x) \cdot \varsigma_i^a(x)\,{\dd}\mu + \int u^{(n)}_{x_i}(x)\,{\dd}\varsigma_i^{\mathrm{sing}}\right)\le\\
&\le \sum_{i\in\mathcal{I}} \left(\int u^{(n)}_{x_i}(x) \cdot \varsigma_i^a(x)\,{\dd}\mu + \varsigma_i^{\mathrm{sing}}(X)\right).
\end{split}
\end{align}
Since $\{u^{(n)}\}$ converges to $u$ uniformly on ${X}$, the sequence $\{\nabla u^{(n)}(x)\}$ converges to $\nabla u(x)$ for $\mu$-almost all $x$. Therefore, by the Lebesgue's dominated convergence theorem 
\begin{align*}
&\lim_{n \to \infty}\int (\langle \nabla u^{(n)}(x), x \rangle  - u^{(n)}(x))\,{\dd}\mu = \int (\langle \nabla u(x), x \rangle  - u(x))\,{\dd}\mu, \\
&\lim_{n \to \infty}\int u^{(n)}_{x_i}(x) \cdot \varsigma_i^a(x)\,{\dd}\mu = \int u_{x_i}(x) \cdot \varsigma_i^a(x)\,{\dd}\mu.
\end{align*}
Thus, passing to the limits in~\eqref{eq:ineq_on_un_and_c_tmp}, we obtain the desired inequality.
\end{proof}

The following proposition extends the complete duality result for the monopolist's problem with fixed costs (Proposition~\ref{thm:max_is_min_in_Q}) so that the minimum in the dual is attained.
\begin{proposition}\label{thm:monopolist_strong_duality}
For any given family of functions $(\varphi_i)_{i\in\mathcal{I}} \subset \mathcal{U}_{[0, 1]}^{+\infty}$ and an absolutely continuous measure $\mu$ on ${X}$, the following duality relation holds:
\begin{equation*}
    \max_{u \in \mathcal{U}_{{\lip}, 1}}\Phi(u, \varphi_i) = \min_{\varsigma \in \mathcal{C}^{mes}} \sum_{i\in\mathcal{I}} \left(\varsigma_i^{\mathrm{sing}}(X) + \int \varphi_i^*(\varsigma_i^a(x))\,{\dd}\mu\right),
\end{equation*}
where $\varsigma_i = \varsigma_i^a(x)\,{\dd}\mu + \varsigma_i^{\mathrm{sing}}$ is a decomposition of the component $\varsigma_i$ into an absolutely continuous  and a singular part w.r.t. $\mu$.
\end{proposition}
\begin{proof}
By Lemma~\ref{prop:c_def_works_for_u_lip}, for any $u \in \mathcal{U}_{{\lip}, 1}$ and $\varsigma \in \mathcal{C}^{mes}$, we have
\[
\int (\langle \nabla u(x), x \rangle  - u(x))\,{\dd}\mu \le \sum_{i\in\mathcal{I}} \left(\varsigma_i^{\mathrm{sing}}(X) + \int u_{x_i}(x) \cdot \varsigma_i^a(x)\,{\dd}\mu\right).
\]
Therefore,
\begin{align*}
\Phi(u, \varphi_i)  &\le \sum_{i\in\mathcal{I}}\left(\varsigma_i^{\mathrm{sing}}(X) +  \int\left[u_{x_i}(x) \cdot \varsigma_i^a(x) - \varphi_i(u_{x_i}(x))\right]\,{\dd}\mu\right) \le \\
&\le \sum_{i\in\mathcal{I}} \left(\varsigma_i^{\mathrm{sing}}(X) + \int \varphi_i^*(\varsigma_i^a(x))\,{\dd}\mu\right),
\end{align*}
where the last part of the inequality follows from the inequality $\varphi_i(u_{x_i}(x)) + \varphi_i^*(\varsigma_i^a(x)) \ge u_{x_i}(x) \cdot \varsigma_i^a(x)$, which holds for all $x$. Thus we conclude that
\begin{equation}\label{eq:max_u_le_min_c}
\max_{u \in \mathcal{U}_{{\lip}, 1}}\Phi(u, \varphi_i) \le \min_{\varsigma \in \mathcal{C}^{mes}} \sum_{i\in\mathcal{I}} \left(\varsigma_i^{\mathrm{sing}}(X) + \int \varphi_i^*(\varsigma_i^a(x))\,{\dd}\mu\right).
\end{equation}

By Lemma~\ref{prop:approx_U_with_Q}, there exist increasing sequences $\{\varphi_i^{(n)}\}_n \subset \mathcal{Q}$, $1 \le i \le {I}$ that converge to $\varphi_i$ pointwise on $[0, 1]$ and that 
\begin{equation}\label{eq:max_Phi=lim_max_Phi_n}
\max_{u \in \mathcal{U}_{{\lip}, 1}}\Phi(u, \varphi_i) = \lim_{n \to \infty}\max_{u \in \mathcal{U}^1}\Phi\left(u, \varphi_i^{(n)}\right).
\end{equation}
Denote by $M$ the maximal value of $\Phi(u, \varphi_i)$. We may assume that for all $n$ we have
\[
2M \ge \max_{u \in \mathcal{U}^1}\Phi\left(u, \varphi_i^{(n)}\right).
\]

By Proposition~\ref{thm:max_is_min_in_Q}, for each $n$ there exists a tuple of functions $\{c_i^{(n)}\}_{i\in\mathcal{I}} \subset \overline{\mathcal{C}}$ such that
\begin{equation}\label{eq:max_is_sum_c}
\max_{u \in \mathcal{U}^1} \Phi(u, \varphi_i^{(n)}) = \sum_{i\in\mathcal{I}}\int \left(\varphi_i^{(n)}\right)^*(c_i^{(n)}(x))\,{\dd}\mu.
\end{equation}
Denote by $\varsigma_i^{(n)}$ the measure $c_i^{(n)}(x)\,{\dd}\mu$. By the definition of $\big(\varphi_i^{(n)}\big)^*$, for all $x \ge 0$ and for every $t \in [0, 1)$ we have
\begin{equation}\label{eq:phi_*_bounded_from_below}
(\varphi_i^{(n)})^*(x) \ge t \cdot x - \varphi_i^{(n)}(t) \ge t \cdot x - \varphi_i(t);
\end{equation}
in the last inequality, we use that $\{\varphi_i^{(n)}\}_n$ is an increasing sequence of functions.
So, for each $i$ and for every $t \in [0, 1)$ the following inequality holds:
\begin{multline*}
2M \ge \sum_{i\in\mathcal{I}}\int \left(\varphi_i^{(n)}\right)^*(c_i^{(n)}(x))\,{\dd}\mu\ge \\
\ge \int \left(\varphi_i^{(n)}\right)^*(c_i^{(n)}(x))\,{\dd}\mu \ge t\int c_i^{(n)}(x)\,{\dd}\mu - \varphi_i(t).
\end{multline*}

This means that the sequence $\varsigma_i^{(n)}(X) = \int c_i^{(n)}(x)\,{\dd}\mu$ is bounded from above by $(2M + \varphi_i(t)) / t$. Applying the Prokhorov theorem and passing to a subsequence, we may assume that the sequence of measures $\{\varsigma_i^{(n)}\}_n$ converges weakly to some non-negative measure $\varsigma_i$. Also, applying the Komlos theorem and passing to a subsequence, we may assume that
\[
\frac{1}{n}\sum_{i = 1}^n c_i^{(n)} \underset{n \to \infty}{\to} c_i
\]
for some $c_i \in L^1(\mu)$ almost everywhere.

Since $\varsigma_i^{(n)}$ converges weakly to $\varsigma_i$, one has
\[
\lim_{n \to \infty} \int c_i^{(n)}(x)\,{\dd}\mu = \lim_{n \to \infty} \varsigma_i^{(n)}(X) = \varsigma_i(X).
\]
So, combining equations~\eqref{eq:max_Phi=lim_max_Phi_n} and~\eqref{eq:max_is_sum_c}, we conclude that for every $t \in [0, 1)$ the following equality holds:
\begin{align}\label{eq:proof_main_dual_1}
\begin{split}
\max_{u \in \mathcal{U}_{{\lip}, 1}}\Phi(u, \varphi_i) &= \lim_{n \to \infty}\sum_{i\in\mathcal{I}}\int (\varphi_i^{(n)})^*(c_i^{(n)}(x))\,{\dd}\mu=
\\ &=\sum_{i\in\mathcal{I}}t \cdot \varsigma_i(X) +  \lim_{n \to \infty} \sum_{i\in\mathcal{I}}\int \left[(\varphi_i^{(n)})^*(c_i^{(n)}(x)) - t \cdot c_i^{(n)}(x)\right]\,{\dd}\mu.
\end{split}
\end{align}

Consider the last item of the previous expression's right-hand side. By the Cesaro means,
\begin{align}\label{eq:proof_main_dual_2}
\begin{split}
\lim_{n \to \infty}\sum_{i\in\mathcal{I}}\int &\left[\left(\varphi_i^{(n)}\right)^*(c_i^{(n)}(x)) - t \cdot c_i^{(n)}(x)\right]\,{\dd}\mu =\\
&= \lim_{n \to \infty}\int \sum_{i\in\mathcal{I}}\frac{1}{n}\sum_{k = 1}^n\left[\left(\varphi_i^{(k)}\right)^*(c_i^{(k)}(x)) - t\cdot c_i^{(k)}(x)\right]\,{\dd}\mu\ge\\
&\ge \sum_{i\in\mathcal{I}}\liminf_{n \to \infty}\left(\int\frac{1}{n}\sum_{k = 1}^n\left[\left(\varphi_i^{(k)}\right)^*(c_i^{(k)}(x)) - t \cdot c_i^{(k)}(x)\right]\,{\dd}\mu\right).
\end{split}
\end{align}

Denote by $\psi_i^{(n)}(x)$ the function $(\varphi_i^{(n)})^*(c_i^{(n)}(x)) - t \cdot c_i^{(n)}(x)$. Inequality~\eqref{eq:phi_*_bounded_from_below} implies that each function $\psi_i^{(n)}(x)$ is bounded from below by $-\varphi_i(t)$, so the function $(\psi_i^{(1)}(x) + \dots + \psi_i^{(n)}(x)) / n$ is also bounded from below by $-\varphi_i(t)$. Therefore, it follows from the Fatou lemma that
\begin{align}\label{eq:proof_main_dual_3}
\begin{split}
\liminf_{n \to \infty}\Bigg(\int\frac{1}{n}\sum_{k = 1}^n&\left[(\varphi_i^{(k)})^*(c_i^{(k)}(x)) - t \cdot c_i^{(k)}(x)\right]\,{\dd}\mu\Bigg)\ge \\
&\ge \int \liminf_{n \to \infty}\left(\frac{1}{n}\sum_{k = 1}^n\left[(\varphi_i^{(k)})^*(c_i^{(k)}(x)) - t \cdot c_i^{(k)}(x)\right]\right)\,{\dd}\mu=\\
&=\int \liminf_{n \to \infty}\left(\frac{1}{n}\sum_{k = 1}^n(\varphi_i^{(k)})^*(c_i^{(k)}(x))\right)\,{\dd}\mu - t \cdot \int c_i(x)\,{\dd}\mu,
\end{split}
\end{align}
where the last equation follows from the fact that $\frac{1}{n}\sum_{i = 1}^n c_i^{(k)}$ converges to $c_i(x)$ for $\mu$-almost every~$x$.

Finally, for every $x \ge 0$ and $y \in [0, 1)$ the following inequality holds:
\[
(\varphi_i^{(n)})^*(c_i^{(n)}(x)) + \varphi_i^{(n)}(y) \ge y \cdot c_i^{(n)}(x);
\]
therefore,
\[
\frac{1}{n}\sum_{k = 1}^n(\varphi_i^{(k)})^*(c_i^{(k)}(x)) + \frac{1}{n}\sum_{k = 1}^n\varphi_i^{(k)}(y) \ge y \cdot \frac{1}{n}\sum_{k = 1}^n c_i^{(k)}(x).
\]
Passing to the limits, we conclude that
\[
 \liminf_{n \to \infty}\left(\frac{1}{n}\sum_{k = 1}^n(\varphi_i^{(k)})^*(c_i^{(k)}(x))\right) + \varphi_i(y) \ge y \cdot c_i(x)
\]
for all $y \in [0, 1]$. Thus
\begin{equation}\label{eq:proof_main_dual_4}
    \liminf_{n \to \infty}\left(\frac{1}{n}\sum_{k = 1}^n(\varphi_i^{(k)})^*(c_i^{(k)}(x))\right) \ge \max_{y \in [0, 1]} \left[y \cdot c_i(x) - \varphi_i(y)\right] =  \varphi_i^*(c_i(x)).
\end{equation}

Combining inequalities~\eqref{eq:proof_main_dual_1}, \eqref{eq:proof_main_dual_2},  \eqref{eq:proof_main_dual_3}, and \eqref{eq:proof_main_dual_4}, we conclude that for all $t \in [0, 1)$ the following inequality holds:
\[
\max_{u \in \mathcal{U}_{{\lip}, 1}}\Phi(u, \varphi_i) 
\ge \sum_{i\in\mathcal{I}}\left(t \cdot \varsigma_i(X) +  \int\varphi_i^*(c_i(x))\,{\dd}\mu - t \cdot \int c_i(x)\,{\dd}\mu\right).
\]
Letting $t$ tend to $1$, we obtain the following inequality:
\[
\max_{u \in \mathcal{U}_{{\lip}, 1}}\Phi(u, \varphi_i)
\ge \sum_{i\in\mathcal{I}}\left(\varsigma_i(X) +  \int\varphi_i^*(c_i(x))\,{\dd}\mu - \int c_i(x)\,{\dd}\mu\right).
\]

Next, we check that $c_i(x)\,{\dd}\mu \le \varsigma_i$ for all ${i\in\mathcal{I}}$. By the Cesaro means, the sequence of measures $(c_i^{(1)}(x)\,{\dd}\mu + \dots + c_i^{(n)}(x)\,{\dd}\mu) / n$ converges weakly to $\varsigma_i$. Then by the well-known property of the weak convergence for any closed subset $A$ of ${X}$ we have
\[
\limsup \int_A\frac{1}{n}\sum_{k \in \mathcal{I}} c_i^{(k)}(x)\,{\dd}\mu \le \varsigma_i(A).
\]
By the Fatou lemma,
\[
\limsup \int_A\frac{1}{n}\sum_{k \in \mathcal{I}} c_i^{(k)}(x)\,{\dd}\mu \ge \liminf \int_A\frac{1}{n}\sum_{k = 1}^{\mathcal{I}} c_i^{(k)}(x)\,{\dd}\mu \ge \int_A c_i(x)\,{\dd}\mu.
\]
Thus $\int_A c_i(x)\,{\dd}\mu \le \varsigma_i(A)$ for all closed subsets $A$ of ${X}$; therefore, $c_i(x)\,{\dd}\mu \le \varsigma_i$.

Let $\varsigma_i = \varsigma_i^a(x)\,{\dd}\mu + \varsigma_i^{\mathrm{sing}}$ be a decomposition of the component $\varsigma_i$ into an absolutely continuous and a singular part w.r.t. $\mu$. Since $c_i(x) \,{\dd}\mu \le \varsigma_i$, we conclude that $c_i(x) \le \varsigma_i^a(x)$ for $\mu$-almost every $x$. Therefore,
\begin{align*}
\varsigma_i(X) + \int\varphi_i^*(c_i(x))\,{\dd}\mu &- \int c_i(x)\,{\dd}\mu =\\ &=\varsigma_i^{\mathrm{sing}}(X) + \int(\varphi_i^*(c_i(x)) - c_i(x) + \varsigma_i^a(x))\,{\dd}\mu\ge\\
&\ge  \varsigma_i^{\mathrm{sing}}(X) + \int \varphi_i^*(\varsigma^a_i(x))\,{\dd}\mu,
\end{align*}
where the last inequality follows from the fact that $\varphi_i^*$ is a 1-Lipschitz function. Finally,
\[
\max_{u \in \mathcal{U}_{{\lip}, 1}}\Phi(u, \varphi_i) \ge \sum_{i\in\mathcal{I}} \left(\int_0^1 \varphi_i(x)\,{\dd}\eta_i + \varsigma_i^{\mathrm{sing}}(X) + \int \varphi_i^*(\varsigma^a_i(x))\,{\dd}\mu\right).
\]

Comparing this inequality to~\eqref{eq:max_u_le_min_c}, we conclude that the equality holds and thus complete the proof of Proposition~\ref{thm:monopolist_strong_duality}.
\end{proof}

\begin{proof}[Proof of Theorem~\ref{th_duality_min_appendix}]
The theorem is a combination of Proposition~\ref{thm:monopolist_strong_duality} and Theorem~\ref{duality-theorem-maxinf}.
\end{proof}

 \section{Examples and applications}\label{app_examples}
 \ed{We show how vector fields solving the dual problems from Theorem~\ref{th_vector_fields_inf} and~\ref{th_vector_fields_min} can be constructed explicitly. First, we consider several bidders competing for one item and demonstrate that the optimal vector field (a scalar, in this case) is equal to the ironed virtual valuation function. Then we use this insight to recover the result by \cite{jehiel2007mixed} that, for any number of items and bidders having independent values over them, auctioning the items separately is never an optimal mechanism. 
  Finally, we consider the one-bidder problem with two items with i.i.d. values uniform on $[0,1]$ and recover the result by \cite{manellivincent} showing that the optimal mechanism is selling each item separately together with offering the bundle for a discounted price.} 
 
\subsection{The case of one item}
\ed{Consider the auctioneer's problem with one item and $B > 1$ bidders whose values are distributed with  continuously differentiable strictly positive density $\rho$ on $[0,1]$. We will see that the optimal vector field in the dual problem
coincides with the Myersonian ironed virtual valuation function thus proving Proposition~\ref{prop_ironed}.

We allow for generalized vector fields represented by vector measures with singular components as in Theorem~\ref{th_vector_fields_min} but, as we will see below, there are no singularities in the optimum. The dual problem we start with is to find a positive measure $\varsigma$ defined on $[0, 1]$ that satisfies the constraint
\[
\int_0^1 (x \cdot u'(x) - u(x))\rho(x)\,{\dd}x \le \int u'(x)\,{\dd}\varsigma
\]
for any smooth non-decreasing convex $u$ with $u(0)=0$ and minimizes the functional
\[
\mathrm{Dual}(\varsigma) = \inf_{\varphi}\left[ \varsigma^{\mathrm{sing}}([0,1]) + \int_0^1 \varphi^*(\varsigma^a(x))\rho(x)\,{\dd}x + \int_0^1\varphi(t^{B-1})\,{\dd}t\right],
\]
where infimum is taken over convex non-decreasing functions $\varphi$ equal zero at zero.

\begin{remark}\label{rem:measure_restriction}
By  complementary slackness conditions (Corollary~\ref{cor_slackness_extended}), the absolutely continuous component   $\varsigma^a(x)$ of $\varsigma$  is a non-decreasing function of $x$ and the singular component can only be supported on $x=1$, i.e., $\varsigma^{\mathrm{sing}}$ is either absent or is a point mass at~$1$.
%
\end{remark}
}

\begin{lemma}
For any smooth $u$, 
\[
\int_0^1 (x \cdot u'(x) - u(x))\rho(x)\,{\dd}x = \int_0^1 u'(x)\cdot V(x)\rho(x)\,{\dd}x - u(0),
\]
\ed{where $V$ is the virtual valuation function:}
\[
V(x) = x - \frac{1 - \mathcal{P}(x)}{\rho(x)}, \quad \mathcal{P}(t) = \int_0^t \rho(x)\,{\dd}x.
\]
\end{lemma}
\begin{proof}
Integrating by parts,
\[
\int_0^1 u(x) \rho(x)\,{\dd}x = -\int_0^1u(x)\,d(1 - \mathcal{P}(x)) = u(0) + \int_0^1 u'(x) \cdot (1 - \mathcal{P}(x))\,{\dd}x;
\]
therefore,
\begin{multline*}
\int_0^1(x\cdot u'(x) - u(x))\rho(x)\,{\dd}x = \int_0^1 (x \cdot \rho(x) - (1 - \mathcal{P}(x))) \cdot u'(x)\,{\dd}x - u(0)= \\
= \int_0^1 u'(x)\cdot V(x)\rho(x)\,{\dd}x - u(0).
\end{multline*}
\end{proof}

\begin{proposition}
Denote $F_V(t) = \int_t^{1} V(x)\rho(x)\,{\dd}x$, $F_\varsigma(t) = \varsigma\left([t, 1]\right)$. Then the inequality
\begin{equation}\label{eq:1D_majorization}
\int_0^1 (x \cdot u'(x) - u(x))\rho(x)\,{\dd}x \le \int_0^1 u'(x)\,{\dd}\varsigma
\end{equation}
\ed{holds for all smooth convex $u$ with $u(0)=0$ if and only if} $F_\varsigma(t) \ge F_V(t)$ for all $t \ge 0$.
\end{proposition}
\begin{proof}
Integrating by parts twice, we get
\begin{align*}
    \int_0^1 &(x \cdot u'(x) - u(x))\rho(x)\,{\dd}x = \int_0^1 u'(x) V(x)\rho(x)\,{\dd}x =\\
    &= -\int_0^1 u'(x)\,{\dd}F_V(x) = u'(0) \cdot F_V(0) + \int_0^1 F_V(x)\cdot u''(x)\,{\dd}x.
\end{align*}
Next, we decompose the function $F_\varsigma(t)$ as $F_\varsigma^a(t) + \ed{\varsigma^{\mathrm{sing}}(\{1\})}$, where
\[
F_\varsigma^a(t) = \int_t^1 \varsigma^a(x)\rho(x){\dd}x.
\]
\ed{Integrating by parts in the Riemann–Stieltjes integral, we obtain}
\begin{align*}
\int_0^1 u'(x)\,{\dd}\varsigma &= u'(1) \cdot \ed{\varsigma^{\mathrm{sing}}(\{1\})} - \int_0^1 u'(t)\,{\dd}F^a_\varsigma(t)= \\
&= u'(1) \cdot  \ed{\varsigma^{\mathrm{sing}}(\{1\})} + u'(0) \cdot F^a_\varsigma(0) + \int_0^1 F^a_\varsigma(t)u''(t)\, {\dd}t=\\
&=(u'(1) - u'(0)) \cdot  \ed{\varsigma^{\mathrm{sing}}(\{1\})} + u'(0) \cdot F_\varsigma(0) + \int_0^1 F^a_\varsigma(t)u''(t)\, {\dd}t=\\
&= u'(0) \cdot F_\varsigma(0) + \int_0^1 F_\varsigma(x) \cdot u''(x)\,{\dd}x.
\end{align*}

So, the inequality
\[
F_\varsigma(0) \cdot u'(0) + \int_0^1 F_\varsigma(x) \cdot u''(x)\,{\dd}x \ge u'(0) \cdot F_V(0) + \int_0^1 F_V(x)\cdot u''(x)\,{\dd}x
\]
holds for all $u'(0) \ge 0$ and all smooth non-negative functions $u''(x)$. This happens if and only if $F_\varsigma(t) \ge F_V(t)$ for (almost) all $t \in [0, 1]$.
\end{proof}
\ed{Denote $\mathcal{P}^{-1}(x)$ by $\mathcal{Q}(x)$.} \fed{Sasha, please, check if I guessed correctly.}
Consider the function $G_\varsigma(t) = F_\varsigma(\mathcal{Q}(t)) = \int_t^1 \varsigma^a(\mathcal{Q}(x))\,{\dd}x + \ed{\varsigma^{\mathrm{sing}}(\{1\})}$ and the set 
\[
\mathcal{G} = \{g \colon [0, 1] \to \mathbb{R} \mid g(t) = G_\varsigma(t) \text{ for some $\varsigma$ satisfying Remark~\ref{rem:measure_restriction}}\}.
\]
One can check that $g \in \mathcal{G}$ if and only if $g(t)$ is a non-negative non-increasing concave function (concavity follows from the fact that $\varsigma^a(x)$ is non-decreasing). Denoting $G_V(t) = F_V(\mathcal{Q}(t))$, we conclude that the inequality $F_\varsigma(t) \ge F_V(t)$ for all $t \in [0, 1]$ is equivalent to $G_\varsigma(t) \ge G_V(t)$ for all $t \in [0, 1]$. We obtain the following statement.
\begin{corollary}
Inequality~\eqref{eq:1D_majorization} holds if and only if
\[
G_\varsigma(t) \ge \overline{G_V}(t)
\]
for all $t \in [0, 1]$, where
\[
\overline{G_V}(t) = \inf\{g(t) \colon g \in \mathcal{G} \text{ and } g(x) \ge G_V(x) \text{ for all $x \in [0, 1]$}\}
\]
\ed{is the minimal non-negative non-increasing concave function pointwise above $G_V$.}
\end{corollary}
\ed{By concavity, the derivative $-\overline{G_V}\,'(\mathcal{P}(x))$ exists almost everywhere and is known as the ironed virtual valuation function, which we denote by $\overline{V}$. Provided that $\overline{G_V}(t)$ is continuous at $t=1$ (checked in Lemma~\ref{lm:varsigma_description} below), we get $\overline{G_V}(t)=G_{\overline{V}}$.} 
\begin{proposition} For any measure $\varsigma$ satisfying all the conditions of Remark~\ref{rem:measure_restriction},
\[
\mathrm{Dual}(\varsigma) = \int_0^1 G_\varsigma(t)\, {\dd}t^{B - 1},
\]
and the equality holds if and only if $\varphi'(\mathcal{P}(x)^{B - 1}) = \varsigma^a(x)$ for (almost) all $x$.
\end{proposition}
\begin{proof}
\ed{Consider the measure $\mu$ given by $\dd\mu(x)=\rho(x)\dd x$ and let $\eta$ be the distribution of $\xi^{B-1}$ where $\xi$ is uniform on $[0,1]$.
Choose any coupling $\varkappa$ of $\mu$ and $\eta$, i.e., any measure on $[0,1]^2$ with marginals $\mu$ and $\eta$.} \ed{By the Fenchel inequality (Appendix~\ref{sect_convex_analysis}), we have}
\[
\int_0^1 \varphi^*(\varsigma^a(x))\rho(x)\,{\dd}x + \int_0^1\varphi(t^{B-1})\,\dd t \ge \int_{[0,1]^2} t\cdot\varsigma^a(x)\,\dd\varkappa(x,t).
\]
Now consider a particular choice: let $\varkappa$ be the joint law of the variable $(\chi, \mathcal{P}(\chi)^{B - 1})$, where $\chi \sim \mu$. The distribution of  $\mathcal{P}(\chi)^{B - 1}$ coincides with  $\eta$; therefore,
\begin{multline*}
\int_0^1 \varphi^*(\varsigma^a(x))\rho(x)\,{\dd}x + \int_0^1\varphi(t)\,\dd\eta(t) \ge \int_{[0,1]^2} t\varsigma^a(x)\,\dd \varkappa(x, t)= \\
= \int_0^1 \mathcal{P}(x)^{B - 1} \cdot \varsigma^a(x)\rho(x)\,{\dd}x
=-\int_0^1 \mathcal{P}(x)^{B - 1} \,{\dd}F^a_\varsigma(x) = \int_0^1 F^a_\varsigma(x)\,{\dd}\mathcal{P}(x)^{B - 1}.
\end{multline*}
For all $t \in [0, 1]$, we have $F_\varsigma^a(t) = F_\varsigma(t) - \ed{\varsigma^{\mathrm{sing}}(\{1\})}$; therefore,
\begin{align*}
    \ed{\varsigma^{\mathrm{sing}}(\{1\})} + \int_0^1 \varphi^*(\varsigma^a(x))\rho(x)\,{\dd}x& + \int_0^1\varphi(t^{B-1})\,{\dd}t \ge\\
    &\ge\ed{\varsigma^{\mathrm{sing}}(\{1\})} + \int_0^1(F_\varsigma(x) - \ed{\varsigma^{\mathrm{sing}}(\{1\})})\,{\dd}\mathcal{P}(x)^{B - 1}= \\
     &= \int_0^1 F_\varsigma(x)\,{\dd}\mathcal{P}(x)^{B - 1} = \int_0^1 G_\varsigma(x)\, {\dd}t^{B - 1}.
\end{align*}

The equality holds if and only if $\varphi(\mathcal{P}(x)^{B - 1}) + \varphi^*(\varsigma^a(x)) = \mathcal{P}(x)^{B - 1} \cdot \varsigma^a(x)$ for almost all $x$, which is equivalent to the condition
\begin{equation}\label{eq:subgradient_of_phi}
\varsigma^a(x) \in (\partial \varphi)(\mathcal{P}(x)^{B - 1})
\end{equation}
for all $x \in [0, 1]$. \ed{Such a convex function $\varphi$ exists since $\varsigma^a(x)$ is non-decreasing. The function $\varphi$ is unique up to an additive constant pinned down by the requirement $\varphi(0)=0$.}
\end{proof}
\ed{
\begin{corollary}\label{cor_one_item_appendix}
A measure $\varsigma$ solves the dual problem if and only if $G_\varsigma = G_{\overline{V}}$ or, equivalently, $\dd\varsigma(x)=\varsigma^a(x)\rho(x)\dd x$ where $\varsigma^a(x)=\overline{V}(x)$. We conclude that, for one item, the optimum in Theorem~\ref{th_vector_fields_min} is unique, has no singular components and so is attained at a classical ``field'' $c=\varsigma^a$ from Theorem~\ref{th_vector_fields_inf}. This field coincides with the ironed virtual valuation function and, hence, we obtain Proposition~\ref{prop_ironed}. The optimal function $\varphi$ is also unique and is defined by~\eqref{eq:subgradient_of_phi}.
\end{corollary}
}
\ed{It remains to check that $\overline{G_{V}}$ is continuous at $1$. We prove a stronger statement:   $\overline{G_{V}}(t)$ coincides with ${G_{V}}(t)$ for high enough $t$. In particular, $\overline{V}=V$ for high types, i.e., high types are never ironed.
\begin{lemma}\label{lm:varsigma_description}
There exists $a \in (0, 1)$ such that $\overline{G_V}(t)=G_V(t)$ for all $t \in [a, 1]$.
\end{lemma}
\begin{proof}
Consider the derivative
\[
-G'_V(t) = V(\mathcal{Q}(t)) = \mathcal{Q}(t) - (1 - t)\mathcal{Q}'(t).
\]
It is enough to check that there exists $a < 1$ such that $V(\mathcal{Q}(t))$ is increasing on the interval $[a, 1]$ and $V(\mathcal{Q}(t)) \le V(\mathcal{Q}(a))$ for all $t \le a$. Indeed, 
\[
\frac{d}{dt} V(\mathcal{Q}(t)) = 2\mathcal{Q}'(t) - (1 - t)\mathcal{Q}''(t) > 0 \quad\text{for all $t \ge a_0$.}
\]
In addition, $V(\mathcal{Q}(1)) = \mathcal{Q}(1) = 1$ and $V(\mathcal{Q}(t)) \le \mathcal{Q}(t)$ for $t \in [0, 1]$. So, there exists $a \in [a_0, 1)$ such that $V(\mathcal{Q}(a)) \ge \mathcal{Q}(a_0) \ge V(\mathcal{Q}(t))$ for all $t \le a_0$. For such a parameter $a$, the function $V(\mathcal{Q}(t))$ is increasing on $[a, 1]$.
\end{proof}
}

\subsection{The case of several items and bidders: suboptimality of selling separately} 
\ed{Consider $B>1$ bidders competing for $I>1$ items and assume that their values are distributed on $X=[0,1]^I$ according to a density
\[
\rho(x_1,\ldots, x_I) = \rho_1(x_1)\cdot \rho_2(x_2)\cdot  \ldots \cdot\rho_I(x_I)
\]
with  continuously differentiable $\rho_i$ strictly positive on $[0,1]$.  Building on explicitly solved dual problem for $I=1$ item, we will show that it is never optimal to sell $I>1$ items separately using the Myersonian optimal auction for each of them.  

It is enough to demonstrate suboptimality of any mechanism such that the reduced-form allocation of each item $i$ depends on the value $x_i$ for this item only. Equivalently, it is enough to show that the optimal solution $u=u^\opt(x_1, \dots, x_I)$ to the primal problem cannot have the form
\[
u(x_1, \dots, x_I) = u_1(x_1) + \dots + u_I(x_I).
\]

Towards a contradiction, assume that the optimal $u$ has such a separable form. Let us show that the optimal dual solution $\varphi_1, \dots, \varphi_I$ and $\varsigma = (\varsigma_1, \dots, \varsigma_I)$ from Theorem~\ref{th_vector_fields_min} is composed of $(\varsigma_i,\varphi_i)$ that are optimal in the corresponding one-item problem.
By  complementary slackness conditions (Corollary~\ref{cor_slackness_extended}), 
\begin{itemize}
    \item the absolutely continuous component $\varsigma_i^a(x_1, \dots, x_I)$ is a non-decreasing function of $x_i$ for all $i = 1, \dots, I$;
   \item the support of the singlular component $\varsigma_i^{\mathrm{sing}}$ is contained in the set $\{x_i = 1\}$.
\end{itemize}
 Denote the marginal of $\varsigma_i$ on $x_i$ by  $\widehat{\varsigma}_i$. By the feasibility of $\varsigma$, for any smooth convex non-deceasing $v=v(x_i)$ with $v(0)=0$, we get}
\begin{align*}
\int_0^1 (x_i v'(x_i) - v(x_i)) \rho_i(x_i)\,{\dd}x_i &= \int_X (\langle x, \nabla v \rangle - v)\,\rho(x_1,\ldots,x_I)\dd x_1\ldots \dd x_I\le \\
&\le \sum_{k = 1}^I \int_X \frac{\partial v}{\partial x_k } {\dd}\varsigma_k = \int_0^1 v'(x_i)\,\dd \widehat{\varsigma}_i(x_i).
\end{align*}
\ed{Therefore, $\widehat{\varsigma}_i$  is feasible in the one-item dual problem with item $i$ and density $\rho_i$. 
By Theorem~\ref{th_vector_fields_min} and the Jensen's inequality,}
\begin{align*}
\sum_{i = 1}^I \int_0^1& (x_i u'(x_i) - u(x_i)) \rho_i(x_i)\,{\dd}x_i = \int_X (\langle x, \nabla u \rangle - u(x))\,\rho(x_1,\ldots,x_I)\dd x_1\ldots \dd x_I =\\
&= \sum_{i = 1}^I\left(\int_0^1 \varphi_i(t^{B-1})\,{\dd}t + \int_X \varphi_i^*(\varsigma_i^a(x))\,\rho(x_1,\ldots,x_I)\dd x_1\ldots \dd x_I + \varsigma_i^{\mathrm{sing}}(X)\right) \ge\\
&\ge \sum_{i = 1}^I\left(\int_0^1 \varphi_i(t^{B-1})\,{\dd}t + \int_0^1 \varphi_i^*(\widehat{\varsigma}_i^a(x_i))\,\rho_i(x_i){\dd}x_i + \widehat{\varsigma}_i\,^{\mathrm{sing}}([0,1])\right).
\end{align*}
\ed{We derive the following conclusions:
\begin{itemize}
    \item The pair $(\widehat{\varsigma}_i$, $\varphi_i)$ is an optimal solution to the corresponding dual one-item problem with density $\rho_i$.
    \item By Corollary~\ref{cor_one_item_appendix}, the singular component $\widehat{\varsigma}_i\,^{\,\mathrm{sing}}$ (and thus $\varsigma_i^{\mathrm{sing}})$ are absent.
    \item By Jensen's inequality, the absolutely continuous components $\varsigma_i^a(x_1, \dots, x_n) = \widehat{\varsigma}_i^a(x_i)$ for almost all $x_1, \dots, x_{i-1}, x_{i+1}, \dots, x_I$ and $x_i$ such that $\varphi^*$ is strictly convex at the point $\widehat{\varsigma}^a_i(x_i)$. By~\eqref{eq:subgradient_of_phi}, we have
    \[
    \widehat{\varsigma}_i^a(x_i) \in (\partial\varphi_i)(\mathcal{P}_i^{B - 1}(x_i)) \quad \Leftrightarrow \quad \mathcal{P}_i^{B - 1}(x_i) \in (\partial\varphi_i^*)(\widehat{\varsigma}_i^a(x_i)).
    \]
    Hence, $\varphi_i^*$ is strictly convex at  $\widehat{\varsigma}_i^a(x_i)$ if $\widehat{\varsigma}_i^a(x_i)$ is continuous at $x_i$.
\end{itemize}
By Lemma~\ref{lm:varsigma_description}, $\widehat{\varsigma}_i^a(x_i)$ coincides with the virtual valuation function $V_i$ for all $x_i \in [a_i, 1]$; therefore, $\varsigma_i^a(x_1, \dots, x_I) = V_i(x_i)$ for almost all $x_i \in [a_i, 1]$ and all other coordinates.

Consider a particular test function $v(x_1, \dots, x_I) = \max(x_1 + \dots + x_I - a, 0)$ such that $B> a > B - 1 + \max_i a_i$. This function is convex, non-decreasing, and equals zero at zero.} Integration by parts and the definition of the virtual values $V_i$ imply
\begin{align*}
\int_X \left(x_i \frac{\partial v}{\partial x_i}(x_1, \dots, x_I) - v(x_1, \dots, x_I)\right)\,\rho(x_1,\ldots,x_I)\dd x_1\ldots \dd x_I&= \\
= \int_X V_i(x_i) \cdot \frac{\partial v}{\partial x_i}(x_1, \dots, x_I)\,\rho(x_1,\ldots,x_I)\dd x_1\ldots \dd x_I &= \int_X \frac{\partial v}{\partial x_i} {\dd}\varsigma_i.
\end{align*}
\ed{Summing up these identities over $i$, we obtain} 
\begin{align*}
\int_X (&\langle x, \nabla v\rangle - v)\,\rho(x_1,\ldots,x_I)\dd x_1\ldots \dd x_I =\\
&=\sum_{i = 1}^I\int_X \frac{\partial v}{\partial x_i} {\dd}\varsigma_i + (I - 1)\int_X v\,{\dd}\,\rho(x_1,\ldots,x_I)\dd x_1\ldots \dd x_I 
>\sum_{i = 1}^I\int_X \frac{\partial v}{\partial x_i} {\dd}\varsigma_i.
\end{align*}
\ed{Thus $\varsigma$ is not a feasible solution to the dual problem. This contradiction implies that, in the optimal mechanism, the allocation of item 
$i$ cannot depend on  $x_i$ exclusively. In particular, running $I$ separate auctions is not optimal.}

\subsection{The case of one bidder}  
 \ed{We consider a benchmark problem with  two items and one bidder whose values are uniformly distributed on $[0,1]^2$ and show how to solve the dual problem from Theorem~\ref{th_vector_fields_min}.  As we will see, the solutions may be non-unique and singular.} 

\ed{The optimal mechanism for this problem  was obtained by \cite{manellivincent}: each item is offered for the price of $\frac{2}{3}$ and the grand bundle, for $\frac{4-\sqrt{2}}{3}$. Our dual solution gives an optimality certificate for this mechanism and thus provides an alternative proof of its optimality.}

\cite{daskalakis2017strong}
\ed{derived the mechanism of  \cite{manellivincent}} via the the associated Monge-Kantorovich transportation problem~\eqref{eq_Daskalakis_result}. 
The optimal function $u=u^\opt$ is given by
\begin{equation}
\label{uauctionsol}
    u^\opt(x,y) = \begin{cases}
    0 & (x,y) \in  \mathcal{Z} \\
    x - \frac{2}{3}& (x,y) \in \mathcal{A} \\
    y -\frac{2}{3} & (x,y) \in \mathcal{B} \\
    x+y - \frac{4-\sqrt{2}}{3} & (x,y) \in \mathcal{W}
    \end{cases},
\end{equation}
where the sets $\mathcal{Z},\mathcal{A}, \mathcal{B}$ and $\mathcal{W}$ are depicted in Figure~\ref{fig:u_map} borrowing the notation  from the original paper.
\begin{figure}[!h]
   \centering
   \includegraphics[width=0.5\textwidth]{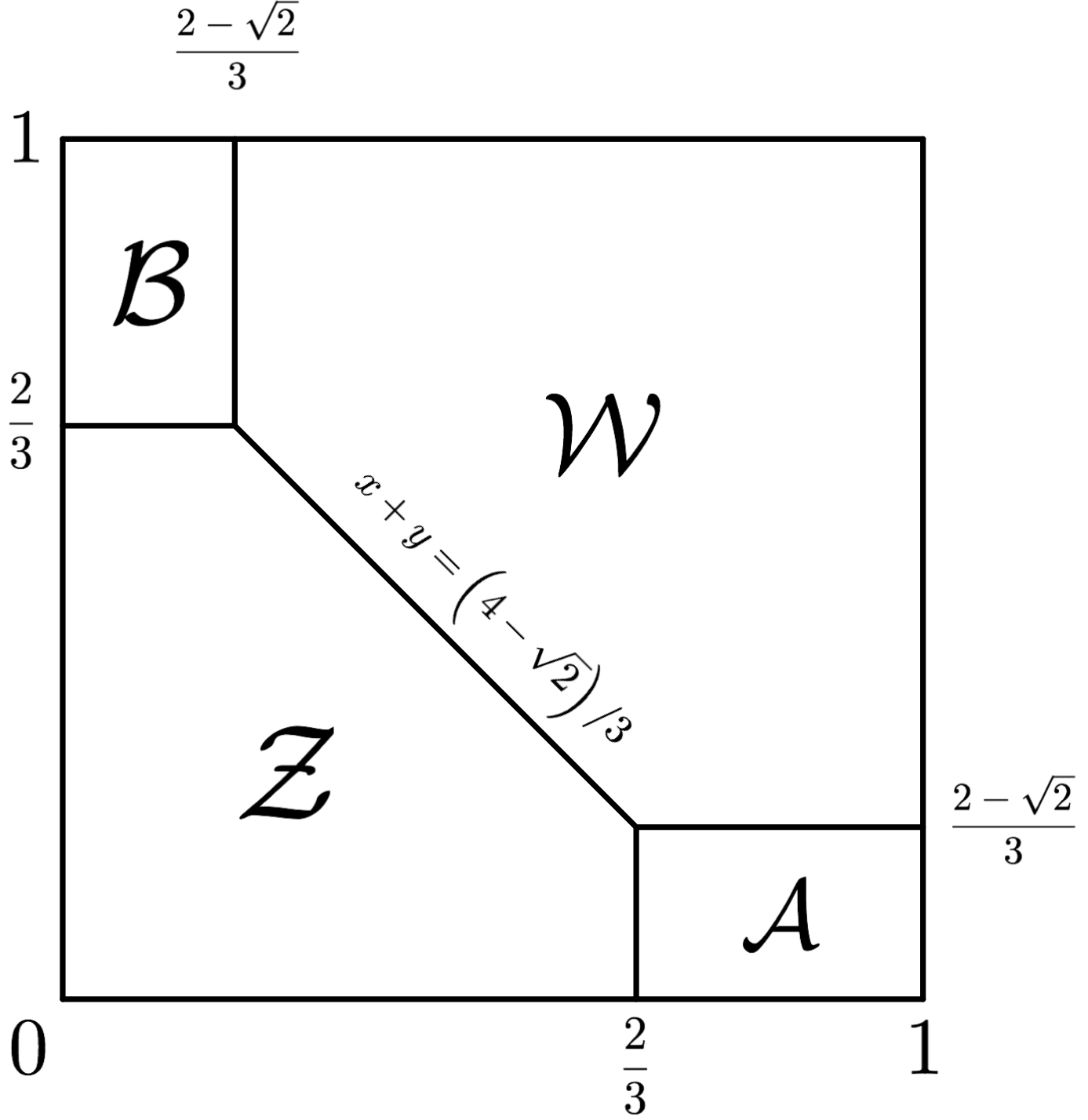}
    \caption{Partition of the square with respect to the optimal $u$.}
    \label{fig:u_map}
\end{figure}
The answers for the transform measure $m$ defined by~\eqref{eq_transform_measure} and for the optimal ``imbalance'' $\pi=\pi^\opt$ majorizing  $m$ are as follows: 
\begin{equation}\label{eq_m_uniform}
m =   \delta_0 
+ \lambda_1|_{[0,1] \times \{0\}} + \lambda_1|_{\{0\} \times[0,1] }
- 3 \lambda_2|_{[0,1]^2},
\end{equation}
\begin{equation}\label{eq_pi_uniform}
\pi^\opt =    \lambda_1|_{[0,1] \times \{0\}} + \lambda_1|_{\{0\} \times[0,1] }
-3 \lambda_2|_{[0,1]^2 \setminus \mathcal{Z}},
\end{equation}
where $\lambda_2, \lambda_1$ are the two- and one-dimensional Lebesgue measures, respectively. \ed{We rely on these observations to simplify the construction.}

Recall that for $B=1$ bidder, the dual problem from Theorem~\ref{th_vector_fields_min} can be simplified (Corollary~\ref{1bidderbeckmann}). It takes the following form:
\begin{equation}\label{eq_dual_uniform}
\mbox{minimize:}\quad\int_{[0,1]^2} \bigl( d \varsigma_1 + d \varsigma_2 \bigr)
\end{equation}
over vector measures $\varsigma = (\varsigma_1, \varsigma_2)$  satisfying
\begin{equation}\label{eq_dual_feasibility_uniform}
  \int u {\dd} m
  \le \int (u_x d \varsigma_1 + 
  u_y d \varsigma_2)
\end{equation}
for all convex non-decreasing $u$ with $u(0)=0$.

First, let us construct an absolutely continuous solution, i.e., such that 
${\dd}\varsigma_i = c^i {\dd} x {\dd} y$. 
We will need the following lemma.

\begin{lemma}
\label{c1c2lemma}
Assume that a couple of nonnegative functions $c^1, c^2$ satisfy
\begin{enumerate}
\item
\begin{equation}
\label{zerosets}
c^1|_{\mathcal{Z} \cup \mathcal{B}} =0, \ c^2|_{\mathcal{Z} \cup \mathcal{A}} =0, 
\end{equation}
\item
$c^1$ is weakly differentiable along $x$ and satisfies the following integration by parts identity for every smooth $\varphi$
$$
\int_{[0,1]^2} \varphi_x \cdot c^1 {\dd} x {\dd} y
=  - \int_{[0,1]^2} \varphi\cdot  c^1_x {\dd} x {\dd} y + \int_0^1 \varphi(1,y) {\dd} y.
$$
Similarly
$c^2$ is weakly differentiable along $y$ and satisfies the following integration by parts identity for every smooth $\varphi$
$$
\int_{[0,1]^2} \varphi_y \cdot c^2 {\dd} x {\dd} y
=  - \int_{[0,1]^2} \varphi \cdot c^2_y {\dd} x {\dd} y + \int_0^1 \varphi(x,1) {\dd} x.
$$
\item 
\begin{equation}
    \label{div3}
c^1_x + c^2_y =3,
\end{equation}
on $[0,1]^1 \setminus \mathcal{Z}$ almost everywhere.
\end{enumerate}
Then the vector field $c=(c^1,c^2)$  satisfies
\begin{equation}
    \label{2itemdual}
    {\div} [c] +\pi^\opt=0
\end{equation}
and $\varsigma=(\varsigma_1,\varsigma_2)$ such that ${\dd}\varsigma_i=c^i{\dd} x{\dd} y$
is a solution to the dual problem~\eqref{eq_dual_uniform}.
\end{lemma}
\begin{proof}
Take any convex non-decreasing  $1$-Lipschitz  function $u$ with $u(0)=0$. Then
\begin{align*}
\int \bigl( x u_x + y u_y - u\bigr) {\dd} x {\dd} y = \int u {\dd} m \le \int u {\dd} \pi^\opt.
\end{align*}
 The first equality is  the definition of the transform measure $m$, the second one holds since
$ m \preceq_{}  \pi^\opt$. 
Since $\pi^\opt$ is given by an explicit formula~\eqref{eq_pi_uniform}, the identity for the divergence (\ref{2itemdual}) follows from an elementary computation. 
 Using (\ref{2itemdual}) and the definition of divergence, we obtain
\begin{align*}
\int u {\dd}\pi^\opt &= - \int u \ {\div} [c] {\dd} x {\dd} y = \int \langle \nabla u, c \rangle \ {\dd} x {\dd} y\le\\
&\le \int (c^1 + c^2) {\dd} x {\dd} y= \int \bigl( d \varsigma_1 + d \varsigma_2 \bigr),
\end{align*}
where, we used that $0 \le u_x \le 1$ and  $0 \le u_y \le 1$ to get the inequality.
Substituting $u=u^\opt$ given by (\ref{uauctionsol}), we see that the two inequalities become equalities because 
$u|_{\mathcal{Z}}=0$ and $c|_{\mathcal{Z}}=0$. We conclude that the objective in the dual problem on $\varsigma$ coincides with the optimal value of the primal problem $\int u^\opt{\dd} m$. Thus $\varsigma$ is the optimal solution of the dual.
\end{proof}

\begin{example}[Absolutely continuous solution]
Consider the following vector field:
$$
c(x,y) = (c^1(x,y), c^2(x,y)),
$$
where 
$$
c^2(x,y) = c^1(y,x)
$$
and
\begin{equation*}
c^1(x,y) = \begin{cases}
0 & (x,y) \in \mathcal{Z} \cup \mathcal{B} \\
3x -2 &  (x,y) \in \mathcal{A} 
\\
\frac{3}{2} \bigl( x+y - \frac{4-\sqrt{2}}{3} \bigr) 
& (x,y) \in \mathcal{W}, x \le \frac{2}{3}, y \le \frac{2}{3}
\\ \frac{9}{2}(1-x) \bigl( y - \frac{2-\sqrt{2}}{3}\bigr)
\\\quad\quad\quad\quad+ 3  (x - \frac{2}{3})
&  (x,y) \in \mathcal{W},  \frac{2}{3} \le x \le 1, \frac{2-\sqrt{2}}{3} \le y \le \frac{2}{3}
\\
\frac{9}{4} \bigl( x - \frac{2-\sqrt{2}}{3} \bigr)^2 
& (x,y) \in \mathcal{W}, \frac{2-\sqrt{2}}{3} \le x \le \frac{2}{3}, \frac{2}{3} \le y \le 1
\\ \frac{1}{2} + \frac{3}{2} \bigl( x - \frac{2}{3} \bigr) &
(x,y) \in \mathcal{W}, \frac{2}{3} \le x \le 1, \frac{2}{3} \le y \le 1
\end{cases}
\end{equation*}
One can   check that the vector field $c=(c^1,c^2)$ satisfies the assumptions of Lemma~\ref{c1c2lemma}. We conclude that $c$ solves the dual problem.

\end{example}

\begin{remark}[Non-uniqueness]
It turns out that there are many solutions to the dual problem.
 However, the reader should be not confused by existing results on uniqueness of the optimal vector field $c$ in Beckmann's problem; see, e.g., \cite{santambrogio2015optimal}. Unlike most of the works on Beckmann's problem, our cost function is given by the $l^1$-norm $\sum_i |c_i|$ instead of the Euclidean $l^2$-norm $\sqrt{\sum_i |c_i|^2}$. The $l^1$-norm constitutes a degenerate case.
Indeed, 
if $c$ is a solution and $\varphi$ is a smooth function,
supported on a small neighbourhood of a point $(x_0,y_0) \in {\rm int}(\mathcal{W})$, where $c^i(x_0,y_0)>0$, then for sufficiently small $\varepsilon$ the vector field
$$
c_{\varepsilon} =(c^1 + \varepsilon \varphi_y, c^2 - \varepsilon \varphi_x)
$$
satisfies all the assumptions. Integrating by parts one gets
$$
\int \sum_{i=1}^2 c^i_{\varepsilon} \ {\dd} x {\dd} y
= \int \sum_{i=1}^2 c^i \ {\dd} x {\dd} y.
$$
Thus $c_{\varepsilon}$ is also a solution.

Moreover, one can easily find
 solutions which are  not weakly differentiable.
Let  $a, \delta$ be numbers and $Q$ be the square with the center $(a,a)$ and vertices
$$
q_{-,a} =(a-\delta,a),\ \  q_{a,-}=(a, a-\delta),\ \ 
q_{+,a}=(a+\delta, a), \ \ 
q_{a,+} = (a, a+\delta).
$$
Define
\begin{equation*}
\psi_1(x,y) = 
I_{Q}(x,y) \bigl( -I_{a-\delta \le x \le a}(x) + I_{a \le x \le a+\delta}(x) \bigr) 
\end{equation*}
\begin{equation*}
\psi_2(x,y) = 
I_{Q}(x,y) \bigl( -I_{a-\delta \le y \le a}(y) + I_{a \le y \le a+\delta}(y) \bigr) 
\end{equation*}
It is easy to verify that
\begin{equation}
\label{psixy}
\frac{\partial \psi_1(x,y)}{\partial y}
=
\frac{\partial \psi_2(x,y)}{\partial x}
=
\frac{1}{\sqrt{2}} \Bigl[ -\lambda_1|_{[q_{-,a}, q_{a,-}]} 
+ \lambda_1|_{[q_{-,a}, q_{a,+}]} 
-\lambda_1|_{[q_{a,+}, q_{+,a}]} 
+ \lambda_1|_{[q_{a,-}, q_{+,a}]} 
\Bigr]
\end{equation}
in the weak sense,
where $[a,b]$ denotes the segment joining $a$ and $b$.
Clearly,
$$
\int \psi_1 {\dd} x {\dd} y 
=\int \psi_2 {\dd} x {\dd} y 
=0
$$
and
(\ref{psixy}) implies that
$$
{\div} [(\psi_2, -\psi_1)] =0.
$$
Thus for any solution $c$ to the dual problem, strictly positive
in some neighbourhood $U$ of a point $(a,a) \in \rm{int}(\mathcal{W})$, the vector field
$$
c + (\psi_2,-\psi_1)
$$
is a solution to the dual problem for sufficiently small $\delta$.
\end{remark}
\subsubsection{Singular solutions}
It may seem intuitive --- at least for our toy example --- that vector fields solving the dual problem must be integrable functions. Surprisingly, there exist singular solutions. We construct a measure-valued solution $\varsigma$  with the following properties:
\begin{itemize}
\item The vector field $\varsigma$ is singular, i.e., its components are not  absolutely continuous measures: namely, $\varsigma$ has an atom at $(1,1)$.
\item Because of this atom, the divergence of $\varsigma$ is not a measure and can only be defined in the space of generalized functions.
\end{itemize}
\begin{example}[Singular solution]
Define a couple of measures $(\varsigma_1, \varsigma_2)$ as follows:
\begin{align*}
    \varsigma_1 &=  \mathbbm{1}((x, y) \in \mathcal{A}) \cdot 3\left(x - \frac{2}{3}\right)\,{\dd} x\,{\dd} y +\\
    &\quad\quad+ \mathbbm{1}\left((x, y) \in \mathcal{W} \cap \Bigl\{y \le \frac{2}{3}\Bigr\}\right) \cdot 3\left(x + y - \frac{4 - \sqrt{2}}{3}\right)\,{\dd} x\,{\dd} y,\\
    \varsigma_2 &= \mathbbm{1}\left(y \ge \frac{2}{3}\right) \cdot 3\left(y - \frac{2}{3}\right)\,{\dd} x\,{\dd} y + C \cdot \delta(x=1, y=1),
\end{align*}
where $C = \frac{1}{18} + \frac{\sqrt{2}}{27}$;
see Figure~\ref{fig:u_map}.

Let us show that the vector-measure $\varsigma=(\varsigma_1, \varsigma_2)$ is a solution to the dual problem. We need to demonstrate that $\varsigma$
satisfies the majorization condition~\eqref{eq_dual_feasibility_uniform} and minimizes the dual objective~\eqref{eq_dual_uniform} over such vector measures.

First, we check that $\varsigma$ satisfies~\eqref{eq_dual_feasibility_uniform}. Integrating by parts, we conclude that for any smooth $u$ defined on $[0, 1]^2$,
\begin{align}\label{eq:dc_by_parts}
\begin{split}
    \int \frac{\partial u}{\partial x}\,{\dd}\varsigma_1 &= \int_0^{\frac{2 - \sqrt{2}}{3}} v(y)\,{\dd} y + \int_{\frac{2 - \sqrt{2}}{3}}^{\frac{2}{3}} v(y) \cdot 3\left(y + \frac{\sqrt{2} - 1}{3}\right)\,{\dd} y- \\
    &\quad- 3\int_{ (\mathcal{W} \cup \mathcal{A}) \cap \{y \le \frac{2}{3}\} } u(x, y)\,{\dd} x\,{\dd} y,\\
    \int \frac{\partial u}{\partial y}\,{\dd}\varsigma_2 &= \int_0^1 u(x, 1)\,{\dd} x + C \cdot v'(1) - 3\int_{y \ge \frac{2}{3}} u(x, y)\,{\dd} x\,{\dd} y,
\end{split}
\end{align}
where $v(y) = u(1, y)$.

Let us prove that for any smooth convex non-decreasing $u$ with $u(0)=0$,
\begin{align*}
\int_0^1\int_0^1 \left( x \cdot \frac{\partial u}{\partial x} + y \cdot \frac{\partial u}{\partial y} - u\right)\,{\dd} x\,{\dd} y \le \int \frac{\partial u}{\partial x}\,{\dd}\varsigma_1 + \int \frac{\partial u}{\partial y}\,{\dd}\varsigma_2.
\end{align*}
Integrating by parts,
\begin{multline*}
\int_0^1\int_0^1 \left( x \cdot \frac{\partial u}{\partial x} + y \cdot \frac{\partial u}{\partial y} - u\right)\,{\dd} x\,{\dd} y =\\
= \int_0^1 u(x, 1)\,{\dd} x + \int_0^1 v(y)\,{\dd} y - 3\int_0^1\int_0^1 u(x, y)\,{\dd} x\,{\dd} y.
\end{multline*}
Comparing it to~\eqref{eq:dc_by_parts}, we conclude that the inequality above is equivalent to the following one:
\[
C \cdot v'(1) - \int_{\frac{2}{3}}^1 v(y)\,{\dd} y + \int_{\frac{2 - \sqrt{2}}{3}}^{\frac{2}{3}} v(y) \cdot 3\left(y - \frac{2 - \sqrt{2}}{3}\right)\,{\dd} y \ge -3\int_\mathcal{Z}u(x, y)\,{\dd} x\,{\dd} y.
\]

The right-hand side is non-positive, so it is sufficient to prove that the left-hand side is non-negative for any smooth  convex non-decreasing $u$ with $u(0)=0$. For any such $u$, the function $v(y) = u(1, y)$ is a non-decreasing convex function defined on $[0, 1]$. So, for any $y \in \left[\frac{2}{3}, 1\right]$, we have $v(y) \le v(2 / 3) + (y - 2/3) \cdot v'(1)$. Therefore,
\[
\int_{\frac{2}{3}}^1 v(y)\,{\dd} y \le \frac{1}{3} \cdot  v\left(\frac{2}{3}\right) + \frac{1}{18} \cdot v'(1).
\]
In addition, for any $y \in \left[\frac{2 - \sqrt{2}}{3}, \frac{2}{3}\right]$, we have $$v\left(\frac{2}{3}\right) - v(y) \le v'\left(\frac{2}{3}\right) \cdot \left(\frac{2}{3} - y\right) \le v'(1) \cdot \left(\frac{2}{3} - y\right).$$ Therefore,
\begin{multline}\label{eq:v_est_1}
\int_{\frac{2 - \sqrt{2}}{3}}^{\frac{2}{3}} v(y) \cdot 3\left(y - \frac{2 - \sqrt{2}}{3}\right)\,{\dd} y \ge \\
\ge v\left(\frac{2}{3}\right) \cdot \int_{\frac{2 - \sqrt{2}}{3}}^{\frac{2}{3}} 3\left(y - \frac{2 - \sqrt{2}}{3}\right)\,{\dd} y - v'(1) \cdot \int_{\frac{2 - \sqrt{2}}{3}}^{\frac{2}{3}}3\left(y - \frac{2 - \sqrt{2}}{3}\right) \left(\frac{2}{3} - y\right)\,{\dd} y \\
= \frac{1}{3} \cdot v\left(\frac{2}{3}\right) - \frac{\sqrt{2}}{27} \cdot v'(1).
\end{multline}
Finally,
\begin{multline}\label{eq:v_est_2}
    C \cdot v'(1) - \int_{\frac{2}{3}}^1 v(y)\,{\dd} y + \int_{\frac{2 - \sqrt{2}}{3}}^{\frac{2}{3}} v(y) \cdot 3\left(y - \frac{2 - \sqrt{2}}{3}\right)\,{\dd} y \ge \\
    \ge C \cdot v'(1) - \frac{1}{3} \cdot  v\left(\frac{2}{3}\right) - \frac{1}{18} \cdot v'(1) + \frac{1}{3} \cdot v\left(\frac{2}{3}\right) - \frac{\sqrt{2}}{27} \cdot v'(1) = 0.
\end{multline}
We conclude that $(\varsigma_1, \varsigma_2)$ satisfies~\eqref{eq_dual_feasibility_uniform}.

Let us verify the optimality of $(\varsigma_1, \varsigma_2)$. It is enough to check that the value of the dual objective~\eqref{eq_dual_uniform}
on $\varsigma$ coincides with the optimal value of the primal problem. Recall that $u^\opt$ denotes the optimal function in the primal problem. Hence,
\[
\int \left({\dd} \varsigma_1 + {\dd}\varsigma_2\right) =  \int \frac{\partial u^\opt}{\partial x}\,{\dd}\varsigma_1 + \int \frac{\partial u^\opt}{\partial y}\,{\dd}\varsigma_2.
\]
Thus it is enough to check that the right-hand side is equal to the value of the primal problem:
$$
\int_0^1\int_0^1 \left( x \cdot \frac{\partial u^\opt}{\partial x} + y \cdot \frac{\partial u^\opt}{\partial y} - u\right)\,{\dd} x\,{\dd} y = \int \frac{\partial u^\opt}{\partial x}\,{\dd}\varsigma_1 + \int \frac{\partial u^\opt}{\partial y}\,{\dd}\varsigma_2.
$$
Equivalently,
\begin{multline*}
C \cdot (v^\opt)'(1) - \int_{\frac{2}{3}}^1 v^\opt(y)\,{\dd} y + \int_{\frac{2 - \sqrt{2}}{3}}^{\frac{2}{3}} v^\opt(y) \cdot 3\left(y - \frac{2 - \sqrt{2}}{3}\right)\,{\dd} y=\\ =-3\int_\mathcal{Z}u^\opt(x, y)\,{\dd} x\,{\dd} y.
\end{multline*}
By~\eqref{uauctionsol}, the right-hand side is equal to~0. The function $v^\opt$ is linear on $\left[\frac{2 - \sqrt{2}}{3}, 1\right]$ and, hence, both inequalities~\eqref{eq:v_est_1} and~\eqref{eq:v_est_2} hold as equalities. Therefore, the left-hand side is also~0. Thus $\varsigma$ is an optimal solution to the dual problem as its objective~\eqref{eq_dual_uniform} on $\varsigma$ is equal to the optimal value of the primal problem.
\end{example}


\subsection{Upper bound on auctioneer's revenue}\label{app_full_surplus}

In Section~\ref{sec_duality}, we showed that the auctioneer's revenue is bounded from above by
\begin{equation}\label{rhozetB_appendix}
B \cdot \inf_{
  \footnotesize{\begin{array}{c}
       (\varphi_{i})_{i\in\mathcal{I}}
  \end{array}}}
 \sum_{i\in \mathcal{I}}\left( \int_X  \varphi_i^*(x_i) \rho(x){\dd} x+ \int_0^1 \varphi_{i}\left(z^{B-1}\right){\dd} z\right).
\end{equation}
for any number $B$ of bidders, $I$ of items, and any density $\rho$; see formula~\eqref{rhozetB}. Here we show that this upper bound corresponds to full surplus extraction. 

Our goal is to show that the expression~\eqref{rhozetB_appendix} equals to the full surplus defined by
\begin{equation}\label{eq_full_surplus}
\sum_{y\in\mathcal{I}}\mathbb{E}\left[\max_{b\in\mathcal{B}} \chi_{b,i}\right],
\end{equation}
where $\chi_b\in X$ are i.i.d. random  vectors distributed with density $\rho$. 

Let $\rho_i$ be the one-dimensional marginals
of $\rho$ onto the $i$-th coordinate. Then~\eqref{rhozetB_appendix} equals to
$B \sum_{i \in \mathcal{I}} DMK_i$, where
$$
DMK_i = \inf_{
  \footnotesize{\begin{array}{c}
       (\varphi_{i})_{i\in\mathcal{I}}
  \end{array}}}
 \left( \int_0^1  \varphi_i^*(x_i) \rho_i(x_i){\dd} x_i+ \frac{1}{B-1} \int_0^1 \varphi_{i} (y_i) y^{\frac{2-B}{B-1}}_i {\dd} y_i\right).
$$
The value $DMK_i$ is nothing else but the value of the dual Monge--Kantorovich problem for the cost function $ - x_i y_i$. Adding the terms $\frac{1}{2} \int_0^1 x^2_i \rho_i(x_i) {\dd} x_i$ and   $\frac{1}{2(B-1)}\int_0^1 y^{2+ \frac{2-B}{B-1}}_i {\dd} y_i$ with known value, the reader can easily verify that this problem is equivalent to the transportation problem with the standard
cost $\frac{1}{2}|x_i - y_i|^2$. Thus, according to the one-dimensional version of the Brenier theorem, the solution is concentrated on the graph of the mapping $T_i$ given by
$$
\int_0^{t_i} \rho_i {\dd} x_i
= \frac{1}{B-1} \int_0^{T(t_i)} y_i^{\frac{2-B}{B-1}} {\dd} y_i
= T^{\frac{1}{B-1}}_i(t_i)
$$
and the cost equals $\int_0^1 x_i T_i(x_i) \rho_i{\dd} x_i$. Let $F_i$ be the cumulative distribution function of $x_i$.
Finally, we get
$$
DMK_i =  \int_0^1 x_i  
 F_i^{B-1}(x_i) \rho_i(x_i) {\dd} x_i= \frac{1}{B}\int_0^1 x_i  
  \rho_i(x_i) {\dd} F_i^{B}(x_i)=\frac{1}{B}\mathbb{E}\left[\max_{b\in\mathcal{B}}\chi_{b,i}\right]
$$
and conclude that~\eqref{rhozetB_appendix} is equal to~\eqref{eq_full_surplus}.

\section{Numerical approach}\label{app_numerical}
This section is devoted to computing  the  auctioneer's optimal revenue and an optimal reduced-form mechanism.
By Proposition~\ref{prop_Rochet}, the auctioneer's problem is equivalent to the multi-bidder Rochet-Chon\'e problem~\eqref{eq_Rochet_Chone_extension}. We describe  a numerical approximation scheme for this problem and
prove convergence results.

Recall that in the multi-bidder Rochet-Chon\'e problem, we are given the number~$B$ of bidders, the set $\mathcal{I}$ of $|\mathcal{I}|=I$ items, and a distribution $\mu$ on the set of types $X=[0,1]^\mathcal{I}$ with density $\rho$. Let $\eta$ be the majorizing measure equal to the distribution of $\xi^{B-1}$, where $\xi$ is uniform on $[0,1]$.  The goal is to maximize
$$
B\cdot \int_{X} \Big(\langle \nabla u(x), x\rangle -u(x)\Big)  {\dd} \mu(x)
$$
over functions $u \in \mathcal{U}_{{\lip}, 1}$  satisfying the majorization constraint
$$
\mathrm{law}(u_{x_i}) \preceq \eta,
$$
for all $i\in \mathcal{I}$. Recall that $\mathcal{U}_{{\lip}, 1}$ is the set of $1$-Lipshitz convex non-decreasing functions  $u:\ X\to\R_+$, $\mathrm{law}(\xi)$ denotes the distribution of a random variable $\xi$, and the partial derivative $u_{x_i}=u_{x_i}(\chi)$ is treated as a random variable assuming that its argument  $\chi$ is distributed according to $\mu$. 

We will assume that the distribution $\mu$ satisfies the following assumption: 
\begin{assumption}
The density $\rho$ is a continuously differentiable function, and there exist constants $0 < c < C$ such that $c \le \rho(x) \le C$ for all $x \in {X}$.
\end{assumption}

\paragraph{Outline of the results.} 
The multi-bidder Rochet-Chon\'e problem is a well-defined optimization problem, however, converting it into an algorithm approximating the solution --- a numerical approximation scheme --- is not straightforward. The first obstacle is that the solution $u$ as well as the input data $\mu$ and $\eta$ are continuous objects. Hence, the problem is to be discretized in a way that solutions of the discrete problems approximate those of the continuous one. The second obstacle is that the majorization constraint, in addition to requiring discretization, is non-linear. 

We demonstrate that the majorization constraint is equivalent to a linear constraint suggested by a connection between majorization and martingales and construct provably convergent approximations. As a result, we obtain a finite-dimensional linear program approximating the original Rochet-Chon\'e problem. To summarize, the approach consists of three steps:
\begin{enumerate}
    \item discretize the set of types $X=[0, 1]^\mathcal{I}$ and the distribution $\mu$;
    \item approximate the gradient $\nabla u$ and the convexity constraint $u \in \mathcal{U}_{{\lip}, 1}$;
    \item linearize and approximate the majorization constraint $\mathrm{law}(u_{x_i}) \preceq \eta$;
    \item use an LP solver to find a solution to the resulting linear program.
\end{enumerate}
Section~\ref{sec_conv_approximation} describes the second step, the third step is discussed in Section~\ref{sec_maj_approximation}, and Section~\ref{sec_empiric_optimization} contains provable heuristics improving the run time. Here we provide a high-level overview.
\smallskip

To discretize the domain ${X}=[0, 1]^\mathcal{I}$, we consider the uniform partition of $X$ into $n^{I}$ equal cubes and replace the probability distribution $\mu$ with the associated sum of point masses. After that, we approximate the initial auction design problem with the corresponding discrete version. The convexity constraint can be written as follows:
\begin{equation}\label{eq_convexity_discrete}
u(\theta_i) - u(\theta_j) \ge \langle \nabla u(\theta_j), \theta_i - \theta_j \rangle \quad \text{for all $\theta_i$, $\theta_j$ from the discrete lattice.}
\end{equation}
This approach is based on the algorithm described by \cite{ekeland2010algorithm}. 
To approximate the majorization constraint, we use a generalization of Strassen's theorem \cite{strassen1965existence} reducing the constraint to the existence of the supermartingale with the given marginals. To get a finite-dimensional linear program, we discretize the distribution $\eta$.
The convergence of the discretization  is demonstrated in Theorem~\ref{thm:discrete_approximation_convergence}, Theorem~\ref{thm:stochastic_dominance_convergence_rate}, and Corollary~\ref{cor:stochastic_dominance_convergence}.

In practice, the computation can be sped up by reducing the size of the linear program, which can be achieved via heuristics identifying redundant constraints. The approach of ``directional convexity'' by \cite{oberman2013numerical} allows us to reduce the number of convexity constraints in~\eqref{eq_convexity_discrete} and results in the substantial improvement in computation time.

For simplicity, we focus on the case of a common majorizing measure $\eta=\mathrm{law}(\xi^{B-1})$ with $\xi \sim \mathrm{Uniform}[0, 1]$. The results can be easily extended to the Rochet-Chon\'e problem with general majorization~\eqref{eq_Rochet_Chone_general_majorization} and distinct majorizing measures $\eta_i$.

\subsection{Convexity constraint approximation}\label{sec_conv_approximation}
With the continuous problem, we associate its discrete version as described in \cite[Section 3]{ekeland2010algorithm}.

Fix a positive integer $n$. We partition the domain $X=[0, 1]^\mathcal{I}$ into $n^{I}$ equal cubes with the edge length $1 / n$. The elements of the partition will be denoted by $\sigma_j^{(n)}$, $1 \le j \le n^{I}$. Denote 
\[
\Theta_n = \{\theta_j^{(n)} \colon 1 \le j \le n^{I}\},
\]
where $\theta_j^{(n)}$ is the center of the cube $\sigma_j^{(n)}$. Finally, we denote 
\[
\mu_j^{(n)} = \frac{1}{n^{I}} \cdot \min\{\rho(x) \colon x \in \sigma_j^{(n)}\}.
\]
Note that the weight sum $\sum_j \mu_k^{(n)}$ is not necessary equal to 1; therefore, we define $\mu_0^{(n)} = 1 - \sum_{j = 1}^{n^{I}}\mu_j^{(n)}$.

For every $1 \le j \le n^{I}$, we associate with the cube $\sigma_j^{(n)}$ the scalar variable $u_j^{(n)}$ that corresponds to the value of the utility function $u$ at $\theta_j^{(n)}$, and the vector variable $p_j^{(n)} = (p_{j, 1}^{(n)}, \dots, p_{j, {I}}^{(n)}) \in \mathbb{R}^\mathcal{I}$ that corresponds to the value of $\nabla u$ at $\theta_j^{(n)}$. After that, we define the following non-linear program $\mathcal{D}_n$:
\begin{align*}
   &\text{maximize:}\quad\sum_{j} \mu_j^{(n)} \cdot \left( \langle \theta_j^{(n)}, p_j^{(n)}\rangle - u_j^{(n)}\right)&& \tag{$\mathcal{D}_n$}\\
    &\text{subject to:}\quad&&\\
    &\quad\text{\textbf{(ir)}}\quad u_j^{(n)} \ge 0 \quad&&\text{for all } 1 \le j \le n^{I};\\
    &\quad\text{\textbf{(fs)}}\quad 0 \le p_{j, k}^{(n)} \le 1 \quad&&\text{for all } 1 \le j \le n^{I}, k\in\mathcal{I};\\
   &\quad\text{\textbf{(ic)}}\quad u_i^{(n)} - u_j^{(n)} \ge \langle \theta_i^{(n)} - \theta_j^{(n)}, p_j^{(n)} \rangle\quad&&\text{for all } 1 \le i, j \le n^{I};\\
    &\quad\text{\textbf{(mj)}}\quad \mu_0^{(n)} \cdot \delta(x=0) + \sum_{j} \mu_j^{(n)} \cdot \delta\left(x = p_{j, k}^{(n)}\right) \preceq \eta\quad&& \text{for all } k\in\mathcal{I}.
\end{align*}
Here, the shortcuts \textbf{(ir)}, \textbf{(fs)}, \textbf{(ic)}, and \textbf{(mj)} correspond to the individual rationality, feasibility, incentive compatibility, and majorization, respectively. The only non-linear constraint in this program is \textbf{(mj)} discussed in the next section.

Given a solution $(\overline{u}^{(n)}, \overline{p}^{(n)})$, we define a function
\[
\overline{u}^{(n)}(x) = \max\left\{0,\, \max_j \left\{u_j^{(n)} + \langle x - \theta_j^{(n)}, p_j^{(n)}\rangle\right\}\right\}.
\]
One can easily check that $\overline{u}^{(n)} \in \mathcal{U}_{{\lip}, 1}$, $\overline{u}^{(n)}\big(\theta_j^{(n)}\big) = \overline{u}^{(n)}_j$, and $\overline{p}_j^{(n)} \in \partial \overline{u}^{(n)}\big(\theta_j^{(n)}\big)$ for all $1 \le j \le n^{I}$. Unfortunately, it does not necessary true that $\mathrm{law}(\overline{u}^{(n)}_{x_i}) \preceq \eta$; therefore, the function $\overline{u}^{(n)}$ does not necessary correspond to the interim utility function of a feasible auction mechanism. Nevertheless, we prove that the limiting function $\overline{u}$ satisfies the majorization constraint:
\begin{proposition}\label{prop:u_overline_partial_convergence}
    There exists a subsequence $\{\overline{u}_{n_k}\} \subset \{\overline{u}_n\}$ such that:
    \begin{enumerate}[label=\upshape{(\alph*)}]
    \item the subsequence $\{\overline{u}_{n_k}\}$ converges uniformly to $\overline{u}$;
    \item the function $\overline{u} \in \mathcal{U}_{{\lip}, 1}$ and $\mathrm{law}(\overline{u}_{x_i}) \preceq \eta$ for all $i\in\mathcal{I}$;
    \item \(\displaystyle\lim_{k \to \infty}\sum_{j} \mu_j^{(n)} \cdot \left( \langle \theta_j^{(n_k)}, \overline{p}_j^{(n_k)}\rangle - \overline{u}_j^{(n_k)}\right) = \int_{{X}} \left(\langle x, \nabla \overline{u}(x)\rangle - \overline{u}(x)\right)\,{\dd}\mu\)
    \end{enumerate}
\end{proposition}
Before proving Proposition~\ref{prop:u_overline_partial_convergence}, we need the following technical result:
\begin{lemma}\label{lem:riemann_approx}
Let $Q$ be a compact subset of $\mathbb{R}^{\mathcal{I}}$. Consider a function $\phi(\theta, z, p) \in C^1({X} \times \mathbb{R} \times Q \to \mathbb{R})$. Let $f_k \colon {X} \to \mathbb{R}$ be a family of convex functions such that $\partial f_k(\theta) \in Q$ for all $\theta \in {X}$, whose uniform limit is $f$. Then
\[
\lim_{k \to \infty} \sum_j \mu_j^{(n)} \cdot \phi\left(\theta_j^{(n)}, f_n(\theta_j^{(n)}), \nabla f_n(\theta_j^{(n)})\right) = \int_{{X}} \phi(\theta, f(\theta), \nabla f(\theta))\,{\dd}\mu.
\]
\end{lemma}
\begin{proof}[Proof of Lemma~\ref{lem:riemann_approx}]
The proof is based on the following result.
\begin{lemma*}[{{\cite{ekeland2010algorithm}, Lemma A.6}}] Under the assumptions of Lemma~\ref{lem:riemann_approx}, for every $\varepsilon > 0$ there exists $K \in \mathbb{N}$ such that
\[
\left|\frac{1}{n^{I}}\sum_j \phi\left(\theta_j^{(n)}, f_n(\theta_j^{(n)}), \nabla f_n(\theta_j^{(n)})\right) - \int_{{X}} \phi(\theta, f_n(\theta), \nabla f_n(\theta))\,{\dd}\theta\right| < \varepsilon
\]
for all $n > K$.
\end{lemma*}

Consider a function $\psi(\theta, z, p) = \rho(\theta) \cdot \phi(\theta, z, p)$, where $\rho$ is the density of $\mu$. The function $\psi$ is continuously differentiable; therefore, for every $\varepsilon > 0$ there exists $K \in \mathbb{N}$ such that
\begin{align}\label{eq:riemann_approx_1}
\begin{split}
   &\left|\frac{1}{n^{I}}\sum_j \psi\left(\theta_j^{(n)}, f_n(\theta_j^{(n)}), \nabla f_n(\theta_j^{(n)})\right) - \int_{{X}} \psi(\theta, f_n(\theta), \nabla f_n(\theta))\,{\dd}\theta\right|= \\
   &=\left|\frac{1}{n^{I}}\sum_j \rho(\theta_j^{(n)}) \cdot \phi\left(\theta_j^{(n)}, f_n(\theta_j^{(n)}), \nabla f_n(\theta_j^{(n)})\right) - \int_{{X}} \phi(\theta, f_n(\theta), \nabla f_n(\theta))\,{\dd}\mu\right| < \varepsilon
   \end{split}
\end{align}
for all $n > K$.

Since all $f_n$ are convex and $\{f_n\}$ converges uniformly to $f$, the sequence of gradients $\{\nabla f_n(\theta)\}$ converges to $\nabla f(\theta)$ for almost all $\theta \in {X}$. Thus it follows from the continuity of $\phi$ that
\[
\lim_{n \to \infty}\phi(\theta, f_n(\theta), \nabla f_n(\theta)) = \phi(\theta, f(\theta), \nabla f(\theta))
\]
for almost all $\theta \in {X}$. We may assume that the family of functions $f_n$ is uniformly bounded. Therefore, there exists a constant $M$ such that
\[
\left|\phi(\theta, f_n(\theta), \nabla f_n(\theta))\right| < M
\]
for all $\theta \in {X}$ and for all $n$. Thus it follows from Lebesgue's dominated convergence theorem that
\begin{equation}\label{eq:riemann_approx_2}
\lim_{n \to \infty}\int_{{X}} \phi(\theta, f_n(\theta), \nabla f_n(\theta)) \,{\dd}\mu = \int_{{X}} \phi(\theta, f(\theta), \nabla f(\theta)) \,{\dd}\mu.
\end{equation}

Finally, we have
\begin{align}\label{eq:riemann_approx_3}
\begin{split}
&\Bigg|\frac{1}{n^{I}}\sum_j \rho(\theta_j^{(n)}) \cdot \phi\left(\theta_j^{(n)}, f_n(\theta_j^{(n)}), \nabla f_n(\theta_j^{(n)})\right)- \\
&\quad\quad\quad\quad- \sum_j \mu_j^{(n)} \cdot \phi\left(\theta_j^{(n)}, f_n(\theta_j^{(n)}), \nabla f_n(\theta_j^{(n)})\right)\Bigg|\le\\
&\quad\le M \cdot \sum_j \left|\frac{\rho(\theta_j^{(n)})}{n^{I}} - \mu_j^{(n)}\right| \le M \cdot \sup_{j;\,\,x, y \in \sigma_j^{(n)}}|\rho(x) - \rho(y)|,
\end{split}
\end{align}
and the latter expression tends to 0 as $n \to \infty$ by the uniform continuity of $\rho$.

Combining the expressions~\eqref{eq:riemann_approx_1},   \eqref{eq:riemann_approx_2} and \eqref{eq:riemann_approx_3}, we obtain the desired convergence equality.
\end{proof}
\begin{proof}[Proof of Proposition~\ref{prop:u_overline_partial_convergence}]
Since $\overline{u}^{(n)} \in \mathcal{U}_{{\lip}, 1}$ and this set space is sequentially compact in the uniform convergence topology (Lemma~\ref{lm_compactness}), there is a subsequence $\{\overline{u}_{n_k}\}$ that converges uniformly to the function $\overline{u} \in \mathcal{U}_{{\lip}, 1}$.

To prove the majorization condition, it is sufficient to check that for any continuously differentiable non-decreasing convex function $\varphi$, we have
\begin{equation}\label{eq:phi_riemann_approx}
\int \varphi(\overline{u}_{x_i})\,{\dd}\mu \le \int_0^1 \varphi(x)\,{\dd}\eta(x).
\end{equation}
It follows from Lemma~\ref{lem:riemann_approx} that
\[
\left|\int \varphi(\overline{u}_{x_i})\,{\dd}\mu - \sum_j \mu_j^{(n_k)} \cdot \varphi\big(\overline{p}_{j, i}^{(n_k)}\big)\right| \le \varepsilon(k)
\]
and $\varepsilon(k) \to 0$ as $k \to +\infty$. It follows from the \textbf{(mj)} constraint that
\[
\sum_j \mu_j^{(n_k)} \cdot \varphi\big(\overline{p}_{j, i}^{(n_k)}\big) \le \int_0^1 \varphi(x)\,{\dd}\eta(x);
\]
therefore, letting $k$ tend to $+\infty$, we conclude that the inequality~\eqref{eq:phi_riemann_approx} holds. The point (c) also follows directly from Lemma~\ref{lem:riemann_approx}. 
\end{proof}

Let $u^{\opt} \in \mathcal{U}_{{\lip}, 1}$ be the optimum of the multi-bidder Rochet-Chon\'e problem. For each positive integer $n$, denote
\begin{align*}
&u_j^{\opt,(n)} = n^{I} \cdot \int_{\sigma_j^{(n)}} u^{\opt}(x)\,{\dd} x,\\
&p_{j, k}^{\opt,(n)} = n^{I} \cdot \int_{\sigma_j^{(n)}} u^{\opt}_{x_k}(x)\,{\dd} x.
\end{align*}
\begin{proposition}
The variables $u_j^{\opt,(n)}$ and $ p_j^{\opt,(n)}$ satisfy all the constraints of the program~$\mathcal{D}_n$.
\end{proposition}
\begin{proof}
The constraints \textbf{(ir)} and \textbf{(fs)} follow from the inequalities $u^{\opt}(x) \ge 0$ and $0 
\le u^{\opt}_{x_k}(x) \le 1$.

Since $u^{\opt}(x)$ is convex, we have
\[
u^{\opt}\big(x - \theta_j^{(n)} + \theta_i^{(n)}\big) - u^{\opt}(x) \ge \langle \theta_i^{(n)} - \theta_j^{(n)}, \nabla u^{\opt}(x) \rangle
\]
for almost all $x \in \sigma_j^{(n)}$. Integrating this inequality over the cube $\sigma_j^{(n)}$, we conclude that
\[
u_i^{\opt,(n)} - u_j^{\opt,(n)} \ge \langle \theta_i^{(n)} - \theta_j^{(n)}, p_j^{\opt,(n)}\rangle.
\]
Thus the constraint \textbf{(ic)} holds.

Consider any non-decreasing convex function function $\varphi$. Since $\mathrm{law}(u_{x_i}^{\opt}) \preceq \eta_i$, we conclude that
\[
\int_0^1 \varphi(x)\,{\dd}\eta(x) \ge \sum_{j}\int_{\sigma_j^{(n)}}\varphi(u_{x_i}^{\opt})\,{\dd}\mu.
\]
Since $\varphi(x)$ is non-decreasing, we have
\[
\int_{\sigma_j^{(n)}}\varphi(u_{x_i}^{\opt})\,{\dd}\mu \ge \mu_j^{(n)} \cdot n^{I} \cdot \int_{\sigma_j^{(n)}}\varphi(u_{x_i}^{\opt})\,{\dd} x + \left(\mu(\sigma_j^{(n)}) - \mu_j^{(n)}\right) \cdot \varphi(0).
\]
Finally, it follows from Jensen's inequality that
\[
n^{I} \cdot \int_{\sigma_j^{(n)}}\varphi(u_{x_i}^{\opt})\,{\dd} x \ge \varphi\left(n^{I} \cdot \int_{\sigma_j^{(n)}} u_{x_i}^{\opt}(x)\,{\dd} x\right) = \varphi\big(p_{j, i}^{\opt,(n)}\big)
\]

Thus
\[
\int_0^1 \varphi(x)\,{\dd}\eta(x) \ge \sum_j \mu_j^{(n)} \cdot \varphi\big({p}_{j, i}^{\opt,(n)}\big) + \mu_0^{(n)} \cdot \varphi(0).
\]
Since this inequality holds for all $\varphi$, we conclude that the constraint \textbf{(mj)} holds.
\end{proof}
The next result demonstrates that the optimal revenue in the continuous problem is approximated by its discretization.   
\begin{proposition}\label{prop:riemann_approx_to_tilde_u_converges} The following identity holds:
$$ \lim_{n \to \infty}\sum_{j} \mu_j^{(n)} \cdot \left( \langle \theta_j^{(n)}, p_j^{\opt,(n)}\rangle - u_j^{\opt,(n)}\right) = \int_{{X}} \left(\langle x, \nabla u^{\opt}(x)\rangle - u^{\opt}(x)\right)\,{\dd}\mu.$$
\end{proposition}
\begin{proof}
 The definition of $u_j^{\opt,(n)}$ and $p_j^{\opt,(n)}$ implies that
\begin{multline*}
\int_{{X}} \left(\langle x, \nabla u^{\opt}(x)\rangle - u^{\opt}(x)\right)\,{\dd}\mu - \sum_{j} \mu_j^{(n)} \cdot \left( \langle \theta_j^{(n)}, p_j^{\opt,(n)}\rangle - u_j^{\opt,(n)}\right)= \\
= \sum_{j}\int_{\sigma_j^{(n)}}\left(\langle x, \nabla u^{\opt}(x)\rangle - u^{\opt}(x)\right) \cdot \left(\rho(x) - n^{I}\mu_j^{(n)}\right)\,{\dd} x\leq \\
\le \sup_{x \in {X}} |\langle x, \nabla u^{\opt}(x)\rangle - u^{\opt}(x)| \cdot \sup_{j, x \in \sigma_j^{(n)}}\left|\rho(x) - n^{I} \mu_j^{(n)}\right|.
\end{multline*}
The function $|\langle x, \nabla u^{\opt}(x)\rangle - u^{\opt}(x)|$ is bounded, and the result follows from the uniform continuity of $\rho$.
\end{proof}
Putting all the pierces together, we obtain the following convergence result.
\begin{theorem}\label{thm:discrete_approximation_convergence}
\,
\begin{enumerate}[label=\upshape{(\alph*)}]
\item The function $\overline{u} = \lim_{k \to \infty} \overline{u}^{(n_k)}$ is a solution to the multi-bidder Rochet-Chon\'e problem.
\item \(\displaystyle\lim_{k \to \infty}\sum_{j} \mu_j^{(n)} \cdot \left( \langle \theta_j^{(n_k)}, \overline{p}_j^{(n_k)}\rangle - \overline{u}_j^{(n_k)}\right) = \max_{\substack{u \in \mathcal{U}_{{\lip}, 1}, \\ \mathrm{law}(u_{x_i}) \preceq \eta}} \int_{{X}} \left(\langle x, \nabla u(x)\rangle - u(x)\right)\,{\dd}\mu\).
\end{enumerate}
\end{theorem}
\begin{proof}
Recall that $u^{\opt}(x)$ is a solution to the multi-bidder Rochet-Chon\'e problem. Hence, it follows from Proposition~\ref{prop:riemann_approx_to_tilde_u_converges} that
\begin{align*}
\max_{\substack{u \in \mathcal{U}_{{\lip}, 1}, \\ \mathrm{law}(u_{x_i}) \preceq \eta}} \int_{{X}} \left(\langle x, \nabla u(x)\rangle - u(x)\right)\,{\dd}\mu &= \int_{{X}} \left(\langle x, \nabla u^{\opt}(x)\rangle - u^{\opt}(x)\right)\,{\dd}\mu= \\
&= \lim_{n \to \infty}\sum_{j} \mu_j^{(n)} \cdot \left( \langle \theta_j^{(n)}, p_j^{\opt,(n)}\rangle - u_j^{\opt,(n)}\right).
\end{align*}

Since $(u_j^{\opt,(n)}, p_j^{\opt,(n)})$ satisfies all the constraints of the problem $\mathcal{D}_n$ and $(\overline{u}_j^{(n)}, \overline{p}_j^{(n)})$ is an optimal solution to this problem, we have
\begin{align*}
&\sum_j\mu_j^{(n)} \cdot \left( \langle \theta_j^{(n)}, p_j^{\opt,(n)}\rangle - u_j^{\opt,(n)}\right) \le \sum_j\mu_j^{(n)} \cdot \left( \langle \theta_j^{(n)}, \overline{p}_j^{(n)}\rangle - \overline{u}_j^{(n)}\right)\\
&\quad\quad\lim_{n \to \infty}\sum_j\mu_j^{(n)} \cdot \left( \langle \theta_j^{(n)}, p_j^{\opt,(n)}\rangle - u_j^{\opt,(n)}\right) \le \lim_{n \to \infty}\sum_j\mu_j^{(n)} \cdot \left( \langle \theta_j^{(n)}, \overline{p}_j^{(n)}\rangle - \overline{u}_j^{(n)}\right).
\end{align*}
By Proposition~\ref{prop:u_overline_partial_convergence},
\[
\lim_{n \to \infty}\sum_j\mu_j^{(n)} \cdot \left( \langle \theta_j^{(n)}, \overline{p}_j^{(n)}\rangle - \overline{u}_j^{(n)}\right) = \int_{{X}} \left(\langle x, \nabla \overline{u}(x)\rangle - \overline{u}(x)\right)\,{\dd}\mu.
\]
Thus
\[
\int_{{X}} \left(\langle x, \nabla \overline{u}(x)\rangle - \overline{u}(x)\right)\,{\dd}\mu \ge \max_{\substack{u \in \mathcal{U}_{{\lip}, 1}, \\ \mathrm{law}(u_{x_i}) \preceq \eta}} \int_{{X}} \left(\langle x, \nabla u(x)\rangle - u(x)\right)\,{\dd}\mu.
\]
At the same time, by Proposition~\ref{prop:u_overline_partial_convergence}, we have $\overline{u} \in \mathcal{U}_{{\lip}, 1}$ and $\mathrm{law}(\overline{u}_{x_i}) \preceq \eta$ for all $i\in\mathcal{I}$. Thus the equality holds; therefore, $\overline{u}$ is a solution the multi-bidder Rochet-Chon\'e problem and
\[
\lim_{k \to \infty}\sum_{j} \mu_j^{(n)} \cdot \left( \langle \theta_j^{(n_k)}, \overline{p}_j^{(n_k)}\rangle - \overline{u}_j^{(n_k)}\right) = \max_{\substack{u \in \mathcal{U}_{{\lip}, 1}, \\ \mathrm{law}(u_{x_i}) \preceq \eta}} \int_{{X}} \left(\langle x, \nabla u(x)\rangle - u(x)\right)\,{\dd}\mu.
\]
\end{proof}

\subsection{Approximation of the majorization constraints}\label{sec_maj_approximation}

The majorization constraint \textbf{(mj)} is non-linear. An equivalent linear constraint can be obtained using the following characterization of the majorization order.
\begin{theorem}[{{\cite{shaked2007stochastic}, Theorem 4.A.5}}]\label{thm:stochastic_dominance_criterion}
Two random variables $X$ and $Y$ satisfy $X \preceq Y$ if and only if there exist two random variables $\widehat{X}$ and $\widehat{Y}$ defined on the same probability space such that
\begin{align*}
    &\mathrm{law}(X) = \mathrm{law}(\widehat{X}),\\
    &\mathrm{law}(Y) = \mathrm{law}(\widehat{Y}),
\end{align*}
and $\{\widehat{X}, \widehat{Y}\}$ is a supermartingale, that is,
\[
\mathbb{E}\left[\widehat{Y} \mid \widehat{X}\right] \ge \widehat{X} \quad \text{almost surely.}
\]
\end{theorem}
 Using this criterion, we reformulate the \textbf{(mj)}-constraints as a condition of the existence of the joint distribution of   $\widehat{X}$ and $\widehat{Y}$. In what follows, we fix an item $k\in\mathcal{I}$.
\begin{proposition}\label{prop:stochastic_dominans criterions}
Denote $X = T = [0, 1]$, $J = \{1, 2, \dots, n^{I}\}$, and $\overline{J} = \{0\} \cup J$. The following statements are equivalent:
\begin{enumerate}[label=\upshape{(\alph*)}]
    \item The majorization condition \(
\mu_0^{(n)} \cdot \delta(x=0) + \sum_{j} \mu_j^{(n)} \cdot \delta\left(x = p_{j, k}^{(n)}\right) \preceq \eta \) holds.
\item There exists a probability distribution $\pi$ concentrated on $X \times T$ such that
\begin{align*}
&\mathrm{pr}_T \pi = \eta,\\
&\mathrm{pr}_X \pi = \mu_0^{(n)} \cdot \delta(x=0) + \sum_{j} \mu_j^{(n)} \cdot \delta\left(x = p_{j, k}^{(n)}\right),\\
&\int_T t \cdot {\dd}\pi(x, t) \ge x \cdot \mathrm{pr}_X \pi\big(\{x\}\big)\quad \text{for all } x \in X,
\end{align*}
where $\mathrm{pr}_X \pi$ and $\mathrm{pr}_T \pi$ denotes the marginals of $\pi$ on $X$ and $T$, respectively.
\item There exists a probability distribution $\pi$ concentrated on $\overline{J} \times T$ such that
\begin{align*}
    &\mathrm{pr}_T \pi = \eta,\\
    &\mathrm{pr}_{\overline{J}} \pi(j) = \mu_j^{(n)} \quad \text{for all } j \in \overline{J},\\
    &\int_T t \cdot {\dd}\pi(j, t) \ge p_{j, k}^{(n)} \cdot \mu_j^{(n)} \quad \text{for all } j \in J.
\end{align*}
\item There exists a (not necessary probability) measure $\pi$ concentrated on $J \times T$ such that
\begin{align*}
    &\mathrm{pr}_T \pi \le \eta,\\
    &\mathrm{pr}_J \pi (j) \le \mu_j^{(n)} \quad \text{for all } j \in J,\\
    &\int_T t \cdot {\dd}\pi(j, t) \ge p_{j, k}^{(n)} \cdot \mu_j^{(n)}\quad\text{for all } j \in J.
\end{align*}
\end{enumerate}
\end{proposition}
\begin{proof}
The equivalence (a) $\Leftrightarrow$ (b) is a reformulation of Theorem~\ref{thm:stochastic_dominance_criterion}. The distribution $\pi$ can be considered as the joint law of $\widehat{X}$ and $\widehat{Y}$.

(c) $\Rightarrow$ (b). Let $\pi$ be a distribution satisfying all the conditions of (c). Consider a mapping $f \colon \overline{J} \to X$ defined as $f(0) = 0$ and $f(j) = p_{j, k}^{(n)}$ for all $j \in J$. Define by $\widehat{\pi}$ the pushforward measure $f_{\#}\pi$ concentrated on $X \times T$. It follows directly from the construction that
\[\mathrm{pr}_T \widehat{\pi} = \eta \quad \text{and}\quad \mathrm{pr}_X \widehat{\pi} = \mu_0^{(n)} \cdot \delta(x=0) + \sum_{j} \mu_j^{(n)} \cdot \delta\left(x = p_{j, k}^{(n)}\right).
\]
Finally, we need to check the inequality $\int_T t \cdot {\dd}\pi(x, t) \ge x \cdot \mathrm{pr}_X \pi\big(\{x\}\big)$. If $x = 0$, there is nothing to prove. Otherwise,
\[
\int_T t {\dd}\widehat{\pi}(x,t) = \sum_{j \in J \colon f(j) = x} \int_T t {\dd}\pi(j, t) \ge \sum_{j \in J \colon p_{j, k}^{(n)} = x} p_{j, k}^{(n)} \cdot \mu_j^{(n)} = x \cdot \mathrm{pr}_X \pi\big(\{x\}\big).
\]
Thus $\widehat{\pi}$ satisfies all the restrictions of (b).

(b) $\Rightarrow$ (c). Let $\pi$ be a distribution satisfying all the conditions of (b). For each $j \in J$, define
\begin{align*}
&\pi_j = \frac{\mu_j^{(n)}}{\mathrm{pr}_X \pi\big(\{p_{j, k}^{(n)}\}\big)} \cdot \pi|_{x = p_{j, k}^{(n)}},
\intertext{and}
&\pi_0 = \frac{\mu_0^{(n)}}{\mathrm{pr}_X \pi\big(\{0\}\big)} \cdot \pi|_{x = 0}.
\end{align*}
One can check easily that $\pi = \delta_0 \otimes \pi_0 + \sum_{j \in J} \delta_{p_{j, k}^{(n)}} \otimes \pi_j$, where $\delta_x$ is the Dirac delta measure concentrated at a point $x$.

Define a measure $\widehat{\pi} = \sum_{j \in \overline{J}} \delta_j \otimes \pi_j$ concentrated on $\overline{J} \times T$. We have
\begin{align*}
    &\mathrm{pr}_T \widehat{\pi} = \sum_{j \in \overline{J}} \mathrm{pr}_T \pi_j = \mathrm{pr}_T \pi = \eta,\\
    &\mathrm{pr}_{\overline{J}} \widehat{\pi}(j) = |\pi_j| = \mu_j^{(n)} \quad \text{for all } j \in \overline{J},\\
    &\int_T t \cdot {\dd}\widehat{\pi}(j, t) = \int_T t \cdot {\dd}\pi_j(t)= \\
    &\quad\quad= \frac{\mu_j^{(n)}}{\mathrm{pr}_X\big(\{p_{j, k}^{(n)}\}\big)} \int_T t \cdot {\dd} \pi\big(p_{j, k}^{(n)}, t\big) \ge \mu_j^{(n)} \cdot p_{j, k}^{(n)} \quad\text{for all } j \in J.
\end{align*}
Thus $\widehat{\pi}$ satisfies all the restrictions of (c).

(c) $\Rightarrow$ (d). If a distribution $\pi$ satisfies all the restrictions of (c), then the restriction of $\pi$ to the set $J \times T$ satisfies all the restrictions of (d). 

(d) $\Rightarrow$ (c). Let $\pi$ be a measure concentrated on $J \times T \subset \overline{J} \times T$ satisfying all the restrictions of (d). One can easily prove that there exists a distribution $\widehat{\pi}$ concentrated on $\overline{J} \times T$ such that
$\pi \le \widehat{\pi}$, $\mathrm{pr}_T \widehat{\pi} = \eta$, and $\mathrm{pr}_{\overline{J}}\widehat{\pi}(j) = \mu_j^{(n)}$ for all $j \in \overline{J}$. As a consequence,
\[
\int_T t {\dd}\widehat{\pi}(j, t) \ge \int_T t {\dd}\pi(j, t) \ge \mu_j^{(n)} \cdot p_{j, k}^{(n)} \quad\text{for all } j \in J.
\]
Thus the distribution $\widehat{\pi}$ satisfies all the restrictions of (c).
\end{proof}

The measure $\pi$ obtained in Proposition~\ref{prop:stochastic_dominans criterions}(d) is not discrete. To discretize this measure, we discretize the space $T$. Let $0 = q_0 < q_1 < \dots < q_M = 1$ be any partition of the space $T = [0, 1]$. For each $1 \le m \le M$, denote
\begin{align*}
    w_m = \eta\big([q_{m - 1}, q_m]\big), \quad t_m = \frac{1}{w_m}\int_{q_{m - 1}}^{q_m}t\,{\dd}\eta, \quad \eta_m = \frac{1}{w_m} \cdot \eta|_{[q_{m - 1}, q_m]}.
\end{align*}
As a discrete approximation of $\pi$, we will only consider measures of the form
\begin{equation}\label{eq:pi_approx_def}
    \pi = \sum_{1 \le j \le n^{I}, 1 \le m \le M} \pi_{j, m} \cdot \delta_j \otimes \eta_m
\end{equation}
for some non-negative coefficients $\pi_{j, m}$. The following statement characterizes all such measures that satisfy the restrictions of Proposition~\ref{prop:stochastic_dominans criterions}(d).

\begin{lemma}\label{lem:discrete_stochastic_dominance_criterion}
The measure $\pi$ defined in~\eqref{eq:pi_approx_def} satisfies all the restrictions of Proposition~\ref{prop:stochastic_dominans criterions}(d) if and only if the following inequalities hold:
\begin{align*}
    &\sum_{1 \le j \le n^{I}} \pi_{j, m} \le w_m &&\text{for all }\quad 1 \le m \le M,\\
    &\sum_{1 \le m \le M} \pi_{j, m} \le \mu_j^{(n)} &&\text{for all }\quad 1 \le j \le n^{I},\\
    &\sum_{1 \le m \le M} t_m \cdot \pi_{j, m} \ge p_{j, k}^{(n)} \cdot \mu_j^{(n)} && \text{for all }\quad 1 \le j \le n^{I}.
\end{align*}
\end{lemma}
This suggests considering the following linear problem.
\begin{definition}
Given a partition $0 = q_0 < q_1 < \dots < q_M = 1$, consider the following linear problem $\mathcal{D}_{n, M}$:
\begin{align*}
    \text{maximize:}&\quad\sum_{j} \mu_j^{(n)} \cdot \left( \langle \theta_j^{(n)}, p_j^{(n)}\rangle - u_j^{(n)}\right)&&\tag{$\mathcal{D}_{n, M}$}\\
    \text{subject to:}&\quad&&\\
    \text{\textbf{(ir)}}&\quad u_j^{(n)} \ge 0 \quad&&\text{for all } 1 \le j \le n^{I};\\
    \text{\textbf{(fs)}}&\quad 0 \le p_{j, k}^{(n)} \le 1 \quad&&\text{for all } 1 \le j \le n^{I}, k\in\mathcal{I};\\
    \text{\textbf{(ic)}}&\quad u_i^{(n)} - u_j^{(n)} \ge \langle \theta_i^{(n)} - \theta_j^{(n)}, p_j^{(n)} \rangle\quad&&\text{for all } 1 \le i, j \le n^{I};\\
    \text{\textbf{(mj-T)}}&\quad \sum_{1 \le j \le n^{I}} \pi_{j, m, k}^{(n)} \le w_m&& \text{for all } 1 \le m \le M, k\in\mathcal{I};\\
    \text{\textbf{(mj-J)}}&\quad \sum_{1 \le m \le M} \pi_{j, m, k}^{(n)} \le \mu_j^{(n)}&& \text{for all } 1 \le j \le n^{I}, k\in\mathcal{I};\\
    \text{\textbf{(mj-E)}}&\quad \sum_{1 \le m \le M} t_m \cdot \pi_{j, m, k}^{(n)} \ge p_{j, k}^{(n)} \cdot \mu_j^{(n)}&& \text{for all } 1 \le j \le n^{I}, k\in\mathcal{I};\\
    \text{\textbf{(mj-P)}}&\quad\pi^{(n)}_{j, m, k} \ge 0 && \text{for all }j, m, k.
\end{align*}
\end{definition}

A direct consequence of Proposition~\ref{prop:stochastic_dominans criterions} and Lemma~\ref{lem:discrete_stochastic_dominance_criterion} is the following connection between the problems $\mathcal{D}_n$ and $\mathcal{D}_{n, M}$.
\begin{corollary}\label{cor:D_nM_more_restrictive_than_D_n}
If $(u_j^{(n)}, p_j^{(n)}, \pi_{j, m, k}^{(n)})$ satisfies all the constraints of the problem $\mathcal{D}_{n, M}$, then $(u_j^{(n)}, p_j^{(n)})$ satisfies all the constraints of $\mathcal{D}_n$.
\end{corollary}

The constraint \textbf{(fs)} partially follows from the constraints \textbf{(mj-J)},  \textbf{(mj-E)}, and \textbf{(mj-P)}:
\begin{align*}
p_{j, k}^{(n)} \cdot \mu_j^{(n)} \le \sum_{1 \le m \le M} t_m \cdot \pi_{j, m, k}^{(n)} \le \max_{1 \le m \le M} t_m \cdot \sum_{1 \le m \le M} \pi_{j, m, k}^{(n)} \le t_M \cdot \mu_j^{(n)} \\
\Rightarrow \quad p_{j, k}^{(n)} \le t_M \le 1.
\end{align*}
Thus the linear problem $\mathcal{D}_{n, M}$ is equivalent to the problem $\mathcal{D}_{n, M}'$, where the \textbf{(fs)}-constraints are replaced with the following:
\[
\text{\textbf{(fs')}} \quad p_{j, k}^{(n)} \ge 0 \quad \text{for all } 1 \le j \le n^{I}, k\in\mathcal{I}.
\]

Our goal is to prove that the sequence $\big(\overline{u}_j^{(n, M)}, \overline{p}_j^{(n, M)}, \overline{\pi}_{j, m, k}^{(n, M)}\big)$ of optimal solutions to the problem $\mathcal{D}_{n, M}$ contains a maximizing subsequence to the problem $\mathcal{D}_{n}$ as $M \to \infty$. In order to do it, we formulate a dual problem to $\mathcal{D}'_{n, M}$.

\begin{lemma}\label{lem:alternative_dual}
Consider the following finite-dimensional convex programs:
\begin{align*}
    (\mathcal{P}) \text{max}&\quad c^Tx & (\mathcal{D})\text{ min}&\quad b^Ty_1 & (\mathcal{C}) \text{ min}&\quad b^Ty &\\
    \text{s.t.}&\quad Ax \le b, &\text{s.t.}&\quad A^Ty_1 + Q^Ty_2 \ge c & \text{s.t.}&\quad y^TAx \ge c^Tx\\
    &\quad Qx \le 0, && \quad y_1 \ge 0, y_2 \ge 0; &&\quad\forall x \ge 0\colon Qx \le 0,\\
    &\quad x \ge 0; &&&&\quad y \ge 0.
\end{align*}
Assume that the problem $(\mathcal{P})$ is feasible and bounded. Then 
\begin{enumerate}[label=\upshape{(\alph*)}]
\item if $(y_1, y_2)$ satisfies all the restrictions of $(\mathcal{D})$, then $y_1$ satisfies all the restrictions of $(\mathcal{C})$;
\item if $(\overline{y}_1, \overline{y}_2)$ solves the problem $(\mathcal{D})$, then $\overline{y}_1$ solves the problem $(\mathcal{C})$;
\item the strong duality holds: \( \max_{\mathcal{P}} c^Tx = \min_{\mathcal{D}}b^Ty_1 = \min_{\mathcal{C}}b^Ty\).
\end{enumerate}
\end{lemma}
\begin{proof}
(a) Consider any couple of non-negative vectors $(y_1, y_2)$ satisfying the inequality 
\[
A^T y_1 + Q^T y_2 \ge c \quad \Leftrightarrow \quad y_1^TA + y_2^TQ \ge c^T.
\]
Then, for any vector $x \ge 0$, we have
\[
y_1^TAx + y_2^TQx \ge c^Tx.
\]
Assume in addition that $Qx \le 0$. Then it follows from the non-negativity of $y_2$ that $y_2^TQx \le 0$; therefore,
\[
y_1^TAx \ge y_1^TAx + y_2^TQx \ge c^Tx.
\]
So, the vector $y_1$ satisfies all the restrictions of the problem $(\mathcal{C})$.

(b) and (c). First, we prove the weak duality
\[
\min_{\mathcal{C}}b^Ty \ge \max_{\mathcal{P}} c^Tx.
\]
Let $\overline{x}$ be a solution to $(\mathcal{P})$ and let $y$ be any vector satisfying all the restrictions of the problem $(\mathcal{C})$. We have
\[
y^TA \overline{x} \ge c^T\overline{x}.
\] Since $A \overline{x} \le b$ and $y \ge 0$, we conclude that $y^Tb \ge y^TA\overline{x}$. Thus
\[
b^Ty = y^Tb \ge c^T\overline{x} = \max_{\mathcal{P}}c^Tx\quad \Rightarrow \quad \min_{\mathcal{C}} b^Ty \ge \max_{\mathcal{P}}c^Tx.
\]

Let $(\overline{y}_1, \overline{y}_2)$ be a solution to the problem $(\mathcal{D})$. It follows from the duality theorem that $b^T\overline{y}_1 = \max_{\mathcal{P}} c^Tx$. In addition, $\overline{y}_1$ satisfies all the restrictions of the problem $(\mathcal{C})$; therefore,
\[
\max_{\mathcal{P}} c^Tx = b^T\overline{y}_1 = \min_{\mathcal{C}}b^Ty.
\]
Thus $\overline{y}_1$ is a solution to $(\mathcal{C})$ and the strong duality holds.
\end{proof}

Using this lemma, we formulate a dual convex program to $\mathcal{D}'_{n, m}$.
\begin{definition}
Given a positive integer $n$ and a partition $0 = q_0 < q_1 < \dots < q_M = 1$, we define a convex program $\mathcal{D}^*_{n, M}$ as follows:
\begin{align*}
    \text{minimize:}\quad&\mathrlap{\sum_{k\in \mathcal{I}}\left(\sum_{1 \le m \le M} \varphi^{(n)}_{m, k} \cdot w_m + \sum_{1 \le j \le n^{I}} \psi^{(n)}_{j, k} \cdot \mu_j^{(n)}\right)} \tag{$\mathcal{D}^*_{n, M}$}\\
    \text{subject to:}\quad&\varphi^{(n)}_{m, k} \ge 0, \psi^{(n)}_{j, k} \ge 0, c^{(n)}_{j, k} \ge 0 \\
    \textbf{(lt)}\quad&\varphi^{(n)}_{m, k} + \psi^{(n)}_{j, k} \ge t_m \cdot c^{(n)}_{j, k} \\
    \textbf{(c-def)}\quad&\mathrlap{\sum_{j} \mu_j^{(n)} \cdot \left( \langle \theta_j^{(n)}, p_j^{(n)}\rangle - u_j^{(n)}\right) \le \sum_j \mu_j^{(n)} \cdot \langle c_j^{(n)}, p_j^{(n)} \rangle}\\
    &\quad\quad\quad\text{for all $(u_j^{(n)}, p_j^{(n)})$ satisfying \textbf{(ir)}, \textbf{(fs')}, and \textbf{(ic)}.}
\end{align*}
Here, the shortcut \textbf{(lt)} indicates a relation to the Legendre transform, and \textbf{(c-def)}, to the set of vector fields $\mathcal{C}$ from Appendix~\ref{app_Dual} defined by
\[
\int (\langle x, \nabla u(x) \rangle - u(x)){\dd}\mu \le \int \langle c(x), \nabla u(x) \rangle{\dd}\mu
\]
for all  convex non-decreasing $u$ with $u(0)=0$.
\end{definition}
\begin{proposition}\label{prop:dual_D'_nM}
The strong duality holds: $\max \mathcal{D}'_{n, M} = \min \mathcal{D}^*_{n, M}$.
\end{proposition}
\begin{proof}
Let $\varphi_{m, k}^{(n)}$ be a dual variable for the \textbf{(mj-T)}-constraint $\sum_j \pi_{j, m, k}^{(n)} \le w_m$, let $\psi_{j, m}^{(n)}$ be a dual variable for the \textbf{(mj-J)}-constraint $\sum_m \pi_{j, m, k}^{(n)} \le \mu_j^{(n)}$, and let $c_{j, k}^{(n)}$ be a dual variable for the \textbf{(mj-E)}-constraint $p^{(n)}_{j, k} \cdot \mu_j^{(n)} - \sum_m t_m \cdot \pi_{j, m, k}^{(n)} \le 0$. The following duality equation  follows from Lemma~\ref{lem:alternative_dual} applied to the linear program $\mathcal{D}'_{n, M}$:
\[
\max \mathcal{D}'_{n, M} = \min \sum_{k = 1}^\mathcal{I}\left(\sum_{1 \le m \le M} \varphi^{(n)}_{m, k} \cdot w_m + \sum_{1 \le j \le n^{I}} \psi^{(n)}_{j, k} \cdot \mu_j^{(n)}\right),
\]
where minimum in the right-hand side is taken over all non-negative variables $\big(\varphi_{m, k}^{(n)}, \psi_{j, k}^{(n)}, c_{j, k}^{(n)}\big)$ such that the inequality
\begin{multline}\label{eq:dual_D'_nM_ineq_temp}
    \sum_{m, k} \varphi^{(n)}_{m, k} \cdot \sum_j \pi^{(n)}_{j, m, k} + \sum_{j, k} \psi^{(n)}_{j, k} \cdot \sum_m \pi^{(n)}_{j, m, k} + \sum_{j, k} c^{(n)}_{j, k} \cdot \left(p^{(n)}_{j, k} \cdot \mu_j^{(n)} - \sum_m t_m \cdot \pi_{j, m, k}^{(n)}\right)\\
    \ge \sum_{j} \mu_j^{(n)} \cdot \left( \langle \theta_j^{(n)}, p_j^{(n)}\rangle - u_j^{(n)}\right)
\end{multline}
holds for all $\big(u_j^{(n)}, p_j^{(n)}, \pi_{j, m, k}^{(n)}\big)$ satisfying the constraints \textbf{(ir)}, \textbf{(fs')}, and \textbf{(ic)}.

After the rearrangement of the left-hand side terms, inequality~\eqref{eq:dual_D'_nM_ineq_temp} transforms into
\begin{multline}\label{eq:dual_D'_nM_ineq}
\sum_{j, m, k} \pi^{(n)}_{j, m, k} \cdot \left(\varphi_{m, k}^{(n)} + \psi_{j, k}^{(n)} - t_m \cdot c_{j, k}^{(n)}\right) + \sum_{j} \mu_j^{(n)} \cdot \langle p_j^{(n)}, c_j^{(n)} \rangle\geq \\
\ge \sum_{j} \mu_j^{(n)} \cdot \left( \langle \theta_j^{(n)}, p_j^{(n)}\rangle - u_j^{(n)}\right)
\end{multline}
Since the constraints \textbf{(ir)}, \textbf{(fs')}, and \textbf{(ic)} do not contain any restrictions on $\pi_{j, m, k}^{(n)}$, it follows from~\eqref{eq:dual_D'_nM_ineq} that the \textbf{(lt)}-constraints hold:
\[
\varphi_{m, k}^{(n)} + \psi_{j, k}^{(n)} - t_m \cdot c_{j, k}^{(n)} \ge 0\quad\text{for all } 1 \le j \le n^{I}, 1 \le m \le M, k\in\mathcal{I}.
\]

Substituting $\pi_{j, m, k}^{(n)} = 0$ into~\eqref{eq:dual_D'_nM_ineq}, we conclude that \textbf{(c-def)}-constraint holds. Vice versa, if all the \textbf{(lt)}- and \textbf{(c-def)}-constraints hold, then the inequality~\eqref{eq:dual_D'_nM_ineq} holds for all $\big(u_j^{(n)}, p_j^{(n)}\big)$ satisfying the constraints \textbf{(ir)}, \textbf{(fs')}, and \textbf{(ic)}, and for all non-negative $\pi_{j, m, k}^{(n)}$.
\end{proof}

\begin{remark}\label{rem:dual_interpretation}
The problem $\mathcal{D}^*_{n, M}$ can be seen as a discrete approximation of the dual problem described in Theorem~\ref{th_vector_fields_inf_appendix}. The variable $c_{j, k}^{(n)}$ corresponds to the value of the component $c_k$ of the vector field $c = (c_1, \dots, c_I)$ at the point $\theta_j^{(n)}$; the constraint \textbf{(c-def)} is a discrete approximation of the inequality
\[
\int (\langle x, \nabla u(x) \rangle - u(x)){\dd}\mu \le \int \langle c(x), \nabla u(x) \rangle{\dd}\mu.
\]

The variable $\varphi^{(n)}_{m, k}$ corresponds to the value of the function $\varphi_k(x)$ at the point~$t_m$. By the constraint \textbf{(lt)}, the optimal value of $\varphi^{(n)}_{m, k}$ is equal to
\[
\overline{\varphi}^{(n)}_{m, k} = \max\left\{0,\, \max_j\big(t_m \cdot \overline{c}_{j, k}^{(n)} - \overline{\psi}^{(n)}_{j, k}\big)\right\} = \overline{\varphi}_k(t_m),
\]
where the function $\overline{\varphi}(x)$ is convex and non-negative as a maximum of non-decreasing linear functions. Similarly, the optimal value of $\psi^{(n)}_{j, k}$ is equal to $\overline{\varphi}_k^*\big(c_{j, k}^{(n)}\big) \approx \overline{\varphi}_k^*\big(c_k(\theta_j^{(n)})\big)$. The term    $\sum_m \overline{\varphi}_{m, k}^{(n)} \cdot w_k$ in the objective function is an approximation of the integral $\int_0^1 \overline{\varphi}_k(t)\,{\dd}\eta(t)$, and the term $\sum_j \overline{\psi}_{j, k}^{(n)} \cdot \mu_j^{(n)}$ approximates the integral $\int \overline{\varphi}_k^*(c_k(x))\,{\dd}\mu$.
\end{remark}
\begin{remark}
Lemma~\ref{lem:alternative_dual}(b) provides a practical way of solving the problem $\mathcal{D}^*_{n, M}$: we need to solve the problem $\mathcal{D}'_{n, M}$ and extract the optimal values of the dual variables that correspond to the constraints \textbf{(mj-T)}, \textbf{(mj-J)}, and \textbf{(mj-E)}. 
\end{remark}
Next, we formulate a weak duality for the problem $\mathcal{D}_n$.
\begin{proposition}\label{prop:weak_discrete_duality}
Consider a family of functions $\varphi_k(x) \colon [0, 1] \to \mathbb{R}$, $k\in\mathcal{I}$, a family of variables $c_{j, k}^{(n)} \ge 0$ satisfying the constraint \textbf{(c-def)}, and a family of variables $\psi^{(n)}_{j, k}$ satisfying the inequality
\[
\varphi_k(t) + \psi^{(n)}_{j, k} \ge t \cdot c_{j, k}^{(n)}\quad\text{for all }\quad  t \in [0, 1], \ \ 1 \le j \le n^{I}, \ \  k\in\mathcal{I}.
\]
Then
\[
\max \mathcal{D}_n \le \sum_{k = 1}^\mathcal{I}\left(\int_0^1 \varphi_k(x)\,{\dd}\eta(x) + \sum_{j = 1}^{n^{I}}\psi^{(n)}_{j, k} \cdot \mu_j^{(n)}\right).
\]
\end{proposition}
\begin{proof}
Consider any family of variables $\big(u_j^{(n)}, p_j^{(n)}\big)$ satisfying all the constraints of~$\mathcal{D}_{n}$. First, by the \textbf{(c-def)}-constraint, we have
\[
\sum_{j} \mu_j^{(n)} \cdot \left( \langle \theta_j^{(n)}, p_j^{(n)}\rangle - u_j^{(n)}\right) \le \sum_j \mu_j^{(n)} \cdot \langle c_j^{(n)}, p_j^{(n)} \rangle = \sum_{k \in\mathcal{I}}\sum_{j = 1}^{n^{I}}c_{j, k}^{(n)} \cdot p_{j, k}^{(n)} \cdot \mu_j^{(n)}.
\]

Since $\big(u_j^{(n)}, p_j^{(n)}\big)$ satisfy the \textbf{(mj)}-constraint, we can find a family of measures $\pi^{(n)}_k$ satisfying all the constraints of Proposition~ \ref{prop:stochastic_dominans criterions}(d). We have
\[
c_{j, k}^{(n)} \cdot p_{j, k}^{(n)} \cdot \mu_j^{(n)} \le \int_T t \cdot c_{j, k}^{(n)}\,{\dd} \pi_k(j, t) \quad \Rightarrow \quad \sum_{j = 1}^{n^{I}}c_{j, k}^{(n)} \cdot p_{j, k}^{(n)} \cdot \mu_j^{(n)} \le \int_{J \times T}t \cdot c_{j, k}^{(n)}\,{\dd}\pi_k(j, t).
\]
Next, it follows from the inequality $\varphi_k(t) + \psi^{(n)}_{j, k} \ge t \cdot c_{j, k}^{(n)}$ that
\begin{align*}
\int_{J \times T}t \cdot c_{j, k}^{(n)}\,{\dd}\pi_k(j, t) &\le \int_{J \times T}\left(\varphi_k(t) + \psi^{(n)}_{j, k}\right)\,{\dd}\pi_k(j, t)= \\
&= \int_T \varphi_k(t) \,{\dd} \mathrm{pr}_T\pi_k(t) + \int_J \psi^{(n)}_{j, k} {\dd}\mathrm{pr}_J \pi_k(j).
\end{align*}
Finally, since $\mathrm{pr}_T\pi_k \le \eta$ and $\mathrm{pr}_J\pi_k(j) \le \mu_j^{(n)}$,
\[
\int_T \varphi_k(t) \,\mathrm{pr}_T{\dd}\pi_k(t) \le \int_T \varphi_k(t)\,{\dd}\eta(t), \qquad \int_J \psi^{(n)}_{j, k} \mathrm{pr}_J {\dd}\pi_k(j) \le \sum_{j = 1}^{n^{I}} \psi_{j, k}^{(n)} \cdot \mu_j^{(n)}.
\]

Summing up, we conclude that for all $\big(u_j^{(n)}, p_j^{(n)}\big)$ satisfying all the constraints of $\mathcal{D}_{n}$, the following inequality holds:
\begin{multline*}
 \sum_{j} \mu_j^{(n)} \cdot \left( \langle \theta_j^{(n)}, p_j^{(n)}\rangle - u_j^{(n)}\right) \le    \sum_{k = 1}^\mathcal{I}\sum_{j = 1}^{n^{I}}c_{j, k}^{(n)} \cdot p_{j, k}^{(n)} \cdot \mu_j^{(n)}\le \\
 \le \sum_{k\in\mathcal{I}} \int_{J \times T}t \cdot c_{j, k}^{(n)}\,{\dd}\pi_k(j, t) \le \sum_{k\in\mathcal{I}}\left(\int_0^1 \varphi_k(x)\,{\dd}\eta(x) + \sum_{j = 1}^{n^{I}}\psi^{(n)}_{j, k} \cdot \mu_j^{(n)}\right).
\end{multline*}
\end{proof}
\begin{remark}
In fact, the following strong duality holds:
\[
\max \mathcal{D}_n = \min \sum_{k\in\mathcal{I}}\left(\int_0^1 \varphi_k(x)\,{\dd}\eta(x) + \sum_{j = 1}^{n^{I}}\psi^{(n)}_{j, k} \cdot \mu_j^{(n)}\right).
\]
This duality can be proven similarly to Theorem~\ref{th_duality_min_appendix}.
\end{remark}
Finally, using a solution to the dual problem, we can estimate how well the problem $\mathcal{D}_{n, M}$ approximates the problem $\mathcal{D}_n$.
\begin{theorem}\label{thm:stochastic_dominance_convergence_rate}
Denote $\varepsilon = \max_m |q_m - q_{m - 1}|$. Then for all $\varepsilon \le \frac{1}{6}$,
\[
\max \mathcal{D}_{n, M} \le \max \mathcal{D}_n \le (1 + \varepsilon C) \cdot \max \mathcal{D}_{n, M},
\]
where $C$ is a constant that depends only on $\eta$ and is independent of $n$, $\mu$, and ${I}$.
\end{theorem}
\begin{proof}
Consider an optimal solution $\big(\overline{u}_{j}^{(n)}, \overline{p}_j^{(n)}, \overline{\pi}^{(n)}_{j, m, k}\big)$ to the problem $\mathcal{D}_{n, M}$. By Corollary~\ref{cor:D_nM_more_restrictive_than_D_n}, the variables $\big(\overline{u}_{j}^{(n)}, \overline{p}_j^{(n)}\big)$ satisfy all the constraints of the problem $\mathcal{D}_n$; therefore, since the objective functions of $\mathcal{D}_n$ and $\mathcal{D}_{n, M}$ are identical,
\[
\max \mathcal{D}_{n, M} \le \max \mathcal{D}_n.
\]

Let $\big(\overline{\varphi}^{(n)}_{m, k}, \overline{\psi}^{(n)}_{j, k}, \overline{c}^{(n)}_{j, k}\big)$ be an optimal solution to $\mathcal{D}^*_{n, M}$. Denote
\[
R_k = \sum_{m = 1}^M \overline{\varphi}_{m, k}^{(n)} \cdot w_m + \sum_{j = 1}^{n^{I}} \overline{\psi}_{j, k}^{(n)} \cdot \mu_j^{(n)}.
\]
By the strong duality, $\sum_{k\in\mathcal{I}} R_k = \max \mathcal{D}_{n, M}$. By the \textbf{(lt)}-constraint,
\[
\overline{c}^{(n)}_{j, k} \le \frac{1}{t_m}\left(\overline{\psi}^{(n)}_{j, k} + \overline{\varphi}^{(n)}_{m, k}\right);
\]
therefore, for each index $1 \le m \le M$, we have
\begin{equation}
\sum_{j = 1}^{n^{I}} \overline{c}_{j, k}^{(n)} \cdot \mu_j^{(n)} \le \frac{1}{t_m}\left(\sum_{j = 1}^{n^{I}} \overline{\psi}_{j, k}^{(n)} \cdot \mu_j^{(n)} + \overline{\varphi}^{(n)}_{m, k}\right) \le \frac{1}{t_m}\left(R_k + \overline{\varphi}^{(n)}_{m, k}\right).
\end{equation}

As we mentioned in Remark~\ref{rem:dual_interpretation}, the values $\overline{\varphi}^{(n)}_{m, k}$ are equal to the values of the non-decreasing function $\overline{\varphi}_k$ at the points $t_m$; hence, the sequence $\{\overline{\varphi}^{(n)}_{m, k}\}_m$ is non-decreasing, and
\[
R_k \ge \sum_{m = 1}^M \overline{\varphi}^{(n)}_{m, k} \ge \overline{\varphi}^{(n)}_{m, k} \cdot (w_m + w_{m + 1} + \dots + w_M) = \overline{\varphi}^{(n)}_{m, k} \cdot \eta\big([q_{m - 1}, 1]\big).
\]

Choose  an index $m$ such that $q_{m - 1} \le \frac{1}{2} < q_m$. Since $q_m - q_{m - 1} \le \frac{1}{6}$, we can estimate $t_m \ge q_{m - 1} \ge \frac{1}{3}$. In addition, $\eta\big([q_{m - 1}, 1]\big) \ge \eta\big([1/2, 1]\big)$; therefore,
\[
\sum_{j = 1}^{n^{I}} \overline{c}_{j, k}^{(n)} \cdot \mu_j^{(n)} \le \frac{1}{t_m}\left(R_k + \overline{\varphi}^{(n)}_{m, k}\right) \le 3\left(1 + \frac{1}{\eta\big([1/2, 1]\big)}\right) R_k = C \cdot R_k,
\]
where $C$ is a constant that depends only on $\eta$.

For each $k\in\mathcal{I}$, consider a function $\widehat{\varphi}_k \colon [0, 1] \to \mathbb{R}$, $\widehat{\varphi}_k(t) = \overline{\varphi}^{(n)}_{m, k}$ for all $t \in [q_{m - 1}, q_m)$, and $\widehat{\varphi}_k(1) = \overline{\varphi}^{(n)}_{M, k}$. In addition, consider a family of variables $\widehat{\psi}^{(n)}_{j, k}$ defined as follows:
\[
\widehat{\psi}^{(n)}_{j, k} = \overline{\psi}^{(n)}_{j, k} + \varepsilon \cdot \overline{c}_{j, k}^{(n)}.
\]

We claim that the \textbf{(lt)}-constraint holds:
\[
\widehat{\psi}^{(n)}_{j, k} + \widehat{\varphi}_k(t) \ge t \cdot \overline{c}^{(n)}_{j, k}\quad\text{for all } t \in [0, 1], 1 \le j \le n^{I}.
\]
Indeed, for each $t \in [0, 1]$, there exists an index $m$ such that $t_m \ge t - \varepsilon$ and $\widehat{\varphi}_k(t) = \overline{\varphi}^{(n)}_{m, k}$. Hence,
\[
\widehat{\psi}^{(n)}_{j, k} + \widehat{\varphi}_k(t) = \overline{\psi}^{(n)}_{j, k} + \overline{\varphi}^{(n)}_{m, k} + \varepsilon \cdot \overline{c}_{j, k}^{(n)} \ge (t_m + \varepsilon) \cdot \overline{c}_{j, k}^{(n)} \ge t \cdot \overline{c}_{j, k}^{(n)}.
\]
Thus $(\widehat{\varphi}_k, \widehat{\psi}^{(n)}_{j, k}, \overline{c}^{(n)}_{j, k})$ satisfies all the constraints of Proposition~\ref{prop:weak_discrete_duality}; therefore,
\begin{align*}
\max \mathcal{D}_n &\le \sum_{k\in\mathcal{I}}\left(\int_0^1 \widehat{\varphi}_k(t)\,{\dd}\eta(t) + \sum_{j = 1}^{n^{I}} \widehat{\psi}^{(n)}_{j, k} \cdot \mu_j^{(n)}\right) \\
&= \sum_{k\in\mathcal{I}}\left(\sum_{m = 1}^M \overline{\varphi}^{(n)}_{m, k} \cdot w_m + \sum_{j = 1}^{n^{I}} \overline{\psi}^{(n)}_{j, k} \cdot \mu_j^{(n)} + \varepsilon\sum_{j = 1}^{n^{I}} \overline{c}^{(n)}_{j, k} \cdot \mu_j^{(n)}\right)\\
&\le \sum_{k\in\mathcal{I}}\left(R_k + \varepsilon C \cdot R_k\right) = (1 + \varepsilon C) \cdot \max \mathcal{D}_{n, M}.
\end{align*}
\end{proof}

We conclude that the following convergence result holds.
\begin{corollary}\label{cor:stochastic_dominance_convergence}
Let $\big(\overline{u}_j^{(n, M)}, \overline{p}_j^{(n, M)}, \overline{\pi}_{j, m, k}^{(n, M)}\big)$ be an optimal solution to the problem $\mathcal{D}_{n, M}$. Assume that $\max_{1 \le m \le M}|q_{m} - q_{m - 1}| \to 0$ as $M \to \infty$. Then the sequence $\left\{\big(\overline{u}_j^{(n, M)}, \overline{p}_j^{(n, M)}\big)\right\}_M$ contains a subsequence converging to an optimum of $\mathcal{D}_n$.
\end{corollary}

\subsection{Additional empirical optimizations}\label{sec_empiric_optimization}
The total number of \textbf{(ic)}-constraints is $n^{2{I}}$, which constitutes the majority of all the constraints in the problem $\mathcal{D}_n$. We list some heuristics that allow to get rid of the redundant constraints  improving the run time in practice.

\begin{definition}
A couple of points $\theta_i^{(n)}, \theta_j^{(n)} \in \Theta_n$ is called \textit{irreducible} if the interval $(\theta_i^{(n)}, \theta_j^{(n)})$  does not contain any elements of $\Theta_n$ (by the interval $(\theta_i^{(n)}, \theta_j^{(n)})$ we mean the multi-dimensional linear segment with endpoints $\theta_i^{(n)}$ and $\theta_j^{(n)}$).
\end{definition}
The following proposition shows that the incentive compatibility constraints from the problem $\mathcal{D}_n$ can be verified only for irreducible couples of points.
\begin{proposition}\label{prop:irreducibility_heuristic}
Suppose that the inequality
\[
u_i^{(n)} - u_j^{(n)} \ge \langle \theta_i^{(n)} - \theta_j^{(n)}, p_j^{(n)} \rangle
\]
holds for all couples of irreducible couple of points $\theta_i^{(n)}$, $\theta_j^{(n)}$. Then such an inequality holds for all $1 \le i, j \le n^{I}$.
\end{proposition}
\begin{proof}
Assume the converse and choose a pair of points $\theta_i^{(n)}$, $\theta_j^{(n)}$ such that
\[
u_i^{(n)} - u_j^{(n)} < \langle \theta_i^{(n)} - \theta_j^{(n)}, p_j^{(n)} \rangle
\]
and the interval $(\theta_i^{(n)}, \theta_j^{(n)})$ contains the smallest number of elements of $\Theta_n$. Since the couple of points $\theta_i^{(n)}, \theta_j^{(n)}$ is not irreducible, there exists a point $\theta_k^{(n)} \in \Theta_n \cap (\theta_i^{(n)}, \theta_j^{(n)})$. By the construction, the following inequalities hold:
\begin{align}
    &u_i^{(n)} - u_k^{(n)} \ge \langle \theta_i^{(n)} - \theta_k^{(n)}, p_k^{(n)} \rangle,\label{eq:irreducible_1}\\
    &u_k^{(n)} - u_j^{(n)} \ge \langle \theta_k^{(n)} - \theta_j^{(n)}, p_j^{(n)} \rangle,\label{eq:irreducible_2}\\
    &u_j^{(n)} - u_k^{(n)} \ge \langle \theta_j^{(n)} - \theta_k^{(n)}, p_k^{(n)} \rangle.\label{eq:irreducible_3}
\end{align}

Combining inequalities~\eqref{eq:irreducible_2} and~\eqref{eq:irreducible_3}, we conclude that
\[
\langle \theta_j^{(n)} - \theta_k^{(n)}, p_j^{(n)} \rangle \ge \langle \theta_j^{(n)} - \theta_k^{(n)}, p_k^{(n)} \rangle \quad\Rightarrow\quad \langle \theta_j^{(n)} - \theta_k^{(n)}, p_j^{(n)} - p_k^{(n)}\rangle \ge 0.
\]
Since $\theta_k^{(n)} \in (\theta_i^{(n)}, \theta_j^{(n)})$, there exists a constant $\gamma > 0$ such that $\theta_k^{(n)} - \theta_i^{(n)} = \gamma \cdot (\theta_j^{(n)} - \theta_k^{(n)})$; therefore,
\begin{multline*}
\langle \theta_k^{(n)} - \theta_i^{(n)}, p_j^{(n)} - p_k^{(n)}\rangle = \gamma \cdot \langle \theta_j^{(n)} - \theta_k^{(n)}, p_j^{(n)} - p_k^{(n)}\rangle \ge 0 \\ \Rightarrow \quad \langle \theta_i^{(n)} - \theta_k^{(n)}, p_k^{(n)} \rangle \ge \langle \theta_i^{(n)} - \theta_k^{(n)}, p_j^{(n)} \rangle.
\end{multline*}

Thus, by inequality~\eqref{eq:irreducible_1}, we have
\[
u_i^{(n)} - u_k^{(n)} \ge \langle \theta_i^{(n)} - \theta_k^{(n)}, p_k^{(n)} \rangle \ge \langle \theta_i^{(n)} - \theta_k^{(n)}, p_j^{(n)} \rangle.
\]
Summing it up with~\eqref{eq:irreducible_2}, we conclude that
\(
u_i^{(n)} - u_j^{(n)} \ge \langle \theta_i^{(n)} - \theta_j^{(n)}, p_j^{(n)} \rangle
\). This contradiction proves the statement.
\end{proof}

\begin{remark} We state without a proof that the total number of irreducible couples of points is asymptotically equal to $\zeta({I})^{-1} \cdot n^{2{I}}$, where $\zeta$ is the Riemann zeta function. Informally, this can be show by the following argument. Rescale the points $\theta_j^{(n)}$ in such a way that the lattice $\{\theta_j^{(n)}\}_j$ coincides with the uniform integer lattice $\{1, 2, \dots, n\}^{\mathcal{I}}$. Then the couple of points $\big(\theta_i^{(n)}, \theta_j^{(n)}\big)$ is irreducible if and only if
\[
\mathrm{gcd}\big(\theta_{i, 1}^{(n)} - \theta_{j, 1}^{(n)}, \theta_{i, 2}^{(n)} - \theta_{j, 2}^{(n)}, \dots, \theta_{i, {I}}^{(n)} - \theta_{j, {I}}^{(n)}\big) = 1,
\]
where $\mathrm{gcd}$ denotes the greatest common divider.

For each prime $p$, consider the event
\begin{equation}\label{eq:p_not_divide_gcd_event}
 \mathrm{gcd}\big(\theta_{i, 1}^{(n)} - \theta_{j, 1}^{(n)}, \theta_{i, 2}^{(n)} - \theta_{j, 2}^{(n)}, \dots, \theta_{i, {I}}^{(n)} - \theta_{j, {I}}^{(n)}\big)  \quad\mbox{is not divisible by}\quad p.
\end{equation}
For each $k$, the number $p$ divides $\theta_{i, k}^{(n)} - \theta_{j, k}^{(n)}$ with the probability approximately equal to $\frac{1}{p}$. The events corresponding to different coordinate numbers $k$ are mutually independent; therefore, the probability of the event~\eqref{eq:p_not_divide_gcd_event} is approximately equal to $1 - p^{-{I}}$.

For a finite set of distinct prime numbers $p$, the events~\eqref{eq:p_not_divide_gcd_event} can be considered as approximately mutually independent; therefore, the probability of the irreducibility of the couple $\big(\theta_i^{(n)}, \theta_j^{(n)}\big)$ is approximately equal to
\[
\prod_{p \text{ is prime}} (1 - p^{-{I}}) = \zeta({I})^{-1}.
\]

In the case two items (${I} = 2$), the total number of irreducible couples is approximately equal to $\frac{6}{\pi^2} n^{4} \approx 0.61 n^4$; so, the heuristic described in Proposition~\ref{prop:irreducibility_heuristic} removes approximately 39\% of all the \textbf{(ic)}-constraints.
\end{remark}

Another heuristics that can be useful in practice is the following: we replace the global \textbf{(ic)}-constraints with the following local ones:
\begin{multline*}
\textbf{(ic-local)}\quad u_i^{(n)} - u_j^{(n)} \ge \langle \theta_i^{(n)} - \theta_j^{(n)}, p_j^{(n)} \rangle\\
\text{for all irreducible } (\theta_i^{(n)}, \theta_j^{(n)}) \colon ||\theta_i^{(n)} - \theta_j^{(n)}|| \le \frac{c}{n},
\end{multline*}
where $c$ is a small constant.

This definition is motivated by the notion of the directional convexity considered in \cite{oberman2013numerical}. Let $V$ be a set of vectors $v$ with integer coordinates such that $||v|| \le c$. Then $||\theta_i^{(n)} - \theta_j^{(n)}|| \le \frac{c}{n}$ if and only if $\theta_i^{(n)} - \theta_j^{(n)}$ is proportional to $v$ for some $v \in V$.

Assume for simplicity that the function $u$ defined on ${X}$ is twice-differentiable. We say that $u$ is \textit{directionally convex} with respect to the set of direction vectors $V$ if
\[
\frac{\partial^2u}{\partial v^2} \ge 0 \quad \text{for all } v \in V.
\]
In other words, the function $u$ is directionally convex if and only if the function $t \to u(x + tv)$ is convex for all $x \in {X}$ and all $v \in V$.

The constraint \textbf{(ic-local)} can be considered as a discrete version of the directional convexity; we state without the proof that if the sequence $\left\{\big(u^{(n)}, p^{(n)}\big)\right\}_n$ satisfy the \textbf{(ic-local)} constraint, then the sequence $\{u^{(n)}\}_n$ contains a subsequence that converges weakly to the directionally convex function $u$.

\cite{oberman2013numerical} proved that a directionally convex function is ``nearly'' convex.
\begin{proposition*}[{{\cite{oberman2013numerical}, Proposition 3.1}}]
Denote by $d\theta$ the directional resolution of the set $V$:
\[
d\theta = \max_{||w|| = 1}\min_{v \in V} \arccos \left(\frac{w \cdot v}{||v||}\right).
\]
Assume that $d\theta \le \frac{\pi}{4}$. Then every directionally convex function $u$ is nearly convex, in the sense that
\[
\frac{\lambda_1}{\lambda_{I}} \ge -\tan^2(d\theta),
\]
where $\lambda_1 \le \lambda_2 \le \dots \le \lambda_{I}$ are the eigenvalues of the Hessian matrix at the point~$x$.
\end{proposition*}
If $c \to \infty$, then $d\theta \to 0$, and the nearly convex function $u$ becomes a convex one. This justifies the convenience of the suggested approach. In practice, an iterative Algorithm~\ref{alg:iterative} can be used: we start with a small number of initial \textbf{(ic)}-constraints, and, at each step, we add all the violated constraints to the linear program.

\begin{algorithm}
\caption{Iterative scheme with the local IC-constraints}\label{alg:iterative}
Define the problem $\mathcal{D}^{(0)}_{n, M}$, where \textbf{(ic)}-constraints are replaced with \textbf{(ic-local)}\;

\For{$\mathrm{s} = 0, 1, \dots$}{
find an optimal solution $\big(\overline{u}^{(n)}, \overline{p}^{(n)}\big)$ of the problem $\mathcal{D}^{(s)}_{n, M}$\;
define the set $V^{(s)}$ of all the violated constraints:
\[
V^{(s)} = \left\{(i, j) \colon \overline{u}_i^{(n)} - \overline{u}_j^{(n)} < \langle \overline{p}_j, \theta_i^{(n)} - \theta_j^{(n)}\rangle\right\}\;
\]
\eIf{$V^{(s)}$ is not empty}{
define $\mathcal{D}^{(s + 1)}_{n, M} = \mathcal{D}^{(s)}_{n, M}$\;
add all the violated constraints to the problem $\mathcal{D}^{(s + 1)}_{n, M}$:
\[
u_i^{(n)} - u_j^{(n)} \ge \langle p_j, \theta_i^{(n)} - \theta_j^{(n)}\rangle \quad\text{for all }(i, j) \in V^{(s)}\;
\]
}{
stop with the solution $\big(\overline{u}^{(n)}, \overline{p}^{(n)}\big)$ of the problem $\mathcal{D}^{(s)}_{n, M}$\;
}
}
\end{algorithm}

Finally, we improve the convergence rate of the majorization constraint approximation obtained in Theorem~\ref{thm:stochastic_dominance_convergence_rate} for the case of bounded optimal solution $\overline{c}_{j, k}^{(n)}$ of the dual problem $\mathcal{D}^*_{n, M}$.

\begin{proposition}\label{prop:nonunifrom_partition_estimation} Assume that the partition $0 = q_0 < q_1 < \dots < q_M = 1$ satisfies the following property: for each function $\varphi \colon [0, 1] \to \mathbb{R}$ which is constant on the intervals $[0, t_1]$, $[t_M, 1]$ and linear on each interval $[t_m, t_{m + 1}]$ we have
\[
\int_0^1 \varphi(t)\,{\dd}\eta(t) = \sum_{m = 1}^M \varphi(t_m) \cdot w_m.
\]
Let $\big(\overline{\varphi}^{(n)}_{m, k}, \overline{\psi}^{(n)}_{j, k}, \overline{c}^{(n)}_{j, k}\big)$ be an optimal solution to $\mathcal{D}^*_{n, M}$. Then
\[
\max \mathcal{D}_n \le \max \mathcal{D}_{n, M} + \int_{t_M}^1(t - t_M)\,{\dd}\eta(t) \cdot \sum_{k\in\mathcal{I}} \max_{j} \overline{c}_{j, k}^{(n)}.
\]
\end{proposition}
\begin{remark}
For the case of $B = 2$ bidders, the distribution $\eta$ is uniform on the interval $[0, 1]$, and the uniform partition $q_m = m / M$ satisfies the restriction of Proposition~\ref{prop:nonunifrom_partition_estimation}. In the case of $B > 2$, such a partition can be efficiently found numerically. We checked numerically that for such a partition we have $t_M = 1 - O(M^{-1})$; therefore,
\[
\int_{t_M}^1(t - t_M)\,{\dd}\eta(t) = O(M^{-2}).
\]
So, if $\overline{c}_{j, k}^{(n)}$ is uniformly bounded on $M$, the constructed partition provides a quadratic convergence rate (compared with the linear convergence rate obtained in Theorem~\ref{thm:stochastic_dominance_convergence_rate}).
\end{remark}
\begin{proof}[Proof of Proposition~\ref{prop:nonunifrom_partition_estimation}]
The proof is based on the same ideas as the proof of Theorem~\ref{thm:stochastic_dominance_convergence_rate}. Define
\[
\overline{\varphi}_k(t) = \max\left\{0,\, \max_j\big(t \cdot \overline{c}_{j, k}^{(n)} - \overline{\psi}^{(n)}_{j, k}\big)\right\},\qquad t \in [0, 1].
\]
The function $\overline{\varphi}_k(t)$ is convex, non-decreasing, and non-negative. In addition,
\[
\overline{\varphi}_t(t_m) = \overline{\varphi}_{m, k}^{(n)} \quad\text{for all } 1 \le m \le M.
\]
Besides, by the construction, this function satisfies the \textbf{(lt)}-constraint:
\[
\overline{\varphi}_k(t) + \overline{\psi}_{j, k}^{(n)} \ge t \cdot \overline{c}_{j, k}^{(n)}.
\]

For each $k$, let $\widehat{\varphi}_k(x)$ be a unique function that is constant on the intervals $[0, t_1]$, $[t_{M}, 1]$, linear on each interval $[t_m, t_{m + 1}]$, and equal to $\overline{\varphi}_{m, k}^{(n)}$ at the point $t_m$ for all $1 \le m \le M$. We claim that $\widehat{\varphi}_k(t) \ge \overline{\varphi}_k(t)$ for all $t \in [0, t_M]$. Indeed, for each $m$, the function $\widehat{\varphi}_k(t)$ is linear on the interval $[t_m, t_{m + 1}]$, the function $\widehat{\varphi}_k(t)$ is convex on the same interval, and 
\[
\overline{\varphi}_k(t_m) = \widehat{\varphi}_k(t_m), \quad \overline{\varphi}_k(t_{m + 1}) = \widehat{\varphi}_k(t_{m + 1}).
\]
Thus it follows from Jensen's inequality that $\widehat{\varphi}_k(t) \ge \overline{\varphi}_k(t)$ for all $t \in [t_m, t_{m + 1}]$. Finally, the function $\overline{\varphi}_k(t)$ is non-decreasing on the interval $[0, t_1]$, the function $\widehat{\varphi}_k(t)$ is constant on $[0, t_1]$, and $\widehat{\varphi}_k(t_1) = \overline{\varphi}_k(t_1)$; therefore, $\widehat{\varphi}_k(t) \ge \overline{\varphi}_k(t)$ for all $t \in [0, t_1]$.

By the construction, the derivative of $\overline{\varphi}_k(t)$ cannot exceed  the maximal slope in the family linear functions $t \cdot \overline{c}_{j, k}^{(n)} - \overline{\psi}^{(n)}_{j, k}$. Denoting by $C_k = \max_j \overline{c}^{(n)}_{j, k}$, we conclude that for all $t \in [t_M, 1]$:
\[
\overline{\varphi}_k'(t) \le C_k \quad \Rightarrow \quad \overline{\varphi}_k(t) \le \overline{\varphi}_k(t_M) + C_k \cdot (t - t_M) = \widehat{\varphi}_k(t) + C_k \cdot (t - t_M).
\]
Thus $\widehat{\varphi}_k(t) + C_k \cdot [t - t_M]_+ \ge \overline{\varphi}_k(t)$ for all $t \in [0, 1]$.

Finally, it follows from Proposition~\ref{prop:weak_discrete_duality} that
\begin{align*}
\max\mathcal{D}_n &\le \sum_{k\in\mathcal{I}}\left(\int_0^1 \left(\widehat{\varphi}_k(t) + C_k \cdot [t - t_M]_+\right)\,{\dd}\eta(t) + \sum_{j = 1}^{n^{I}} \overline{\psi}^{(n)}_{j, k} \cdot \mu_j^{(n)}\right)=\\
&=\sum_{k\in\mathcal{I}}\left(\sum_{m = 1}^M\overline{\varphi}^{(n)}_{m, k} \cdot w_m + C_k\int_{t_M}^1(t - t_M)\,{\dd}\eta(t) + \sum_{j = 1}^{n^{I}} \overline{\psi}^{(n)}_{j, k} \cdot \mu_j^{(n)}\right)=\\
&= \max \mathcal{D}_{n, M} + \sum_{k\in\mathcal{I}} C_k \cdot \int_{t_M}^1(t - t_M)\,{\dd}\eta(t). 
\end{align*}

\end{proof}

\section{Beckmann's problem, congested transport, and dynamic viewpoint} \label{app_Beckmann}

Beckmann's problem is equivalent to a Monge-Kantorovich-type problem called ``congested optimal transport'';
see \citep{santambrogio2015optimal} for the detailed presentation and references. Let us describe the equivalence informally for Beckmann's problem with the weight $\rho\equiv 1$. 
Given a domain $\Omega$ of a Euclidean space, an absolutely continuous supply-demand imbalance measure $\pi$ on $\Omega$ satisfying $\pi(\Omega)=0$, and a convex function $\Phi$, the following identity holds:
\begin{equation}\label{B-COT}
\inf_{c:\, {\div}[c]+\pi=0} \int \Phi(c) {\dd} x = \inf_{Q:\, Q_1-Q_0+\pi=0}  \int_{\Omega} \Phi(i_Q) {\dd} x.
\end{equation}
Here $Q$ ranges over to the set of all probability measures on ``curves'', i.e., continuous mappings 
$\gamma\colon  [0,1]\to\Omega$, and $Q_t$ denotes the probability measure on $\Omega$ obtained as the image of $Q$ under the
map $\gamma \to \gamma(t) \in \Omega$.
The object $i_Q$ is the so-called traffic intensity function which is defined so that the following identity holds for any test function $\varphi$: 
$$
\int_{\Omega} \varphi(x) i_Q(x) {\dd} x = 
\int\left( \int_0^1 \varphi(\gamma(t)) |\gamma'(t)| {\dd} t \right) {\dd} Q(\gamma)
$$

Formula (\ref{B-COT}) can be seen as a Lagrangian formulation of  Beckmann's problem taking a form of a ``problem for measures on curves.'' It is a quasi-dynamical formulation, a version of which is well known for the Monge-Kantorovich transportation problem; see \citep{villani2009optimal}.

The problem (\ref{B-COT}) is equivalent to a version of the Monge-Kantorovich problem with the cost function depending on optimal $Q$ (or $c$); see \cite[Theorem 4.33]{santambrogio2015optimal}. This form justifies the term ``congested optimal transport.''

Finally, let us mention that the construction of $Q$ relies on the so-called Dacorogna--Moser \citep{santambrogio2015optimal} interpolation of probability measures, which is a solution to the following transport equation:
$$
\frac{\partial}{\partial t}{\rho} + {\div} \Bigl[ \frac{c}{(1-t) f_+ + t f_{-}} \rho_t\Bigr]=0, \ \rho_0 = f_+,
$$
where $f_+$ and $f_-$ are the densities of the positive and the negative components of $\pi=\pi_c-\pi_p$ with respect to the Lebesgue measure.

On the other hand, to the best of our knowledge, there is no natural variational/dynamical interpretation of congested optimal transport in the spirit of the Benamou--Brenier formula (``problem for curves of measures''; see \cite{villani2009optimal,santambrogio2015optimal}); see also remarks in Section 4.5 of \citep{Carlier}.

\end{document}